\documentclass[letterpaper,11pt]{article}
\usepackage{amsmath}
\usepackage[font=small,labelfont=bf]{caption}
\usepackage[dvipsnames,x11names,table]{xcolor}
\usepackage{amsthm}
\usepackage{amssymb}
\usepackage{xspace}
\usepackage{graphicx}
\usepackage{fullpage}
\usepackage{hyperref}
\hypersetup{colorlinks=true, urlcolor=Blue, citecolor=DeepSkyBlue4, linkcolor=DeepPink4, breaklinks, unicode}
\usepackage{url}
\usepackage[capitalise]{cleveref}
\usepackage[shortlabels]{enumitem}
\usepackage{thm-restate}
\usepackage{mathtools}
\usepackage[T1]{fontenc}
\usepackage{babel}
\usepackage[utf8]{inputenc}
\usepackage[numbers]{natbib}
\usepackage{multirow}
\usepackage{array}
\usepackage{makecell}

\RequirePackage[charter]{mathdesign}
\RequirePackage{berasans,beramono}
\RequirePackage{microtype}
\RequirePackage[mathcal]{eucal}

\setitemize{itemsep=0pt}
\setenumerate{itemsep=0pt}

\RequirePackage{arydshln}

\newtheorem{theorem}{Theorem}
\newtheorem{lemma}[theorem]{Lemma}
\newtheorem{definition}[theorem]{Definition}
\newtheorem{question}[theorem]{Question}
\newtheorem{corollary}[theorem]{Corollary}
\newtheorem{claim}[theorem]{Claim}
\newtheorem{observation}[theorem]{Observation}

\newenvironment{innerproof}
 {\proof}
 {\endproof}

\newcommand{\eps}{\varepsilon}

\newcommand{\Oh}{\mathcal{O}}

\newcommand{\sub}{\subseteq}
\newcommand{\sm}{\setminus}
\newcommand{\Ot}{\widetilde{\Oh}}

\newcommand{\nn}{\ensuremath{\mathbb{N}}}
\newcommand{\fcal}{\ensuremath{\mathcal{F}}}
\newcommand{\qcal}{\ensuremath{\mathcal{Q}}}
\newcommand{\ccal}{\ensuremath{\mathcal{C}}}
\newcommand{\rcal}{\ensuremath{\mathcal{R}}}

\newcommand{\dcal}{\ensuremath{\mathcal{D}}}
\newcommand{\scal}{\ensuremath{\mathcal{S}}}

\newcommand{\mycol}{DeepPink4}
\newcommand{\coltext}[1]{{\textcolor{\mycol}{#1}}}
\newcommand{\coldef}[1]{{\em \color{\mycol}{#1}}}
\newcommand{\re}{{\color{black}{reachability emulator}}\xspace}
\newcommand{\fre}{{\color{black}{faithful reachability emulator}}\xspace}
\newcommand{\reach}{{{\mathrm{Reach}}}}
\newcommand{\reachsys}{{\mathrm{ReachSystem}}}

\newcommand{\mybox}[2]{
	\vspace{1mm}
	\begin{center}
		\noindent\colorbox{Honeydew2}{
            \begin{minipage}{#1\linewidth}{#2}
            \end{minipage}
        }
    \end{center}
	\vspace{1mm}
}

\title{Breaking Two Barriers for Dynamic Planar Reachability}
\author{
	Adam Karczmarz\footnote{Supported by Polish National Science Centre SONATA-18 grant number 2022/47/D/ST6/02184.} \\
	{\small University of Warsaw} \\
	{\small a.karczmarz@mimuw.edu.pl } \and
    Michał Włodarczyk\footnote{Supported by Polish National Science Centre SONATA-19 grant number 2023/51/D/ST6/00155.} \\
	{\small University of Warsaw} \\
	{\small m.wlodarczyk@mimuw.edu.pl }
}
\date{}

\begin{document}

\maketitle

\begin{abstract}
We consider the reachability problem in fully dynamic directed planar graphs: the graph undergoes edge updates (insertions and deletions) and we receive queries asking if one vertex is reachable from another.
The current-best online algorithm for dynamic planar reachability spends worst-case time $\Ot(n^{2/3})$ per query/update~\cite{subramanian1993fully}.
This can be improved to $\Ot(n^{1/2})$ in the offline setting or when the embedding is fixed~\cite{Diks07}.
Our contribution is twofold:

\begin{enumerate}
    \item We present the first algorithm that breaks the square-root barrier for dynamic planar reachability in the offline setting, obtaining time
      $\Oh(n^{0.49})$ per query/update.
    This shows that the problem is easier than
    answering distance queries in a dynamic undirected planar graph,
    for which a~conditional lower bound of $n^{1/2 - o(1)}$ is known~\cite{Abboud16}.
    Notably, this lower bound applies to both online and offline algorithms.

    \item Second, in the general online setting we design a randomized data structure with amortized
    update and query time $\Oh(n^{0.664})$.
    This is the first improvement upon the exponent $2/3$ in over three decades.
\end{enumerate}

The main engine behind our theorems is a new sparsification framework aimed at maintaining information about reachability.
Specifically, we study digraphs through the lens of VC dimensionality of certain set systems.
First, we design sparse reachability emulators for planar digraphs that preserve the VC dimension.
Second, we utilize orderings with a low ``stabbing number''~\cite{chazelle1989quasi} to compute
subquadratic reachability emulators in subquadratic time whenever the VC dimension is bounded.

\end{abstract}
\thispagestyle{empty}

\newpage
\setcounter{page}{1}

\section{Introduction}
The fully dynamic reachability problem (also called dynamic transitive closure) asks for a data structure maintaining a directed graph $G=(V,E)$ under edge insertions and deletions and supporting reachability queries: given $u,v\in V$, does $G$ currently contain an $u\to v$ path?
It is one of the most fundamental and well-studied problems in the area of dynamic graph algorithms.

The naive solution is either to recompute the transitive closure from scratch upon update or to run full graph search upon query.
Numerous data structures have been proposed to improve upon these.
Let $n, m$ denote the numbers of vertices and edges.
The full reachability matrix can be maintained explicitly in $\Oh(n^2)$ time per update~\cite{DBLP:journals/algorithmica/DemetrescuI08,DBLP:journals/talg/Roditty08,DBLP:conf/focs/Sankowski04}.
With update time near-linear in $m$, one can achieve $\Oh(n)$ query time~\cite{DBLP:journals/siamcomp/RodittyZ16}.
Algebraic methods~\cite{DBLP:conf/focs/BrandNS19, DBLP:conf/focs/Sankowski04} can be used to obtain data structures: (a)~with $\Oh(n^{1.41})$ update and query time, and (b) with $\Oh(n^{1.53})$ update time and $\Oh(n^{0.53})$ query time\footnote{The upper bounds of algebraic dynamic reachability data structures depend on the matrix multiplication exponent $\omega$.
Currently, it is known that $\omega\leq 2.3714$.~\cite{DBLP:conf/soda/AlmanDWXXZ25}}.
Via combinatorial means, one can also achieve $\Oh(m\sqrt{n})$ update time and $\Oh(\sqrt{n})$ query time~\cite{DBLP:journals/siamcomp/RodittyZ08}.

Notably, none of the aforementioned results improves upon the naive graph search approach for sparse graphs when queries occur as often as updates.
Specifically, no known data structure achieves $\Oh(n^{1-\epsilon})$ update and query time at the same time.
This is in fact impossible for data structures with $n^{2-\Omega(1)}$ time preprocessing~\cite{DBLP:conf/focs/AbboudW14} unless the so-called $k$-Cycle hypothesis fails~\cite{DBLP:conf/icalp/AnconaHRWW19,DBLP:conf/stoc/GutenbergWW20, DBLP:conf/soda/LincolnWW18}\footnote{Higher lower bounds on the sum of update and query times are known in the dense case~\cite{DBLP:conf/focs/BrandNS19,DBLP:conf/focs/AbboudW14,DBLP:conf/stoc/HenzingerKNS15,DBLP:conf/stoc/JinX22}.}.
Interestingly, this conditional lower bound holds also in the \emph{offline} setting, where the sequence of issued updates and queries is known upfront.
Weaker $n^{1/2-o(1)}$ conditional lower bounds have been established assuming any polynomial preprocessing in either online setting~\cite{DBLP:conf/stoc/HenzingerKNS15} or offline setting when limited to combinatorial means (no fast matrix multiplication allowed)~\cite{DBLP:conf/stoc/JinX22}.

The known lower bounds suggest that, in the search for more efficient
data structures in the sparse case, one should probably focus on special graph classes.
Therefore, in this paper we study fully dynamic reachability in \emph{planar} directed graphs.
Planar graphs constitute, arguably, the most popular and well-studied class of sparse graphs and
a large body of work has been devoted to directed planar reachability, both in the static and dynamic settings.
A~planar digraph can be preprocessed in linear-time so that
a single reachability query can be addressed in constant time~\cite{DBLP:conf/focs/HolmRT15}.
For the dynamic setting, this implies that the naive algorithm already achieves $\Oh(n)$ update time and $\Oh(1)$ query time.
Even before a near-optimal planar reachability oracle was first shown~\cite{thorup2004compact}, Subramanian~\cite{subramanian1993fully} proposed a dynamic
data structure for planar digraphs with $\Ot(n^{2/3})$ update and query time.
It only requires the graph to remain planar after every update, and does not assume any fixed embedding of the maintained~graph.

More efficient
data structures for planar graphs have been developed under stronger assumptions.
Diks and Sankowski \cite{Diks07} improved upon~\cite{subramanian1993fully} in the fixed-embedding setting (i.e., when an edge $(u,v)$ can be inserted only inside a specified face containing $u$ and $v$) and achieved $\Ot(\sqrt{n})$ update and query time.
This setting captures, for instance, the case when the graph $G$ is fixed but the edges can be switched on or off.
The same bound can be achieved in the offline setting without fixed embedding (see \Cref{subsec:offline}).
For acyclic plane-embedded graphs with exactly one source and one sink, a near-optimal data structure with $\Oh(\log{n})$ update and query time has been proposed~\cite{Tamassia93,DBLP:conf/esa/Husfeldt95}.
Bringmann, Künnemann, and Nusser~\cite{Bringmann19} considered directed grids, oriented towards one corner, in a model where each vertex can be switched on/off and one asks whether one corner is reachable from another.
For a grid of size $\sqrt{n}\times \sqrt{n}$, they presented  an offline algorithm processing $t$ updates in total time $\Ot(n+t\cdot n^{1/3})$.
{An online data structure supporting $\Ot(1)$ time reachability queries from a~single-source and $\Ot(n^{4/5})$ time updates and source changes has been described in~\cite{DBLP:journals/jcss/Charalampopoulos22}.}

To the best of our knowledge, fully dynamic reachability has not been previously studied in more expressive classes of sparse graphs, such as minor-free graphs.
{At the same time, partially dynamic reachability, i.e., allowing either only edge insertions or only deletions, has been studied in both planar and minor-free graphs~\cite{DBLP:conf/stoc/ItalianoKLS17,DBLP:conf/soda/Karczmarz18,KarczmarzZ25}.}

\subsection{The square-root barrier}

It is common for square-root terms to appear in running times of various algorithms for planar graphs, see, e.g.,~\cite{Marx13}.
They can be usually traced back to the classic Lipton-Tarjan theorem, ensuring that an $n$-vertex planar graph admits a balanced separator of size $\Oh(\sqrt{n})$~\cite{Lipton79,DBLP:journals/jcss/Miller86}.
It is then no surprise that improving below $n^{1 / 2 - o(1)}$ time forms a natural barrier for several fully dynamic problems in planar graphs.
Interestingly, this barrier is not merely an artifact of the known techniques.
Abboud and Dahlgaard~\cite{Abboud16} proved that no dynamic algorithm can support both edge weight updates and exact distance queries in an {\em undirected} planar graph within amortized time $n^{1 / 2 - \Omega(1)}$ unless
the {All Pairs Shortest Paths} problem can be solved in truly subcubic time.
The conjecture that no such truly subcubic algorithm exists, known as the APSP conjecture\footnote{Very recently, it has been announced that the APSP conjecture is false~\cite{alman2026trulysubquadratic3sumtruly}. When this result is fully confirmed, it will obviously remove the existing square-root barrier for dynamic planar distances.}, is one of the main hardness assumptions in fine-grained complexity~\cite{Williams18, abboud2023subcubic}.
When it comes to maintaining $(1+\eps)$-approximate distances in the \emph{undirected} case, square-root time remained a state-of-the-art for a decade~\cite{abraham2012fully} and breaking this barrier led directly to $n^{o(1)}$ time~\cite{chang2022arxiv, Filtser24} -- see \Cref{tab:results}.
The same phenomenon was observed in dynamic planarity testing~\cite{Holm20}.
This suggests that $\sqrt{n}$ is not only a quantitative barrier and leads to the following question.

\mybox{0.7}{For which dynamic planar problems $\sqrt{n}$ constitutes an impassable barrier?}

Specifically, Dahlgaard~\cite[Open Problem~8]{dagstuhl} asked in 2016
whether this barrier holds for dynamic planar reachability:
can one achieve
$n^{1/2 - \Omega(1)}$ query/update time or is such an algorithm unlikely due to some hardness assumption?

\renewcommand{\arraystretch}{1.4}
\newcolumntype{C}[1]{>{\centering\arraybackslash}p{#1}}

\begin{table}[t]
\centering
\footnotesize
\begin{tabular}{|C{3.5cm}|C{2.5cm}|C{2cm}|C{5cm}|}
\hline
\rowcolor{Honeydew2}
problem & model & time & reference \\ \hline
 \multirow{3}{*}{\makecell{directed\\exact distance}}  & online & $n^{2/3}$ & \cite{FR06}, cf. \cite{KleinMS13,gawrychowski16,DBLP:journals/talg/KaplanMNS17}\\
\cline{2-4}
  & {\makecell{offline\\weight updates}}  & {$n^{1/2}$} & {\cite{Das22}} \\
  \cline{2-4}
 & lower bound & $n^{1/2 - o(1)}$ & \cite{Abboud16} \\
\hline
 \multirow{4}{*}{\makecell{undirected\\$(1+\eps)$-apx. distance}}   & online & $n^{2/3}$ &  \cite{DBLP:journals/algorithmica/KleinS98} \\
 \cline{2-4}
  & online & $n^{1/2}$ & \cite{abraham2012fully} \\
\cline{2-4}
 & online & $n^{o(1)}$ &  \cite{Filtser24} \\
\cline{2-4}
  & offline & $\mathrm{polylog}(n)$ & \cite{chang2022arxiv} \\
  \hline
  & online & $n^{2/3}$ & \cite{subramanian1993fully}  \\
 \cline{2-4}
  & online & ${\textcolor{\mycol}{n^{0.664}}}$ & \textcolor{\mycol}{this work} \\
 \cline{2-4}
directed reachability & \makecell{online\\fixed embedding}  & $n^{1/2}$ & \cite{Diks07} \\
 \cline{2-4}
 & offline  & $n^{1/2}$ & corollary from \cite{thorup2004compact}, see \S\ref{subsec:offline} \\
\cline{2-4}
 & offline & ${\textcolor{\mycol}{n^{0.49}}}$ & \textcolor{\mycol}{this work} \\
 \hline
\end{tabular}
\caption{Previous and new results on dynamic reachability, distances, and approximate distances in planar graphs. The maximum of update and query time bounds is presented. We omit polylogarithmic factors in polynomial running times.
We also do not delve into details regarding amortization or randomization.
The upper bounds for reporting exact distances work for digraphs as well and the lower bound works already for undirected graphs.}\label{tab:results}
\end{table}

\subsection{Our main theorems}

We provide the first evidence that the square-root barrier does not apply to dynamic planar reachability.
Precisely, we prove that the exponent can be improved to $0.49$ in the offline model.

\begin{theorem}\label{thm:main:offine}
Let $G$ be a planar digraph.
Given a sequence of (a) edge updates (insertions or deletions) that preserve planarity of $G$ and (b) $(u \to v)$-reachability queries,
one can compute the answers to all the queries in deterministic time $\Ot(n+t\cdot n^{0.49})$, where $t$ is the length of the sequence.
\end{theorem}

The relevance of the offline regime stems from the fact the conditional lower bounds for dynamic data structures often work against offline algorithms as well~\cite{DBLP:conf/focs/AbboudW14, DBLP:conf/stoc/JinX22, Patrascu2010, roditty2004dynamic, Kopelowitz16, AbboundVH15, DBLP:conf/icalp/AnconaHRWW19, DBLP:conf/stoc/GutenbergWW20}.
This applies to the lower bound for dynamic planar shortest paths~\cite{Abboud16}.
Notably, \Cref{thm:main:offine} rules out the existence of an $n^{ 1 / 2 - o(1)}$ lower bound for offline dynamic planar reachability under a static hardness assumption such as the APSP conjecture.
This provides a sharp separation between maintaining directed reachability and undirected distances in planar graphs.
Another motivation comes from fine-grained complexity.
Improving the running time in \Cref{thm:main:offine} to $n^{o(1)}$ would imply an $n^{4+o(1)}$-time algorithm
 for computing discrete Fr\'echet distance under translation, essentially matching a conditional lower bound~\cite{Bringmann19}.
We believe that the obtained breach constitutes the first step towards achieving subpolynomial time for dynamic planar reachability.
We elaborate on this perspective in \Cref{sec:conclusion}.

Next, we move on to the online setting.
In both fully dynamic problems, directed reachability and exact distances,
$\widetilde\Oh(\sqrt{n})$ time has been so far only achieved in restricted models~\cite{Das22,Diks07}\footnote{The offline fully dynamic data structure of~\cite{Das22} is in fact achieved by showing an \emph{online} $\Ot(\sqrt{n})$ worst-case bound in the incremental setting, i.e., with only edge weight decreases.}
while in general the best exponent remains $2/3$ \cite{subramanian1993fully, FR06,DBLP:journals/talg/KaplanMNS17,DBLP:conf/soda/Klein05,gawrychowski16}\footnote{The literature on dynamic exact distance queries covers edge weight updates. However, the $\Ot(n^{2/3})$ update/query bound carries to arbitrary planarity preserving edge updates by maintaining an $r$-division with few holes as~\cite{subramanian1993fully,DBLP:journals/algorithmica/KleinS98,DBLP:journals/jcss/Charalampopoulos22} did.}.
This corresponds to the trade-off between the boundary size $\Ot(n/\sqrt{r})$ and update time $\Ot(r)$ of a dynamically maintained $r$-division with few holes~\cite{KleinMS13} (see also~\cite{subramanian1993fully, DBLP:journals/jcss/Charalampopoulos22}).
Our second main theorem constitutes the first improvement in the general online setting since the work of Subramanian~\cite{subramanian1993fully} from 1993.

\begin{theorem}\label{thm:planar-dyn-reach}
    There exists a dynamic reachability data structure maintaining a planar digraph $G$ subject to planarity-preserving edge insertions and deletions with $\Oh(n^{0.664})$ amortized update time and $\Oh(n^{0.664})$ query time.
    The data structure is Monte Carlo randomized and the answers are correct with high probability.
\end{theorem}

Our techniques work also beyond planarity: in \Cref{thm:minor-free-dyn-reach}
we present truly
sublinear online algorithms for reachability in minor-free digraphs.

The main engine behind our theorems is a new sparsification framework aimed at maintaining information about reachability.
Specifically, we design novel \coldef{reachability emulators} that exploit  VC dimensionality of certain set systems and utilize the \coldef{Stabbing Number Theorem} (\Cref{thm:prelim:stab})~\cite{chazelle1989quasi, welzl1988partition}.
In the recent years, analyzing VC dimension of various graph-related set systems has led to exciting developments in diameter computation~\cite{li2019planar, Ducoffe22, LeW24, kluk2025faster, wlodarczyk25, chan2025truly}, clustering~\cite{Cohen-AddadD0SS25}, model checking~\cite{Bonnet22, Dreier24}, as well as distance oracles~\cite{LeW24, KarczmarzZ25}.
Among the above, \cite{Ducoffe22,Dreier24,chan2025truly} also utilize the {Stabbing Number Theorem}.
To put our developments in context, we now discuss the relation between sparsification and dynamic algorithms.

\subsection{Sparsification in dynamic algorithms}
\label{sec:sparsification}

Understanding why some dynamic planar problems are easier than others boils down to  understanding what kind of graph information can be effectively compressed.
The APSP-based lower bound for dynamic distance queries in undirected planar graphs~\cite{Abboud16} exploits the fact that a Boolean matrix of size $k\times k$ can be encoded using pairwise distances within a set of $\Oh(k)$ vertices in a planar graph of size $\Oh(k^2)$ with edge weights in $\Oh(k)$.
In particular, this implies that encoding all the pairwise distances between $\Oh(k)$ vertices requires $\Omega(k^2)$ bits (cf.~\cite{gavoille2004distance}).
However, a non-trivial compression is possible
when we only want to preserve undirected distances approximately.
Chang, Krauthgamer, and Tan~\cite{ChangKT22, chang2022arxiv} showed that for any set of $k$ terminal vertices in a planar graph $G$ and $\eps > 0$, there exists a planar graph $H$ of size $\Ot(k / \eps^{\Oh(1)})$ so that the distances among terminals differ by a factor at most  $1+\varepsilon$ in $G$ and $H$ (cf.~\cite{chang2025distance}).
It was already known~\cite{chen2020fast} that such a construction allows for a repetitive compression of a graph\footnote{In fact, \cite{chen2020fast} required the graph $H$ be a minor of $G$ which does not hold in \cite{chang2022arxiv} but this is not an issue therein.} and gives
an offline algorithm for dynamic $(1+\varepsilon)$-approximate distances with query/update time $\mathrm{polylog}(n)$.
This construction was later fully dynamized by
Filtser, Goranci, Patel,  and Gutenberg~\cite{Filtser24}, resulting in an online algorithm with subpolynomial query/update time.

A peculiar issue with directed reachability in planar graphs is that its compressibility status is somewhat in between the two cases above.
For a digraph $G$ and $X \sub V(G)$ we say that a digraph $H$ is a \coldef{reachability emulator}\footnote{A different name appearing in the literature is ``reachability preserver'' but sometimes it
 refers to a minor or a subgraph of the original graph~\cite{AbboudB24, BodwinL25, Goranci20}.
 Another name is ``reachability certificate''~\cite{subramanian1993fully, Diks07}.
 }
for $(G,X)$ if $X \sub V(H)$ and the reachability relation on $X \times X$ is the same in $G$ and~$H$.
On one hand, when $G$ is planar, the reachability oracle by Thorup~\cite{thorup2004compact} can be utilized to
construct a \re $H$ for $(G,X)$ with $\Oh(|X|\log |X|)$ vertices and edges (see \Cref{sec:faithful}).
This means that one cannot mimic the hardness reduction from~\cite{Abboud16} by encoding matrices of size $k \times k$ using reachability relation on $\Oh(k)$ vertices.
On the other hand, the graph $H$ is not planar anymore and in fact no alternative procedure could ensure planarity.
Goranci, Henzinger, and Peng~\cite{Goranci20} showed
that a \re for $G$ that is required to be a minor of $G$ may need $\Omega(k^2)$ vertices.
We strengthen their argument and prove that the same lower bound applies even when we do not require~$H$ to be a minor of~$G$, but only request $H$ to be planar.

\begin{restatable}{theorem}{thmLB}
\label{thm:lb}
    For every $k \ge 3$ there exists a plane digraph $G$ with a vertex subset $T \sub V(G)$ of size $4k + 5$ such that the vertices from $T$ lie on four faces of $G$ and
    every planar digraph that is a \re for $(G,T)$  has at least $(k-2)^2$ vertices.
\end{restatable}

This means that we cannot
follow the strategy of repetitive compression from~\cite{chen2020fast} as this requires
inserting additional edges among $X$ and compressing the digraph again.
Insertion of new edges invalidates the balanced separators used in the construction from~\cite{thorup2004compact} while non-planarity of $H$ precludes the ability to construct this kind of reachability emulators from scratch.
However, this does not rule out other kinds of compression that exploit different properties than planarity.

\mybox{0.8}
{The main idea of this work is that one can compute reachability emulators for planar digraphs that (1) have near-linear size in the number of terminals, and (2) preserve some property that is weaker than planarity but still sufficient to efficiently compress the reachability information again.}

\paragraph{Organization of the paper.}
We begin with an informal overview of our techniques in \Cref{sec:overview}.
In \Cref{sec:prelim} we collect the basic facts about VC dimension and \re{s}.
We develop our sparsification framework in Sections~\ref{sec:faithful} and~\ref{sec:emul-vc}.
\Cref{sec:offline} concerns the offline algorithm.
In \Cref{sec:online} we present our online algorithm, starting from a simplified construction with sublinear query/update time, which works for minor-free digraphs as well.
\Cref{sec:improved-vc} provides an improved bound on the VC dimension of a reachability set system in a planar digraph
and \Cref{sec:lower-bound} is devoted to \Cref{thm:lb}.
We conclude in \Cref{sec:conclusion}.
We also provide an appendix with proofs included for completeness.

\section{Overview}
\label{sec:overview}

\subsection{A simple offline algorithm}
\label{subsec:offline}

As a warm-up before discussing \Cref{thm:main:offine}, we first sketch a simple offline algorithm with query/update time $\Oh(\sqrt{n}\log n)$.
This follows from the work of Thorup~\cite{thorup2004compact}, but we are not aware of
any work stating it explicitly.
The advantage of the offline model is that one can foresee which parts of a graph are going to be modified in the nearest future and which remain unaffected, hence it pays off to compress the second ones.
Similarly as in the offline algorithm for grids~\cite{Bringmann19},
we divide the sequence of updates and queries into \coldef{chunks}.
We choose the length of a chunk as $\sqrt{n}$ and we want to process each chunk in time $\Oh(n\log n)$, which will give $\Oh(\sqrt{n}\log n)$ operations per event on average.

Consider a single chunk and denote by $G_i$ the graph after processing the $i$-th event; also let $G = G_0$.
Let $X \sub V(G)$ be the set of vertices that are \coldef{affected} by this chunk, i.e., they appear in at least one update or query.
Clearly, $|X| \le 2\sqrt{n}$.
We refer to edges with both endpoints in $X$ as the \coldef{$X$-edges}.
Let \coldef{$G \ominus X$} be the graph obtained from $G$ by removing all the $X$-edges.
Observe that  $G_i \ominus X = G \ominus X$ for all $i$, because this graph omits all the edges that may appear or disappear in the current chunk.
We prefer to first remove all the $X$-edges and later insert an appropriate subset of them because maintaining a \re under edge insertion is way easier compared to edge deletion.

We start processing the chunk by computing a \re $H$ for $(G \ominus X,X)$ of size $\Oh(\sqrt{n}\log n)$ in time $\Oh(n\log n)$~\cite{thorup2004compact} (see \Cref{sec:faithful}).
Note that $H \cup G_i[X]$ forms a \re for $(G_i,X)$.
Consider a query $(u\to v)$ that arrives as the $i$-th event.
We have $u,v \in X$ and we want to determine if $v$ is reachable from $u$ in $G_i$.
This can be checked by running depth-first search from $u$ in $H \cup G_i[X]$.
Since $G_i[X]$ is planar, it has $\Oh(|X|)$ edges and so the size of $H \cup G_i[X]$ is $\Oh(\sqrt{n}\log n)$.
Hence, the algorithm processes a chunk of length $\sqrt{n}$ in total time  $\Oh({n}\log n)$.

\subsection{Offline reachability with two layers}
\label{subsec:twolayers}

A natural idea to improve the above algorithm is to divide the chunks into smaller chunks
and proceed recursively, similarly as in the offline algorithm for approximate distances~\cite{chen2020fast, chang2022arxiv}.
However, unlike that approach, we operate in just  two layers and our treatment of the second layer differs significantly from the first one.
Let us check what kind of compression is needed to make this idea~work.

Let $n^\alpha$ and $n^\beta$ be the sizes of a \coldef{large} and \coldef{small chunk}, respectively, for constants $\alpha > \beta > 0$ to be chosen later.
We spend $\Oh(n\log n)$ time at the beginning of each large chunk to compute an~emulator, which translates to $\Ot(n^{1-\alpha})$ time per event.
Hence to improve upon the square-root, we  must set $\alpha > \frac 1 2$.

\begin{figure}[h]
    \centering
    \includegraphics[width=\textwidth]{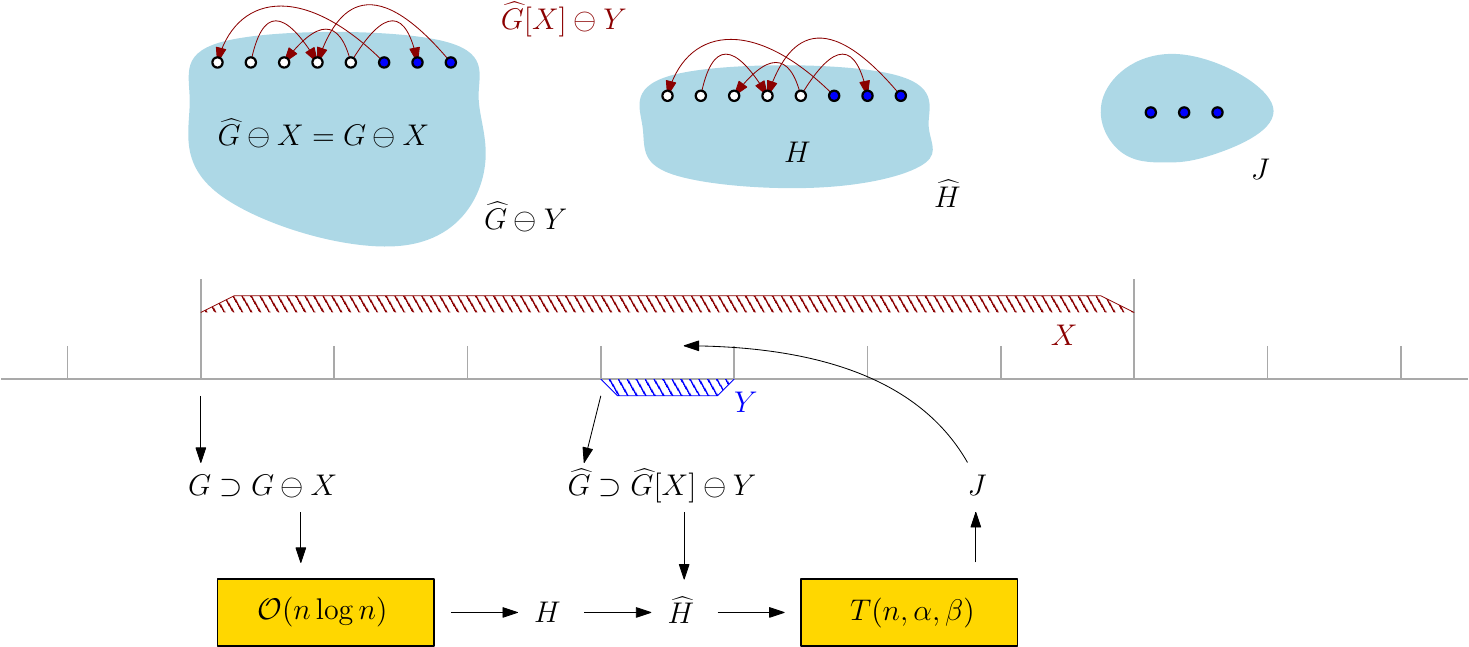}
    \caption{The two-layer scheme of reachability compression.
    The sets $X \supseteq Y$ are the vertices affected by  the large chunk $C_L$ of size $n^\alpha$ and the small chunk $C_S$ of size $n^\beta$, respectively.
    At the start of $C_L$ we observe digraph $G$ and at the start of $C_S$ we observe $\widehat G$.
    The blue disks represent the vertices from $Y$ and the white ones represent $X \sm Y$.
    At the top left we see the digraph $\widehat G \ominus Y$ which is a union of $G \ominus X$ (blue) and $\widehat G[X] \ominus Y$ (dark red edges).
    The first one is replaced by the \re $H$, computed at the start of $C_L$.
    Then at the start of $C_S$ the vertices from $X\sm Y$ become non-terminals
    and digraph $\widehat H = H \cup (\widehat G[X] \ominus Y)$ is compressed to $J$.
    This digraph contains all the information about pairwise connections among $Y$ in $\widehat G$ that does not use the $Y$-edges.
    To process the small chunk $C_S$ we keep track of the subgraph of size $\Oh(n^\beta)$ induced by $Y$ and we address each query by merging it with~$J$.
    }
    \label{fig:twolayers}
\end{figure}

For a fixed large chunk $C_L$, let $G$ denote the graph observed at its start and let $X$ and $H$ be defined as in \Cref{subsec:offline}.
Consider a small chunk $C_S$ inside $C_L$ and let $Y$ be the set of vertices affected by~$C_S$.
Note that $|X| \le 2n^\alpha$ and  $|Y| \le 2n^\beta$.
To process $C_S$ we need a \re for $(\widehat G \ominus Y,Y)$, where $\widehat G$ denotes the graph observed at the start of $C_S$.
It holds that $\widehat G \ominus Y = (\widehat G \ominus X) \cup (\widehat G[X] \ominus Y)$
and $\widehat G \ominus X =  G \ominus X$.
Recall that $H$ is a \re for $(G \ominus X,X)$.
Observe that $\widehat H := H \cup (\widehat G[X] \ominus Y)$ is a \re for $(\widehat G \ominus Y,X)$ and
$Y \sub X$ so $\widehat H$ also forms a \re $(\widehat G \ominus Y,Y)$.
These relations are illustrated in \Cref{fig:twolayers}.

The problem is that $\widehat H$ is too large: it may have  $\Omega(n^\alpha)$ vertices so inspecting it at every event in $C_S$ would result in prohibitive running time  $\Omega(n^\alpha)$ {per event}.
We would like to replace this representation with one of size  depending only on $|Y|$, e.g., by storing the entire reachability matrix for $Y \times Y$.
To this end, we will compute a \re $J$ for $(\widehat H,Y)$ of size $\Oh(n^{2\beta})$ at the start of $C_S$.
Then $J$ will constitute a \re for $(\widehat G \ominus Y,Y)$,  containing all the information about connections on  $Y \times Y$ in $\widehat G$ that do not use any $Y$-edges.

Let $T(n,\alpha,\beta)$ denote the running time required to compute the reachability matrix for $Y \times Y$
when $\widehat H$ of size $\Oh(n^\alpha\log n)$ is given.
Note that we cannot exploit the balanced separators in $G \ominus X$ used for computing $H$ because they become obsolete after adding new edges to $G \ominus X$ and $H$ (see \Cref{subsec:faithful}).
We compute $J$ once for each small chunk, which adds term $T(n,\alpha,\beta)\cdot n^{-\beta}$ to the {average time needed to process an event.}
Then to process the $i$-th query in $C_S$ we run depth-first search in $J \cup G_i[Y]$ where $G_i$ denotes the graph observed when the $i$-th query arrives.
The size of $J$ is $\Oh(n^{2\beta})$ and the number of edges in $G_i[Y]$ is $\Oh(n^{\beta})$ so this term becomes irrelevant whenever we set $\beta < \frac 1 4$.

The crux of the analysis is to ensure that $T(n,\alpha,\beta)\cdot n^{-\beta} = n^{1/2 - \Omega(1)}$ for some $\alpha > \frac 1 2$, $\beta < \frac 1 4$.
To this end, we aim at a bound
\begin{equation}\label{eq:alpha-beta}
T(n,\alpha,\beta) = n^{\alpha + (1-\delta)\cdot \beta} + n^{C \cdot \beta}
\end{equation}

\noindent for some universal constants $\delta, C > 0$.
Indeed, observe that the above formula implies
\[T(n,\alpha,\beta)\cdot n^{-\beta} = n^{\alpha - \delta\cdot \beta} + n^{(C-1) \cdot \beta}.\]

\noindent  By adjusting $\alpha,\beta$ to $\delta, C$ we can always guarantee that
\[M := \max\left(1-\alpha,\, \alpha - \delta\cdot \beta,\, (C-1) \cdot \beta\right) < \frac 1 2.\]

\noindent Concretely, we will get $\delta = \frac 1 3$ and $C = 9$ for which $\alpha = \frac{25}{49}$ and $\beta = \frac 3 {49}$ give $M = \frac{24}{49} < \frac 1 2$.

Recall that $T(n,\alpha,\beta)$ is the time required to compute the reachability matrix for $Y \times Y$ in $\widehat H$ where the size of the graph $\widehat H$ is roughly $n^\alpha$ and $|Y| \approx n^\beta$.
The condition \eqref{eq:alpha-beta} asks for an algorithm that performs such computation in time $\Oh(|\widehat H|\cdot|Y|^{1-\delta} + |Y|^C)$, that is, it asks for any polynomial improvement upon the naive complexity $\Oh(|\widehat H|\cdot|Y|)$ corresponding to running depth-first search from every vertex in $Y$.
While we do not expect such an algorithm to exist for arbitrary $\widehat H$, we will ensure that $\widehat H$ belongs to a certain class of digraphs, in which such an improvement is possible.

\subsection{VC-dimension-based emulators}
\label{subsec:vc}

For a digraph $G$, a set $X \sub V(G)$, and vertex $v \in V(G)$, let $\textcolor{\mycol}{\reach}_G(v,X)$ denote the set of vertices in $X$ reachable from $v$.
We define $\textcolor{\mycol}{\reachsys}(G,X)$ as the set system $\{{\reach}_G(v,X) \colon v \in V(G)\}$ on the ground set $X$.
A set $Z \sub X$ is \coldef{shattered} by a set system $(X,\rcal)$ if for every $Z' \sub Z$ there is $R \in \rcal$ such that $Z' = Z \cap R$.
The \coldef{VC dimension} of a set system $(X,\rcal)$ is the maximum size of a set shattered by $(X,\rcal)$.
It is known~\cite{LeW24, KarczmarzZ25} that when a digraph $G$ is planar then $\reachsys(G,X)$ has VC dimension at most four for every $X \sub V(G)$.
In \Cref{thm:prelim:planar-vc} we improve this bound to three, which results in slightly better running times.

Our strategy is as follows.
First, we will compute a special \re for a large chunk to ensure that
for $\widehat H, Y$ defined in \Cref{subsec:twolayers} the set system
$\reachsys(\widehat H, Y)$ has bounded VC dimension.
Second, we will prove that under this assumption a small \re for $(\widehat H, Y)$
can be computed in time $\Oh(|\widehat H|\cdot|Y|^{1-\delta} + |Y|^C)$, for some $\delta,C > 0$.
Within this subsection we focus on the second step: we assume that the VC dimension bound holds and we explain how to take advantage of it.

The \coldef{Stabbing Number Theorem}~\cite{chazelle1989quasi, welzl1988partition} states that when a set system $\fcal \sub 2^Y$ has VC dimension~$d$, then there is $\delta > 0$ depending on $d$ and an ordering $\pi$ of $Y$ such that every set $F \in \fcal$ is a union of $\widetilde\Oh(|Y|^{1-\delta})$ intervals over $\pi$.
Furthermore, the ordering $\pi$ can be computed in time polynomial in $|Y|+|\fcal|$.
In fact, \coldef{Sauer–Shelah Lemma}~\cite{sauer1972density} guarantees that the size of $\fcal$ cannot exceed $\Oh(|Y|^d)$ so the running time is also polynomial in $|Y|$.

We apply this theorem to the set system $\reachsys(\widehat H, Y)$.
Assume for now that the corresponding ordering $\pi$ of $Y$ is given.
We compute a reachability emulator for $(\widehat H, Y)$ as follows.
First, we can contract each strongly connected component in $\widehat H$ because vertices in a single component have the same reachability sets $\reach_{\widehat H}(v,Y)$.
We can thus assume that $\widehat H$ is acyclic and we
 consider its vertices in reverse topological order, computing $\reach_{\widehat H}(v,Y)$ one by one.
When $v$ has outgoing edges to $u_1,u_2,\dots u_\ell$, then $\reach_{\widehat H}(v,Y)$ is the union of the sets $\reach_{\widehat H}(u_i,Y)$, together with $\{v\}$ if $v \in Y$.
Since every set $\reach_{\widehat H}(u_i,Y)$ is represented as $\widetilde\Oh(|Y|^{1-\delta})$ intervals over $\pi$, we can easily compute their union in time $\widetilde\Oh(\ell \cdot |Y|^{1-\delta})$.
See \Cref{fig:interval} for an illustration.
This yields total running time of $\widetilde\Oh(|\widehat H| \cdot |Y|^{1-\delta})$.
In the end, we have computed an interval representation of $\reach_{\widehat H}(v,Y)$ for every $v \in Y$, which immediately gives a \re for  $(\widehat H, Y)$ of size $|Y|^2$.

\begin{figure}[t]
    \centering
    \includegraphics[width=0.8\textwidth]{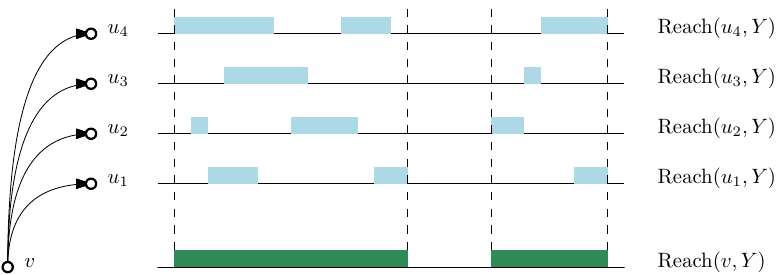}
    \caption{Let $u_1,u_2,u_3,u_4$ be the out-neighbors of $v$ in $\widehat{H}$.
    For each $i \in [4]$ we represent the set $\reach_{\widehat H}(u_i,Y) \sub Y$ as a union of $\widetilde\Oh(|Y|^{1-\delta})$ intervals with respect to the ordering $\pi$ of $Y$, as shown on the right.
    The set $\reach_{\widehat H}(v,Y)$ is the union of $\reach_{\widehat H}(u_i,Y)$ for $i \in [4]$.
    We can compute the interval representation of $\reach_{\widehat H}(v,Y)$ in time near-linear in the total number of the blue intervals.
    }
    \label{fig:interval}
\end{figure}

In the above sketch we have assumed that we know the ordering $\pi$ of $Y$ which yields sublinear interval representation of every set in $\reachsys(\widehat H, Y)$.
Of course, we do not know $\reachsys(\widehat H, Y)$ upfront and so $\pi$ has to be computed on the run.
However, the number of distinct sets in $\reachsys(\widehat H, Y)$ is $\Oh(|Y|^d)$ due to
Sauer–Shelah Lemma.
We can thus afford to recompute
$\pi$ and all the interval representations each time we observe a new set from $\reachsys(\widehat H, Y)$.
In summary, we obtain running time $\widetilde\Oh(|\widehat H| \cdot |Y|^{1-\delta} + |Y|^{\Oh(d)})$ which is exactly what we needed in \Cref{subsec:twolayers}.

For the later applications in the online setting, we need two refinements of this construction.
First, we would like to reduce the size of the emulator for $(\widehat H, Y)$ to be truly subquadratic in $|Y|$.
To this end, we apply the classic \coldef{segment tree} technique~\cite{bentley1977algorithms}
to build a binary tree of size $\Oh(|Y|)$ in which every node represents some ``base'' interval over $\pi$ and every other interval is a union of $\Oh(\log |Y|)$ base intervals (see \Cref{fig:segment-tree} on page \pageref{fig:segment-tree}).
This allows us to reduce the out-degree of every vertex in the emulator to $\widetilde\Oh(|Y|^{1-\delta})$.
Second, we would like to get rid of the additive term $|Y|^{\Oh(d)}$ and obtain purely subquadratic running time.
For this, we incorporate the argument used in the VC-based subquadratic algorithm for computing graph diameter~\cite{Ducoffe22}.
We split the set $Y$ into $|Y|^{1-\gamma}$ subsets $Y_1,\dots,Y_\ell$ of size $|Y|^{\gamma}$ for some small $\gamma > 0$.
Then we compute sets $\reach_{\widehat H}(v,Y_i)$ for each $v \in Y$ and $i \in [\ell]$ by running the algorithm independently for every $i \in [\ell]$.
This reduces the time spent on recomputing $\pi$ and the total running time can be estimated as $\widetilde\Oh(|\widehat H| \cdot |Y|^{1-\gamma\delta})$.

\subsection{Faithful emulators}
\label{subsec:faithful}

Now we justify the assumption that the set system
$\reachsys(\widehat H, Y)$ has bounded VC dimension.
Let us first understand what can go wrong if we compute the \re $H$ for $G \ominus X$ using  Thorup's algorithm~\cite{thorup2004compact}.
It relies on the following property: every spanning tree in an undirected planar graph contains three root-to-leaf paths whose union forms a balanced separator.
Thorup leveraged this fact to decompose a planar digraph into (overlapping) pieces, each of which admits a balanced separator~$S$ consisting of six directed paths $P_i$.
When processing a single piece,
in total linear time one can mark, for every $v \in X$ and $i \in [6]$, the first vertex on $P_i$ that can be  reached from $v$ and the last vertex on $P_i$ that can reach $v$.
In order to preserve information about paths traversing $S$,
it suffices to keep the $\Oh(|X|)$ marked vertices on $S$ and insert the edges as specified above.
Then one recurses on the weakly connected components obtained after removal of $S$ and
creates $\Oh(|X|\log |X|)$ vertices and edges in total.

Unfortunately, even though the VC dimension of  $\reachsys(G \ominus X,X)$ is bounded,
it might be not the case for $\reachsys(H,X)$.
This is because the marking procedure adds to $V(H)$
``Steiner vertices'' from $S \sm X$ and provides no control over their reachability sets:
they can be smaller than their original reachability sets which may increase the VC dimension.
Note that in \Cref{subsec:vc} we needed an interval representation of $\reach_{\widehat H}(v,Y)$ for every vertex $v$, not only these in $Y$.
To mitigate this, we modify the Thorup's construction in two places.
First, we rearrange the order of operations so that we avoid the
decomposition phase.
Second, for every vertex $v \in V(H)$ and every dipath $P_i$ in a balanced separator from each of $\Oh(\log |X|)$ ancestral recursive calls, we insert an edge from $v$ to the first marked vertex on $P_i$ that is reachable from $v$.
This increases the number of edges in $H$ to  $\Oh(|X|\log^2 |X|)$.
Thanks to this modification, we preserve all the connections from $v$ to $X$, that is, $\reach_H(v,X) = \reach_G(v,X)$ (see \Cref{fig:thorup-sep} on page \pageref{fig:thorup-sep}).
We call a \re with this property \coldef{faithful}.

Faithfulness ensures that $\reachsys(H,X) \sub \reachsys(G \ominus X,X)$ so the VC dimension cannot increase.
Moreover, this property is preserved by operations such as edge insertion:
when $H$ is a \fre for $(G \ominus X,X)$ and $u,v \in X$ then $H + (u,v)$ is a \fre for $(G \ominus X + (u,v),\,X)$.
By inspecting the construction of $\widehat H$ in \Cref{subsec:twolayers} we can thus prove that also
$\widehat H$ is a \fre for  $(\widehat G \ominus Y, X)$ and hence also for $(\widehat G \ominus Y, Y)$.
This allows us to bound the VC dimension of $\reachsys(\widehat H, Y)$ and justifies the assumption from \Cref{subsec:vc}, concluding the proof sketch of \Cref{thm:main:offine}.

\newcommand{\bnd}{\partial}
\newcommand{\rdiv}{\mathcal{R}}

\subsection{Online reachability}\label{subsec:online}
Let us start by reviewing a variant of the state-of-the-art online dynamic reachability data structure for planar graphs~\cite{subramanian1993fully}.
For simplicity, we assume that we only need to track reachability for a single fixed vertex pair $(s,t)$.
Then, we can also focus on optimizing the update time exclusively.

The data structure of~\cite{subramanian1993fully} is essentially obtained via a combination of an \coldef{$r$-division}~\cite{DBLP:journals/siamcomp/Frederickson87} with reachability emulators.
An $r$-division $\rdiv$ (see~\Cref{s:rdiv}) is a partition of a graph into $\Oh(n/r)$ pieces, each of size $\Oh(r)$, such that each piece $P$ shares $\Oh(\sqrt{r})$ vertices with other pieces.
The shared vertices are called the boundary vertices of $P$ and are denoted by $\bnd{P}$.
For simplicity, let us assume that the vertex~$s$ (resp. vertex $t$) is included in the boundary set $\bnd{P}$ of each piece that contains $s$ (resp. $t$).

Suppose that one constructs, for each piece $P\in\rdiv$, a reachability emulator $H_P$ for $(P,\bnd{P})$.
Then, since the pieces interact only through their boundary vertices, one can prove that the union $H$ of piecewise emulators $H_P$ through all $P\in\rdiv$ is a reachability emulator for $(G,\bnd{\rdiv})$, where $\bnd{\rdiv}=\bigcup_{P\in\rdiv}\bnd{P}$.
Since $\{s,t\}\subseteq \bnd{\rdiv}$,
one can determine if an $s\to t$ path exists in $G$ by running graph search on $H$ in $\Oh(|H|)$ time.
In particular, if one uses Thorup's reachability emulators~\cite{thorup2004compact} for $H_P$, each of size $\Ot(|\bnd{P}|)$, this yields $|H|=\Ot(n/\sqrt{r})$.

To make the data structure dynamic, it is enough to maintain the $r$-division and the associated emulators $H_P$ under edge updates.
Note that an $r$-division can accept $\Theta(\sqrt{r})$ edge insertions while not compromising its properties\footnote{The number of pieces may increase to $\Theta(n/r+\sqrt{r})$, which exceeds $\Oh(n/r)$ for $r>n^{2/3}$, but this is not an issue in our case.}:
one can make the inserted edges single-edge pieces with both its endpoints boundary vertices. Hence, only $\Oh(1)$ pieces are affected by a single update: a non-boundary vertex of an existing piece $P$ might need to be inserted into~$\bnd{P}$.
After $\Theta(\sqrt{r})$ edge updates we recompute the $r$-division from scratch, leading to amortized cost $\Ot(n/\sqrt{r})$.

The challenging part though is to maintain the emulators $H_P$ under (1) edge deletions and (2) boundary set extensions.
Specifically, it is not clear how to update the emulator $H_P$ faster
than by recomputing it from scratch in $\Ot(r)$ time.
Hence, the update procedure of such a data structure has two bottlenecks:
the $\Ot(n/\sqrt{r})$ cost of running graph search on the ``global'' emulator $H$, and $\Ot(r)$ time to recompute $\Oh(1)$ ``local'' emulators $H_P$.
The optimal choice is $r=n^{2/3}$, leading to $\Ot(n^{2/3})$ update time.

It is worth noting that Subramanian~\cite{subramanian1993fully} used a stronger version of an $r$-division, where each piece's boundary $\bnd{P}$ lies on $\Oh(1)$ faces (called \coldef{holes}) of some embedding of $P$.
This allows one to~rely on weaker emulators that require the terminal set to be located on few faces of the ground plane graph -- recall that~\cite{subramanian1993fully} predates Thorup's work~\cite{thorup2004compact} by almost a decade.
Subramanian's emulator can also be computed statically in $\Ot(r)$ time and has size $\Ot(|\bnd{P}|)$ as well.
The stronger $r$-division is more expensive to maintain, though, as the pieces' boundaries now have to obey an extra guarantee.
Nevertheless, if one applies an additional $\Ot(r)$-time rebalancing step to $\Oh(1)$ pieces per update, the stronger guarantee can be preserved.
However, this introduces another $\Ot(r)$ bottleneck per update.

Our data structure that breaches the $n^{2/3}$ bound in the most general setting of planarity-preserving updates follows the above framework of~\cite{subramanian1993fully}.
However, we show that one can in $\Oh(r^{0.99})$ time handle the situation where the pieces have only a limited number of \emph{non-structured} boundary vertices on average, which is the case when an $r$-division undergoes only $\Oh(\sqrt{r})$ updates.
To this end, we employ both types of our emulators as well as the decremental data structure from~\cite{DBLP:conf/soda/Karczmarz18} that can be applied to static well-structured boundaries.
Recall that, on the other hand, we currently do not know a better-than-naive way to update Thorup's reachability emulators under edge deletions.

We maintain an \coldef{$r$-division with few holes} (as in~\cite{subramanian1993fully} and~\cite{KleinMS13}, see also~\Cref{s:rdiv}) in $\Ot(n/\sqrt{r})$ amortized time per update as described before.
That is, every $\sqrt{r}$ updates, we restart in $\Ot(n)$ time with a fresh $r$-division $\rdiv$ with few holes.
Each inserted edge turns into a separate piece of the $r$-division, so that the original pieces of $\rdiv$ only undergo edge deletions.
Let $B$ be the set of $\Oh(\sqrt{r})$ vertices constituting the endpoints of the edges inserted to $G$ since the last restart, excluding vertices of $\bnd\rdiv$.
For each piece $P\in\rdiv$, we maintain a reachability emulator $H_P$ for $(P,\,\bnd{P}\cup B_P)$, where $B_P=B\cap V(P)$.
Our main idea here is to treat the vertices $B$ differently from the initial boundary vertices $\bnd{P}$ when updating the
piecewise emulators.
To this end, we need to generalize the concept of a reachability emulator to the asymmetric case: we say that $J$ is a reachability emulator for $(G,X,Y)$ if $J$ preserves the reachability relation in $G$ for all pairs $(x,y)\in X\times Y$.
Given that, our goal is to separately maintain emulators for $(P,\,\bnd{P})$, $(P,\,B_P,\,\bnd{P}\cup B_P)$ and $(P,\,\bnd{P}\cup B_P,\,B_P)$ and combine them into the desired emulator $H_P$.

First, since the piece $P$ is decremental and $\bnd{P}$ does not grow,
we may use the specialized decremental transitive closure data structure from~\cite{DBLP:conf/soda/Karczmarz18} to maintain a reachability emulator for $(P,\bnd{P})$.
While this data structure guarantees $\Ot(r)$ total update time and $\Ot(\sqrt{r})$ query time, it does not operate on reachability emulators per se.
Nevertheless, we show that it can be adjusted to yield an $\Ot(|\bnd{P}|)$-sized emulator for $(P,\bnd{P})$ after each deletion with only $\Ot(|\bnd{P}|)=\Ot(\sqrt{r})$ additional overhead per update.
The work~\cite{DBLP:conf/soda/Karczmarz18} is the only source of randomization and amortization in our data structure because the $r$-division reconstruction step can be deamortized in a standard way (see \Cref{sec:online-minor-free}).

Now consider the task of maintaining the emulator $H_P'$ for $(P,\,B_P,\,\bnd{P}\cup B_P)$; dealing with
\linebreak $(P,\,\bnd{P}\cup B_P,\,B_P)$ is analogous by edge reversal.
It suffices to show that $H_P'$ can be updated under either an edge deletion in~$P$ or an insertion of a new vertex to $B_P$ in $\Oh(r^{0.99})$
time in such a way that the total size of all emulators over each $P\in\rdiv$ remains $\Ot(n/\sqrt{r}+r^{0.99})$.
To this end, for each piece $P$ we construct a separate $s$-division $\rdiv_P$ for some $s = r^{1-\Omega(1)}$.
For each subpiece $Q\in\rdiv_P$, we maintain a \emph{faithful} emulator  $H_{P,Q}$ for terminals $\bnd{Q}$ plus the vertices $(\bnd{P}\cup B_P)\cap V(Q)$.
Whenever an edge is deleted from $Q$ or $Q$ receives a new terminal, the emulator $H_{P,Q}$ is recomputed from scratch in $\Ot(s)$ time.
The collection of faithful emulators can be combined into a faithful emulator $H_0$ for $(P,\,\bnd{P}\cup \bnd{\rdiv_P}\cup B_P)$ of size $\Ot(r/\sqrt{s}+|B_P|)=\Ot(r/\sqrt{s})$.
It is still too large to guarantee the total bound $\Ot(n/\sqrt{r}+r^{0.99})$, so we need to extract from $H_0$ a smaller emulator $H_P'$ for $(X,Y) = (B_P,\, \bnd P\cup B_P)$.
By faithfulness, $\reachsys(H_0,\,\bnd{P}\cup B_P)$ has bounded VC dimension.
As explained in the end of~\Cref{subsec:vc}, we can compute emulator $H_P'$ of size
\[\Ot\left(|B_P|\cdot |\bnd{P}\cup B_P|^{1-2\epsilon}+|\bnd{P}\cup B_P|\right)=\Ot\left(|B_P|\cdot r^{1/2-\epsilon}+\sqrt{r}\right)\]
in time
\[\Ot\left(|H_0|\cdot |\bnd{P}\cup B_P|^{1-2\epsilon}\right)=\Ot\left(r^{3/2-\epsilon}/\sqrt{s}\right)\]
 for some sufficiently small $\epsilon>0$.
Indeed, the sum of $|H_P'|$ over all $P\in \rdiv$ is $\Ot(n/\sqrt{r}+r^{1-\epsilon})$.
By setting $s$ and $\eps$ appropriately,
we can also make the term $r^{3/2-\epsilon}/\sqrt{s}$ sublinear in $r$.
Finally, we achieve the desired update time by choosing $r$ slightly above $n^{2/3}$.
This concludes the proof sketch of \Cref{thm:planar-dyn-reach}.

\vspace{0.5cm}
\begin{figure}[h]
    \centering
    \includegraphics[width=\textwidth]{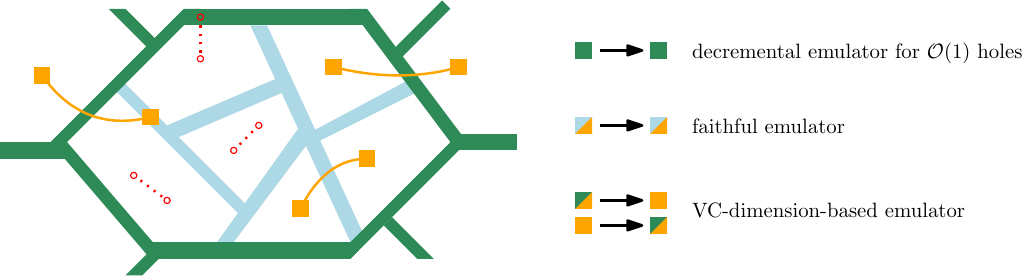}
    \caption{Illustration to the online algorithm. The green cell is the boundary $\bnd P$ of a piece $P$ from an $r$-division, depicted with a single hole.
    It is further decomposed into an $s$-division (in blue).
    The removed edges are marked as dotted red and the inserted ones are orange.
    Their endpoints outside the initial boundary form the set $B$.
    Each orange edge turns into a separate piece of the $r$-division.
    {We employ the decremental data structure of~\cite{DBLP:conf/soda/Karczmarz18} to maintain a \re keeping track of connections among $\bnd P$ within $P$ with respect to the removed edges.}
    We use faithful emulators to maintain reachability within each small piece of the $s$-division, including the vertices from $B_P$ that appeared in there.
    From these emulators we extract an asymmetric VC-dimension-based emulator of size $\Ot(|B_P|\cdot r^{1/2-\eps}+\sqrt{r})$ that keeps track of connections within $P$ from $B_P$ to $\bnd P \cup B_P$ and vice versa.
    The total size of the asymmetric modulators is sufficiently small because the sum of $|B_P|$ over all pieces $P$ is~$\Oh(\sqrt{r})$.
    }
    \label{fig:online}
\end{figure}

\section{Preliminaries}
\label{sec:prelim}

A digraph is a pair $(V,E)$ where $E \sub V\times V$.
A $(u,w)$-dipath in $G$ is a sequence $(v_1,\dots,v_\ell)$ such that $v_1 = u$, $v_\ell = w$, and $(v_i,v_{i+1}) \in E(G)$ for each $i \in [\ell-1]$.
We write
$u \overset{G}\leadsto w$ to indicate that $w$ is reachable from $u$ in $G$, that is, $G$ contains a $(u,w)$-dipath.
The {size} $|G|$ of a digraph $G$ is defined as $|V(G)| + |E(G)|$.

For a digraph $G$, $v \in V(G)$, $X \sub V(G)$, we define $\coltext{\reach}_G(v,X) \sub X$ as the set of vertices in $X$ that are reachable from $v$ in $G$.
Next, $\coltext{\reachsys}(G,X) = \{\reach_G(v,X) \colon v \in V(G)\}$ and $\reachsys(G) = \reachsys(G,V(G))$.

We use the $\Ot$-notation to hide factors polylogarithmic in the instance size.

\subsection{VC dimension}

A \coldef{set system} is a pair $\scal = (X,\rcal)$ where $X$ is called the \coldef{ground set} and $\rcal \sub 2^X$.
The \coldef{dual set system} of $\scal$ is $\scal^* = (\rcal, X^*)$ where $X^* = \{\rcal_x \colon x \in X\}$ and $\rcal_x \sub \rcal$ comprises these sets that contain~$x$.
A set $Z \sub X$ is \coldef{shattered} by $(X,\rcal)$ if for every $Z' \sub Z$ there is $R \in \rcal$ such that $Z' = Z \cap R$.
The \coldef{Vapnik-Chervonenkis dimension} (VC dimension) of $\scal$ is the largest size of a set shattered by $\scal$.
The dual VC dimension of $\scal$ is the VC dimension of $\scal^*$.
It is known that if the VC dimension of $\scal$ is bounded by a constant then so is its dual VC dimension, and vice versa~\cite{chazelle1989quasi}.
It is easy to see that both measures are monotone under truncation of a set system.

\begin{observation}\label{lem:prelim:truncate}
    Let $\scal = (X,\rcal)$ be a set system, $Y \sub X$, and $\scal' = (Y,\rcal')$ where $\rcal' = \{R \cap Y \colon R \in \rcal\}$.
    Then the VC dimension (resp. dual VC dimension) of $\scal'$ is at most the VC dimension (resp. dual VC dimension)~of~$\scal$.
\end{observation}

The classic Sauer–Shelah Lemma states that a set system of bounded VC dimension can contain only polynomially many sets.

\begin{lemma}[Sauer–Shelah Lemma \cite{sauer1972density}]
\label{lem:prelim:size-bound}
    Let $(X,\rcal)$ be a set system  of VC dimension $d$.
    Then $|\rcal| = \Oh(|X|^d)$.
\end{lemma}

Consider a set system $(X,\rcal)$.
For an ordering $\pi = (x_1,x_2, \dots , x_n)$ of $X$ and a set $R \sub X$, let $\textcolor{\mycol}{\mathrm{stab}}(\pi,R)$ denote the set of indices $i \in [1,n-1]$ for which $|\{x_i,x_{i+1}\} \cap R| = 1$.
The \coldef{stabbing number} of $\pi$ is the maximum $|\mathrm{stab}(\pi,R)|$ over $R \in \rcal$.
Observe that each set in $\rcal$ can be represented as a union of at most $|\mathrm{stab}(\pi,R)|+1$ intervals over $\pi$.

Chazelle and Welzl~\cite{chazelle1989quasi, welzl1988partition} proved that when $(X,\rcal)$ has dual VC dimension $d$, then $X$ admits an ordering with stabbing number $\Oh_d(n^{1-1/d}\cdot \log n)$.
\footnote{Throughout this paper, the $\Oh_d(\cdot)$ notation allows the hidden constant to depend on $d$.}
They also provided a polynomial algorithm to find such an ordering,
whose rough complexity estimation appears under {\cite[Theorem 4.2]{welzl1988partition}}.
Faster approximation algorithms for this task rely on LP relaxations~\cite{bilo2004crossing, harpeled2009, Chekuri21} but they involve randomized
LP rounding procedures.
Another fast randomized algorithm is given in~\cite{chan2025truly} but it only provides a bound of the form $\sum_{R\in\rcal}|\mathrm{stab}(\pi,R)| = \widetilde\Oh(m\cdot n^{1-1/d})$, which is insufficient in our arguments (e.g., \Cref{lem:vsbased:dag}).

To avoid randomization in \Cref{thm:main:offine}, we employ the algorithm from~\cite{chazelle1989quasi,welzl1988partition}.
Since the analysis in {\cite[Theorem 4.2]{welzl1988partition}} omits the treatment of large numbers, we provide a detailed analysis for the sake of completeness, following the exposition in \cite{chazelle1989quasi}.

\begin{theorem}[{\cite[Theorem 4.3]{chazelle1989quasi}}]
\label{thm:prelim:stab}
    Let $(X,\rcal)$ be a set system of dual VC dimension $d \ge 2$ and $n = |X|$, $m = |\rcal|$.
    Then $X$ admits an ordering $\pi$
    with stabbing number $\Oh_d(n^{1-1/d}\cdot \log n)$ with respect to $\rcal$.
    Furthermore, $\pi$ can be computed in time $\Oh(n^3(n+m))$.
\end{theorem}
\begin{proof}
    For a set system $(Y,\qcal)$ and distinct $x,y \in Y$, let $\qcal_{x,y}$ denote the family of sets $Q \in \qcal$ satisfying $|\{x,y\} \cap Q| = 1$.
    \cite[Lemma 4.1]{chazelle1989quasi} states that when the dual VC dimension of $(Y,\qcal)$ is $d$ and $w\colon \qcal \to \nn$ is any function,
    then there exist $x,y \in Y$ for which
    \begin{equation}\label{eq:chazelle}
    \sum_{q \in \qcal_{x,y}} w(Q) \le c_d \cdot \sum_{q \in \qcal} w(Q)\cdot \log|Y| \cdot |Y|^{-1/d}
    \end{equation}
    where $c_d$ depends on $d$.

    The algorithm described in \cite[Lemma 4.2]{chazelle1989quasi} starts with $(Y,\qcal) = (X,\rcal)$, $w \equiv 1$, and finds $x,y \in Y$ for which \eqref{eq:chazelle} holds.
    Then the edge $xy$ is added to an intermediate spanning tree of $X$, one of $x,y$ is removed from $Y$, and for each $Q \in \qcal_{x,y}$ the weight $w(Q)$ is multiplied by two.
    The total number of such iterations in \cite[Theorem 4.3]{chazelle1989quasi} is less than $n$ and so each weight $w(Q)$ is always of the form $2^i$ where $0 \le i < n$.

    We justify that a single iteration can be implemented to run in time $\Oh(n^2(n+m))$.
    For each pair $x,y \in Y$ the sum $w(\qcal_{x,y}) := \sum_{Q \in \qcal_{x,y}} w(Q)$ can be computed in time linear in $n + |\qcal_{x,y}| \le n + m$
    because each $w(Q)$ is a power of two.
    This is performed by inspecting the bits from the least significant and moving the surplus to the next bit.
    Note that the number of bits needed to represent $w(\qcal_{x,y})$ is at most $2n$.
    Next, we find a pair $x,y \in Y$ for which $w(\qcal_{x,y})$ is minimized.
    This requires $\Oh(n^2)$ comparisons, each taking time $\Oh(n)$.
    The chosen pair satisfies condition \eqref{eq:chazelle}.
\end{proof}

VC dimension bound $d$ implies a bound of $2^{d+1}$ on the dual VC dimension~\cite{chazelle1989quasi}.
However, in our applications in reachability this relation can be tightened.

\begin{lemma}\label{thm:prelim:dual-vc}
    Let $G$ be a digraph and $G^r$ be obtained from $G$ by reversing all edges.
    Then the dual VC dimension of $\reachsys(G)$ equals the VC dimension of $\reachsys(G^r)$.
\end{lemma}
\begin{proof}
    Let $\rcal = \reachsys(G)$.
    Next, let $\rho \colon \rcal \to V(G)$ assign to $R \in \rcal$ a vertex $v \in V(G)$ such that $\reach_G(v,V(G)) = R$.
    Recall that for $v \in V(G)$ we define $\rcal_v \sub \rcal$ as the subfamily of sets containing $v$.
    The dual VC dimension of $\reachsys(G)$ equals the maximum size of a subset of $\rcal$ that is shattered by the family $\{\rcal_v \colon v \in V(G)\}$.
    Consider a set $\scal \sub \rcal$ that is shattered by this family and let $S = \rho(\scal)$ (we use convention $\rho(\scal) = \bigcup_{s \in \scal} \rho(s)$).
    For every $\qcal \sub \scal$ there exists $v \in V(G)$ such that $\rcal_v \cap \scal = \qcal$.
    It means that the only vertices in $S$ that can reach $v$ are these in $\rho(\qcal)$.
    In other words, $\reach_{G^r}(v, S) = \rho(\qcal)$.
    Because the mapping $\rho$ is injective, we obtain that the set $S$ is shattered by $\reachsys(G^r)$.
    Consequently, when the dual set system to $\reachsys(G)$ has a shattered set of size $d$, then $\reachsys(G^r)$ has a shattered set of size $d$ as well, which gives the $(\le)$ inequality of the lemma.
    The opposite direction is analogous.
\end{proof}

When $G$ is a $K_h$-minor-free digraph,
the bound $h-1$ on the VC dimension of  $\reachsys(G)$ follows from a more general bound for  the set system of directed balls.
In light of \Cref{thm:prelim:dual-vc}, this implies a bound on both the VC dimension and the dual VC dimension of $\reachsys(G)$.

\begin{lemma}[{\cite[Theorem 5]{LeW24}}]
\label{thm:prelim:minor-vc}
    Let $h \in \nn$ and $G$ be a $K_h$-minor-free digraph.
    Then $\reachsys(G)$ has both VC dimension and dual VC dimension at most~$h-1$.
\end{lemma}

Since planar graphs are $K_5$-minor-free, their reachability set systems can have VC dimension at most four.
By inspecting the structure of hypothetical shattered sets more closely, we improve this bound to three.
The proof is given
in \Cref{sec:improved-vc}.

\begin{restatable}{lemma}{lemPlanarVC}
\label{thm:prelim:planar-vc}
    Let $G$ be a planar digraph.
    Then $\reachsys(G)$ has both VC dimension and dual VC dimension at most~three.
\end{restatable}

These bounds are transferred to $\reachsys(G,X)$ for any $X \sub V(G)$ via \Cref{lem:prelim:truncate}.

\subsection{Reachability emulators}

This subsection covers the basic properties of \re{s}.
Below we introduce the novel concept of
 {\em faithfulness}.

\begin{definition}
    Let $G$ be a  digraph and $X,Y \sub V(G)$.
    A  digraph $H$ with $X \cup Y \sub V(H)$ is called a \coldef{reachability emulator}  for $(G,X,Y)$ if for every $u \in X$, $v \in Y$ it holds that: $u \overset{G}\leadsto v$ if and only if $u \overset{H}\leadsto v$.

    If $X=Y$ then we simply say that $H$ is a \re for $(G,X)$.

    Furthermore, when $H$ is a \re for $(G,X)$
    we say that $H$ is  \coldef{faithful} if additionally for every $u \in V(H)$ there exists $v \in V(G)$ such that $\reach_H(u,X) = \reach_G(v,X)$.
\end{definition}

\begin{observation}\label{obs:prelim:faithful-reach}
    Let $G$ be a digraph, $X \sub V(G)$, and $H$ be a \fre for $(G,X)$. Then $\reachsys(H,X) \sub \reachsys(G,X)$.
\end{observation}

\begin{observation}\label{obs:prelim:faithful-subset}
    Let $G$ be a digraph, $Y \sub X \sub V(G)$, and $H$ be a \fre for $(G,X)$. Then $H$ is also a \fre for $(G,Y)$.
\end{observation}

A convenient property of (faithful) \re{s} is that they can be easily combined.
In other words, emulating reachability inside a subgraph with respect to its boundary preserves reachability in the entire graph.

\begin{lemma}\label{lem:emulator-union}
    Let $G_1,G_2$ be digraphs, $X_1 \sub V(G_1),\, X_2 \sub V(G_2)$, such that $V(G_1) \cap V(G_2) = X_1 \cap X_2$.
    Next, let $H_1,H_2$ be \re{s}
    for $(G_1,X_1)$ and $(G_2,X_2)$, respectively, satisfying $V(H_1) \cap V(H_2) = X_1 \cap X_2$.
    Then the union $H_1 \cup H_2$ is a \re for $(G_1 \cup G_2, X_1 \cup X_2)$.
    Additionally, if  $H_1,H_2$ are faithful then so is   $H_1 \cup H_2$.
\end{lemma}
\begin{proof}
Consider any $u,v\in X_1\cup X_2$. A $u\to v$ path $P$ in $H_1\cup H_2$ (if it exists) can be decomposed into maximal subpaths $P_i=u_i\to u_{i+1}$ ($i=1,\ldots,k$) that are fully contained in either $H_1$ or $H_2$.
Note that for $1<i\leq k$, $u_i\in V(H_1)\cap V(H_2)=X_1\cap X_2$.
This shows that for each $i$, if $P_i$ is fully contained in $H_j$, $j\in \{1,2\}$, then $u_i,u_{i+1}\in X_j\subseteq V(G_j)$.
Thus, the path $P_i$ certifies that $u_i$ can reach $u_{i+1}$ in $G_j$ and thus also in $G_1\cup G_2$.
Therefore, we have $u=u_1\overset{G_1\cup G_2}\leadsto u_2\overset{G_1\cup G_2}\leadsto \ldots \overset{G_1\cup G_2}\leadsto u_{k+1}=v$ as well.

Now consider a $u\to v$ path $Q$ in $G_1\cup G_2$, if one exists.
Decompose it into maximal subpaths $Q_i=u_i\to u_{i+1}$ ($i=1,\ldots,l$) that are fully contained in either $G_1$ or $G_2$.
We now argue that for each $i$, we have $u_i\overset{H_j}\leadsto u_{i+1}$ for some $j\in \{1,2\}$.
This will prove that $u=u_1\overset{H_1\cup H_2}\leadsto \ldots \overset{H_1\cup H_2}\leadsto u_{l+1}=v$.
Note that if $Q_1\subseteq G_j$, then $u_1\in (X_1\cup X_2)\cap V(G_j)=X_j$.
Similarly, if $Q_l\subseteq G_j$, then $u_{l+1}\in X_j$.
For each $1<i\leq l$ we have $u_i\in V(G_1)\cap V(G_2)=X_1\cap X_2$.
This proves that for any $i$, if $Q_i\subseteq G_j$, then $u_i,u_{i+1}\in X_j$.
In such a case, $u_i\overset{H_j}\leadsto u_{i+1}$ follows by the definition of $H_j$.

Finally, suppose now that $H_1,H_2$ are additionally faithful and consider any $z\in V(H_1)\cup V(H_2)$.
Without loss of generality assume that $z\in V(H_1)$ and let $z'\in V(G_1)$ be such that $\reach_{H_1}(z,X_1)=\reach_{G_1}(z',X_1)$.
We will prove that $\reach_{H_1\cup H_2}(z,X_1\cup X_2)=\reach_{G_1\cup G_2}(z',X_1\cup X_2)$.

Let $x\in \reach_{H_1\cup H_2}(z,X_1\cup X_2)$ and consider a $z\to x$ path $R$ in $H_1\cup H_2$.
Note that $R$ has some first vertex $y\in X_1$:
if $R$ is fully contained in $H_1$, then the endpoint $x$ belongs to $X_1$,
and otherwise, $R$ contains a vertex from $V(H_1)\cap V(H_2)=X_1\cap X_2$.
Note that the $z\to y$ subpath of $R$ is fully contained in $H_1$,
and thus we obtain $y\in \reach_{G_1}(z',X_1)$.
Moreover, since $H_1\cup H_2$ is a reachability emulator for $(G_1\cup G_2,X_1\cup X_2)$, $x\in \reach_{G_1\cup G_2}(y,X_1\cup X_2)$.
Thus proves that a path $z'\to y\to x$ exists in $G_1\cup G_2$.

Now take any $x'\in \reach_{G_1\cup G_2}(z',X_1\cup X_2)$.
Again, one can argue similarly as before that any path $S=z'\to x'$ in $G_1\cup G_2$ has some first vertex $y'$ in $X_1$.
This implies that $x'\in \reach_{H_1\cup H_2}(y',X_1\cup X_2)$
and $y'\in \reach_{H_1}(z,X_1)$.
Consequently, we get $x'\in \reach_{H_1\cup H_2}(z,X_1\cup X_2)$.
\end{proof}

\begin{corollary}\label{lem:emulator-union-general}
    Let $k\geq 1$. For $i=1,\ldots,k$, let $G_i$ be a digraph and let $X_i\subseteq V(G_i)$.
    Let $H_i$ be a \re for $(G_i,X_i)$.
    Suppose $V(G_i)\cap V(G_j)=V(H_i)\cap V(H_j)=X_i\cap X_j$ for all $1\leq i,j\leq k$, $i\neq j$.
    Then $H_1\cup \ldots\cup H_k$ is a \re for $(G_1\cup\ldots\cup G_k,X_1\cup \ldots\cup X_k)$.
    Additionally, if $H_1,\ldots,H_k$ are all faithful, then so is $H_1\cup\ldots\cup H_k$.
\end{corollary}
\begin{proof}
For $j=1,\ldots,k$, let $H_j^*=\bigcup_{i=1}^j H_i$, and similarly $G_j^*=\bigcup_{i=1}^j G_i$ and $X_j^*=\bigcup_{i=1}^j X_i$.
We prove the claim by induction on $k$.
For $k=1$ the claim is trivial.

Suppose $k\geq 2$ and $H^*_{k-1}$ is a reachability emulator for $(G^*_{k-1},X^*_{k-1})$.
We have $V(G^*_{k-1})\cap V(G_k)=\bigcup_{i=1}^{k-1} (V(G_i)\cap V(G_k))=\bigcup_{i=1}^{k-1}(X_i\cap X_k)=\left(\bigcup_{i=1}^{k-1}X_i\right)\cap X_k=X^*_{k-1}\cap X_k$.
Similarly, we obtain \linebreak $V(H^*_{k-1})\cap V(H_k)=X^*_{k-1}\cap X_k$.
Thus, by~\Cref{lem:emulator-union}, $H^*_{k-1}\cup H_k=H^*_k$ is a reachability emulator for $(G^*_{k-1}\cup G_k,X^*_{k-1}\cup X_k)=(G^*_k,X^*_k)$, as desired.
Similarly, if $H_1,\ldots,H_{k-1}$ are faithful, then so is $H^*_{k-1}$. Hence, if $H_k$ is faithful, then so is $H^*_k$.
\end{proof}

\begin{corollary}\label{lem:emulator-union-simple}
    Let $G, J$ be digraphs, $V(J) \sub X \sub V(G)$, and $F$ be a \fre for $(G,X)$.
    Then $F \cup J$ is a \fre for $(G \cup J,X)$.
\end{corollary}
\begin{proof}
Digraph $J$ is trivially a faithful reachability emulator for $(J,V(J))$ and we have $V(J)\cap V(G)=V(J)\cap X=V(J)$, and similarly $V(J)\cap V(F)=V(J)\cap X=V(J)$, so by~\Cref{lem:emulator-union},
$F\cup J$ is a faithful reachability emulator for $(G\cup J,\,X\cup V(J))=(G\cup J,X)$.
\end{proof}

In the last preliminary lemma, we observe that in order to preserve reachability among $X \times Y$ one can split  $X \times Y$ into several sets and compute an emulator for each of them separately.

\begin{lemma}\label{l:assemble-from-parts}
    Let $G$ be a digraph, $X,Y \sub V(G)$,
    and $(X_1,Y_1), \dots, (X_\ell,Y_\ell)$ be a collection of pairs such that $X_i \sub X$, $Y_i \sub Y$, and for every pair $(x,y) \in X\times Y$ there exists $i \in [\ell]$ with $x \in X_i,\,y\in Y_i$.
    Next,
    {suppose that for each $i\in[\ell]$  we are given
    a \re $H_i$} for $(G,X_i,Y_i)$ .
    Then a reachability emulator for {$(G,X,Y)$} of size $\Oh(|H_1|+\ldots+|H_\ell|)$ can be computed in $\Oh(|H_1|+\ldots+|H_\ell|)$ time.
\end{lemma}
\begin{proof}
For any $i\in[\ell]$ and every $v\in V(H_i)$, introduce a separate copy $v_i$ of vertex $v$.
Let $H_i'$ be the graph~$H_i$ with each vertex $u\in V(H_i)$ replaced with $u_i$, i.e., $H_i'=\{\{u_i:u\in V(H_i)\},\{(u_i,v_i):(u,v)\in E(H_i)\}\}$.
Note that no two graphs $H_i',H_j'$, $i\neq j$, share a common vertex.
Let $F$ be a graph containing edges from the original vertices of $X$ to their copies, i.e., $V(F)=X\cup \{v_i:v\in X_i,i\in[\ell]\}$ and \linebreak $E(F)=\{(v,v_i): v\in X_i, i\in[\ell]\}$.
Similarly, let $B$ be a graph induced on edges from the copies of vertices in $Y$ to their original counterparts, i.e., $V(B)=Y\cup \{v_i:v\in Y_i,i\in[\ell]\}$ and $E(B)=\{(v_i,v) : v\in Y_i,i\in[\ell]\}$.
We will prove that $H^*:=F\cup B\cup H_1' \cup \dots \cup H_\ell'$ is a \re for $(G,X,Y)$.

Let $(x,y)\in X\times Y$ and let $i$ be such that $x\in X_i$ and $y\in Y_i$.
If there is a path $x\to y$ in $G$, then there is also one in $H_i$.
Equivalently, there is a $x_i\to y_i$ path in $H_i'$, and thus
also a path $(x,x_i)\cdot x_i\to y_i\cdot (y_i,y)$ in $H^*$.
This proves $\reach_G(x,Y)\subseteq \reach_{H^*}(x,Y)$ for each $x\in X$.

Now consider an $x\to y$ path $P$ in $H^*$.
Split $P$ at every occurrence of an original vertex in $X\cup Y$, obtaining subpaths $P_1,\ldots,P_k$.
Each subpath $P_i=a\to b$ enters and leaves a single copy $H_j'$ and is of the form $(a,a_j)\cdot (a_j\to b_j)\cdot (b_j,b)$, where $a\in X_j$, $b\in Y_j$, and the $a_j\to b_j$ path is fully contained in $H_j'$.
By the definitions of $H_j$ and $H_j'$, there is a path $P_i'=a\to b$ in~$G$.
Consequently, by joining the corresponding paths $P_i'$, we obtain that there is an $x\to y$ path $P_1'\cdots P'_k$ in $G$ as well.
We have thus proved that $\reach_{H^*}(x,Y)\subseteq \reach_G(x,Y)$ holds for each $x\in X$.

The $\Oh(|H_1|+\ldots+|H_\ell|)$ size bound on $H^*$ follows by $|F|+|B|=O\left(\sum_{i}|X_i|+|Y_i|\right)$.
Clearly, $H^*$ can be constructed in time linear in its size.
\end{proof}

\section{Constructing a faithful emulator}
\label{sec:faithful}

\newcommand{\first}{{\color{\mycol}{\mathrm{first}}}}
\newcommand{\rel}{{\color{\mycol}{\mathrm{rel}}}}
\newcommand{\last}{{\color{\mycol}{\mathrm{last}}}}
\newcommand{\layer}{{\color{\mycol}{\mathrm{layer}}}}

In this section we prove that when a digraph $G$ is planar then one can construct a \fre for $(G,X)$ of size $\Oh(|X|\log^2|X|)$ in near-linear time.
We use the same building blocks as Thorup~\cite{{thorup2004compact}} in his compact reachability oracle for planar digraphs but we put them together in a different fashion.
We remark that yet another combination of these ideas has been used in~\cite{Goranci20} to construct a
{\em reachability preserving minor} for $(G,X)$ of size $\Oh(|X|^2\log |X|)$ (see also \Cref{sec:lower-bound}).

We begin by decomposing $G$ into layers, similarly as in the proof of
\cite[Lemma 2.2]{thorup2004compact}.
However, here we do not use them to decompose $G$ into pieces, each with a separate spanning tree, but we construct a single spanning tree piercing all the layers.
Avoiding the decomposition phase will facilitate inserting additional edges that are needed for faithfulness.

In the statement below, given a tree $T$ and $u,v \in V(T)$ we denote by $T[u \to v]$ the unique $(u,v)$-path in $T$, oriented from $u$ to $v$.

\begin{lemma}\label{lem:faithful:layers}
    Given a weakly connected $n$-vertex digraph $G$ we can in linear time compute a layering function $\layer \colon V(G) \to [n]$ and a rooted spanning tree $T$ of $G$ so that the following holds for any $u,v \in V(G)$.
    \begin{enumerate}\item When $u \overset{G}\leadsto v$ then $|\layer(u) - \layer(v)| \le 1$. \label{item:layers:diff}
        \item When $u$ is an ancestor of $v$ in $T$ then $\layer(u) \le \layer(v)$. \label{item:layers:monotone}
        \item When $u$ is an ancestor of $v$ in $T$
        and $\layer(u) = \layer(v)$ then either $T[u\to v]$ or $T[v \to u]$ forms a~dipath in~$G$. \label{item:layers:dipath}
    \end{enumerate}
\end{lemma}
\begin{proof}
    We construct a partitioning $V(G) = L_1 \cup L_2 \cup \dots \cup L_n$ corresponding to the layering function.
    Let $r \in V(G)$ be an arbitrary vertex.
    We define $L_1 = \reach_G(r,V(G))$.
    In the formula below, $L_{<i}$ stands for $L_1 \cup L_2 \cup \dots \cup L_{i-1}$ and
    $v \leadsto X$ (resp. $X \leadsto v$) means that some vertex in $X$ is reachable from $v$ (resp. $v$ is reachable from some vertex in $X$).
    For $i > 1$ we define

    \begin{align*}
    L_i=\begin{cases}
        v \in V(G) \sm L_{<i}\, \colon\, v \overset{G}\leadsto L_{<i} & \text{if }i\text{ is even},\\
        v \in V(G) \sm L_{<i}\, \colon\, L_{<i} \overset{G}\leadsto v & \text{if }i\text{ is odd}.
    \end{cases}
\end{align*}

It is straightforward to construct the sets $(L_i)_{i=1}^n$ is linear time.
Since $G$ is weakly connected, every vertex belongs to some $L_i$.
We set $\layer(v) = i$ for every $v \in L_i$.

We prove condition \eqref{item:layers:diff}.
Consider any dipath $P$ from $u$ to $v$.
Let $i = \min(\layer(w))$ over $w \in V(P)$ and let $x \in V(P) \cap L_i$.
If $i$ is even then every vertex after $x$ on $P$ is
reachable from $L_{\le i}$ so it is contained in $L_{\le i+1}$.
Also, every vertex $y$ occurring before $x$ on $P$ can reach $L_{\le i}$ so $y \in  L_{\le {i}}$.
Hence, $V(P) \sub L_i \cup L_{i+1}$.
The case when $i$ is odd is symmetric.

\begin{figure}[h]
    \centering
    \includegraphics[width=0.7\textwidth]{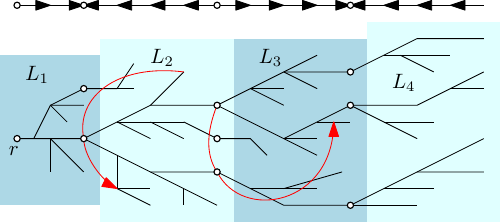}
    \caption{The layers $L_i$ and the tree $T$ in \Cref{lem:faithful:layers}.
    The top arrows indicate the orientation of edges precisely beneath them.
    The marked vertices are the ones in $L_i$ that are adjacent to $L_{i+1}$.
    There are two examples of dipaths in $G$ (red); each can visit only two consecutive layers.
    }
    \label{fig:thorup-layers}
\end{figure}

Now we construct the tree $T$.
For readers familiar with the proof in \cite{thorup2004compact}, we remark that the notation below, concerning graphs $G_i$, diverges from the one therein.
Let $G_1 = G[L_1]$ and $r_1 = r$.
For $i > 1$, let $G_i$ be the graph obtained from $G$ be removing vertices $y$ with $\layer(y) > i$ and contracting $L_{<i}$ into a single vertex $r_i$.
Moreover, when $i$ is even, we reverse all the edges in $G_i$.
Note that every edge in $G$ contributes to exactly one graph $G_i$ and so they can be constructed in linear time.
Next, observe that $r_i \overset{G_i}\leadsto v$ for every $v \in V(G_i)$ and every $i \ge 1$.
We compute a depth-first search tree $D_i$ in $G_i$ starting from $r_i$.
For $i > 1$, we can lift $D_i$ to a rooted forest $T_i$ in $G$, by replacing edges outgoing from $r_i$ with edges incident to $L_{<i}$.
For $i = 1$ we set $T_i = D_i$.
We define $T$ as the union of all $T_i$ over $i \ge 1$.

We argue that $T$ is indeed a spanning tree.
Let $T_{\le i}$ be the subgraph given by the union of $T_1,T_2,\dots,T_i$.
Clearly, $T_1$ is connected and
an inductive argument implies that each  $T_{\le i}$ is connected, therefore $T$ is connected as well.
For $i>1$ we have $|V(G_i)| = |L_i| + 1$ hence $|E(T_i)| = |E(D_i)| =  |V(G_i)| - 1 = |L_i|$.
Since $|E(T_1)| = |L_1| - 1$ we infer that $|E(T)| = n-1$.
Moreover, every vertex is incident to some edge from $T$.
A connected graph on $n$ vertices and $n-1$ edges, without isolated vertices, is indeed a spanning tree.

We set the root of $T$ to $r$.
To see condition \eqref{item:layers:monotone} consider a root-to-leaf path $P$ from $T$.
We argue that function $\layer$ is non-decreasing along $P$.
Suppose otherwise; then there exists $v$ on $P$ with a parent $u$ such that $i \coloneqq \layer(v) < \layer(u)$.
Let $w$ be the last vertex on $P$ before $u$ such that $\layer(w) \le i$; such vertex must exist because $\layer(r) = 1$.
Furthermore, let $Q$ be the subpath of $P$ from $w$ to $v$.
It has at least one internal vertex ($u$) and no internal vertex of $Q$ belongs to $L_{\le i}$.
Since $V(T_{\le i}) \sub L_{\le i}$, the path $Q$ is edge-disjoint from $T_{\le i}$.
As noted above, $T_{\le i}$ is connected and also $v,w \in V(T_{\le i})$.
This means that $v,w$ are connected by two edge-disjoint paths in $T$, which is impossible as $T$ is a tree; a contradiction.

Finally, condition \eqref{item:layers:dipath} follows from \eqref{item:layers:monotone} because when  $\layer(v) = i = \layer(u)$ then $T[u \to v]$ is an ancestor-descendant path in $T_i$.
Depending on the parity of $i$, we infer that either $T[u \to v]$ is a dipath in $G$ or  $T[v \to u]$ is.
\end{proof}

The next ingredient is the following variant of the Lipton-Tarjan Separator Theorem~\cite{Lipton79}.

\begin{lemma}[{\cite[Lemma 2.3]{thorup2004compact}}]
\label{lem:faithful:balanced}
In linear time, given an undirected planar graph $H$ with a rooted
spanning tree $T$ and non-negative vertex weights, we can find three ancestor-descendant paths $P_1,P_2,P_3$ in $T$ such
 that each component of $H - (V(P_1) \cup V(P_2) \cup V(P_3))$ has at most half
the weight of $H$.
\end{lemma}

Let $Q$ be a dipath in a digraph $H$ and $v \in V(H)$, let $\first_H(Q,v)$ denote the first vertex on $Q$ that is reachable from $v$ in $H$.
Next, we define $\last_H(Q,v)$ as the last vertex on $Q$ from which $v$ is reachable in $H$.
If such a vertex on $Q$ does not exist, we set $\first_H(Q,v) = \bot$ or $\last_H(Q,v) = \bot$.
We drop the subscript when $H$ is clear from the context.

\begin{lemma}[{\cite[Lemma 2.5]{thorup2004compact}}]
\label{lem:faithful:dipath}
For a digraph $H$ and a dipath $Q$ in $H$, we can compute  $\first(Q,v)$ and $\last(Q,v)$ for all $v \in V(H)$ in linear time.
\end{lemma}

We are ready to prove the main result of this section.

\begin{lemma}\label{lem:faithful:construct}
    There is an algorithm that, given an $n$-vertex planar digraph $G$ and vertex set $X \sub V(G)$ of size $k$, runs in time $\Oh(n\log k)$ and outputs a \fre $H$ for $(G,X)$ satisfying $|V(H)| = \Oh(k \log k)$ and $|E(H)| = \Oh(k \log^2 k)$.
\end{lemma}
\begin{proof}
    We will recursively construct a vertex set $Y \sub V(G)$, $Y \supseteq X$, and a graph $H$ with $V(H) = Y$ that is a \re for $(G,Y,X)$.
    This will entail that $H$ is faithful.

    The base case occurs when $|X| \le 1$.
    If $X = \emptyset$ we return empty $Y,H$.  If $|X| = 1$ we return $Y=X$ and $H$ with no edges.

    We can assume that $G$ is weakly connected by treating each component separately.
    We apply \Cref{lem:faithful:layers} to compute, in linear time, the layering function $\layer \colon V(G) \to \nn$ and the spanning tree $T$.
    Next, we apply \Cref{lem:faithful:balanced} to tree $T$ and weight function $w \colon V(G) \to \{0,1\}$ assigning 1 to the vertices from $X$ and 0 to $V(G) \sm X$.
    This allows us to compute, in linear time, three ancestor-descendant paths $P_1,P_2,P_3$ in $T$, such that $S = V(P_1) \cup V(P_2) \cup V(P_3)$ forms a balanced separator in $G$ with respect to $w$.
    More precisely, each weakly connected component $C$ of $G-S$ satisfies $|X \cap V(C)| \le |X|/2$.
    We can assume that the paths $P_1,P_2,P_3$ are vertex-disjoint by removing their common prefixes.
    For each  component $C$ of $G-S$ we process the pair $(C, \,X \cap V(C))$ recursively to compute $Y_C \sub V(C)$ and a graph $H_C$, so that the pair $(Y_C,H_C)$ obeys the specification given at the beginning.
    That is, $X\cap V(C)\subseteq Y_C\subseteq V(C)$, and $H_C$ is a reachability emulator for $(C,Y_C,X\cap V(C))$.

    For $j \in [n]$ let $L_j = \layer^{-1}(j)$ be the set of vertices in the $j$-th layer.
    To simplify the formulas, we also define $L_{0} = \emptyset$.
    For each $i \in \{1,2,3\}$ and $j \in [n]$ we define dipath $P_{i}^j$ as the subpath of $P_i$ induced by~$L_j$.
    The condition \eqref{item:layers:dipath} of \Cref{lem:faithful:layers} guarantee that each $P_i^j$ is indeed a dipath in $G$.
    Of course, for some $j$ such a dipath may be empty.
    Note that the sets $V(P_i^j)$ are pairwise disjoint and their union is~$S$.

    \begin{claim}\label{claim:construct:diff}
        Suppose that for $v \in V(G)$ we have $\first_G(P_i^j,v) \ne \bot$ or $\last_G(P_i^j,v) \ne \bot$.
        Then $\layer(v) \in [j-1,j+1]$ and the corresponding $(v,P_i^j)$-dipath (resp.  $(P_i^j,v)$-dipath) is contained in $G[L_{j-1} \cup L_j \cup L_{j+1}]$.
    \end{claim}
    \begin{innerproof}
        Suppose that $\first_G(P_i^j,v) \ne \bot$.
        Let $u = \first_G(P_i^j,v)$ and $Q$ be any $(v,u)$-dipath in $G$.
        Note that $\layer(u) = j$.
        \Cref{lem:faithful:layers}\eqref{item:layers:diff} guarantees that $|j - \layer(v)| \le 1$, hence $\layer(v) \in [j-1,j+1]$.
        Also, $u$ is reachable from every vertex $z$ on $Q$, so by the same argument $\layer(z) \in [j-1,j+1]$.
        This implies that $V(Q) \sub L_{j-1} \cup L_j \cup L_{j+1}$.
        The proof for $\last_G(P_i^j,v) \ne \bot$ is symmetric.
    \end{innerproof}

    We want to compute all the values $\first_G(P_i^j,v)$ and $\last_G(P_i^j,v)$.
    It follows from \Cref{claim:construct:diff} that for each $v \in V(G)$ there are 9 pairs $(i,j)$ for which it may hold $\first_G(P_i^j,v) \ne \bot$ or $\last_G(P_i^j,v) \ne \bot$.
    Therefore, it suffices to consider only these 9 pairs and the value will be $\bot$ for the remaining pairs $(i,j)$.

    \begin{claim}\label{claim:construct:compute}
        One can compute $\first_G(P_i^j,v)$ and $\last_G(P_i^j,v)$ for all $v \in V(G)$ and all $(i,j) \in [3] \times [n]$ in total linear time.
    \end{claim}
    \begin{innerproof}
        On default we evaluate $\first_G(P_i^j,v)$ and $\last_G(P_i^j,v)$ to $\bot$.
        By \Cref{claim:construct:diff} we know that a non-$\bot$ value is possible only for $v \in L_{j-1} \cup L_j \cup L_{j+1}$, so for each $v$ we only need to store $\Oh(1)$ values.

        For each $(i,j) \in [3] \times [n]$ we apply \Cref{lem:faithful:dipath} to the graph $H_j = G[L_{j-1} \cup L_j \cup L_{j+1}]$ and dipath $P_i^j$.
        This takes time linear in the size of $H_j$.
        By \Cref{claim:construct:diff} the values of $\first(P_i^j,v)$ and $\last(P_i^j,v)$ coincide in $G$ and $H_j$ for all $v\in V(H_j)$.
        That is $\first_G(P_i^j,v) = \first_{H_j}(P_i^j,v)$ and $\last_G(P_i^j,v) = \last_{H_j}(P_i^j,v)$.
        Therefore, the values computed within $H_j$ are valid for $G$.
        Since every vertex and edge of $G$ contributes to at most three graphs $H_j$, the total time is linear in the size of $G$.
    \end{innerproof}

    From this point we drop the subscript $G$ at $\first$ and $\last$ since we will consider these only within~$G$.
     We call a vertex $u \in V(P_i^j)$ \coldef{relevant} if
     $u = \last(P_i^j,v)$ for some $v \in X$.
     Let $R \sub S$ denote the set of relevant vertices.
     \Cref{claim:construct:diff} implies that $|R| \le 9\cdot k$.
    For  $(i,j) \in [3] \times [n]$ we define dipath $R_i^j$ on vertex set $V(P_i^j) \cap R$ that connects the relevant vertices on $P_i^j$ in the same order as they lie on $P_i^j$.
    Next, we define function $\rel \colon S \cup \bot \to R\cup \bot$.
    For $v \in V(P_i^j)$ we set $\rel(v)$ as the first relevant vertex on $P_i^j$ after $v$, inclusively.
    If such a vertex does not exist, we set $\rel(v) = \bot$.
    Also $\rel(\bot) = \bot$.
    All the values of $\rel(\cdot)$ can be readily computed in total linear time.

    Let $\ccal$ be the set of weakly connected components of $G-S$.
    The set $Y$ is defined as the union of $X$, $R$, and $Y_C$ over all $C \in \ccal$, computed in the recursive calls.
    The digraph $H$ on the vertex set $Y$ is constructed as follows.
    \begin{enumerate}
        \item Take the union of the digraphs $H_C$ for all $C \in \ccal$.
        \item Take the union of the dipaths $R_i^j$ for all $(i,j) \in [3] \times [n]$.
        \item For each  $v \in X$ and  $(i,j) \in [3] \times [n]$ insert an edge from $\last(P_i^j,v)$ to $v$ as long as $\last(P_i^j,v) \ne \bot$.
        \item For each $v \in Y$ and  $(i,j) \in [3] \times [n]$ insert an edge from $v$ to $\rel(\first(P_i^j,v))$ as long as the latter is not $\bot$.
    \end{enumerate}

    \begin{figure}[t]
    \centering
    \includegraphics[width=0.45\textwidth]{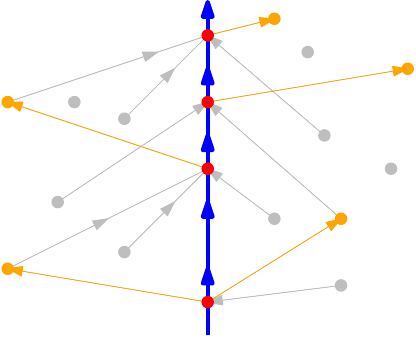}
    \caption{Illustration to the construction of $H$, focused on a single path $P_i^j$ (blue) from the separator $S$.
    The orange vertices are in $X$, the red ones are in $R$, and the gray ones are in $Y \sm (X \cup R)$.
    For each $v \in X$ we insert an edge (orange) to $v$ from the last vertex on $P_i^j$ that reaches $v$ (if it exists).
    Such vertices on $P_i^j$ are {\em relevant} and they are included in $R$.
    Next, for each $v \in Y$ we insert an edge (gray) from $v$ to the first vertex from $R$ that is reachable from $v$ (if it exists).
    This way, we preserve all the connections from any $u \in Y$ to any $v \in X$ that pass through  $P_i^j$.
    }
    \label{fig:thorup-sep}
\end{figure}

    See \Cref{fig:thorup-sep} for an illustration. Note that in the two last steps, for each vertex $v$ we only need to consider $\Oh(1)$ pairs $(i,j)$ (due to \Cref{claim:construct:diff}) so in total these steps take linear time.

    The last step is novel compared to the scheme in~\cite{thorup2004compact} and allows us to preserve $\reach_G(v,X)$ also for $v \in Y \sm X$.
    Note that we do not aim to preserve $\reach_G(v,Y)$
    and so we do not need to add to $R$ relevant vertices corresponding to vertices in $Y \sm X$.
    This is the crucial argument that allows us to keep the size of $H$ near-linear in $k$.

    \begin{claim}
        $H$ is a \re for $(G,Y,X)$.
    \end{claim}
    \begin{innerproof}
        We need to show that for each $u \in Y$ and $v \in X$ it holds that: $u \overset{G}\leadsto v$ if and only if $u \overset{H}\leadsto v$.

        We begin with the first implication $(\Rightarrow)$.
        Let $P$ be a $(u,v)$-path in $G$.
        If $V(P) \cap S = \emptyset$ then $P$ is entirely contained in some $C \in \ccal$.
        Also, $u \in Y_C$ and $v \in V(C) \cap X$.
        By induction, $H_C$ is a \re for $(C,Y_C,V(C) \cap X)$ so it contains a $(u,v)$-path.
        In the second case $V(P) \cap S \ne \emptyset$. Then $V(P) \cap V(P_i^j) \ne \emptyset$ for some $(i,j) \in [3] \times [n]$.
        This implies that $u' := \first(P_i^j,u) \ne \bot$, $v' := \last(P_i^j,v) \ne \bot$, and $v'$ is reachable from $u'$ in $P_i^j$.
        Next, we know that $v'$ is relevant, so $\rel(u') \ne \bot$ and $v'$ is reachable from $\rel(u')$ in $R_i^j$, which is a subgraph of $H$.
        Finally, the edges $(u,\rel(u'))$ and $(v',v)$ are included in $E(H)$, so $H$ contains a $(u,v)$-path.

        Now we show that $u \overset{H}\leadsto v$ implies $u \overset{G}\leadsto v$.
        We will prove by induction that if $(u,v) \in E(H)$ then $u \overset{G}\leadsto v$, which implies the statement above.
        We consider four types of edges corresponding to the four steps in the construction of $H$.
        \begin{enumerate}
            \item $(u,v) \in E(H_C)$ for some $C \in \ccal$. By induction we get that $G_C$ contains a $(u,v)$-path, hence $G$ as well.
            \item $(u,v) \in E(R_i^j)$ for some $(i,j) \in [3] \times [n]$. Then $P_i^j$ contains a $(u,v)$-path and $P_i^j$ is a dipath in $G$.
            \item $u = \last(P_i^j, v)$ for some $(i,j) \in [3] \times [n]$. Then $u \overset{G}\leadsto v$ by the definition of $\last(\cdot,\cdot)$.
            \item $v = \rel(\first(P_i^j, u))$ for some $(i,j) \in [3] \times [n]$. Let $u' = \first(P_i^j, u)$.
            Then $u \overset{G}\leadsto u'$ by the definition of $\first(\cdot,\cdot)$ and
             $u' \overset{G}\leadsto v$ because $\rel(u')$ is reachable from $u'$ in $P_i^j$.
        \end{enumerate}
    The claim follows.
    \end{innerproof}

    \begin{claim}
        $H$ is a \fre for $(G,X)$.
    \end{claim}
    \begin{innerproof}
        We need to show that for every $v \in V(H)$ there exists $u \in V(G)$ so that $\reach_H(v,X) = \reach_G(u,X)$.
        But $V(H) = Y$ and $H$ is a \re for $(G,Y,X)$.
        This means that for every $v \in Y$ the sets $\reach_H(v,X),\, \reach_G(v,X)$ coincide,
        so $H$ is indeed faithful.
    \end{innerproof}

    \begin{claim}\label{claim:construct:vertex-count}
        $|V(H)| = \Oh(k\log k)$.
    \end{claim}
    \begin{innerproof}
        The set $V(H) = Y$ is the union of all sets $Y_C$, collected from the recursive calls, together with $X$ and $R$, where $|R| \le 9\cdot |X|$.
        Therefore, $|Y|$ is proportional to the sum of sizes of $X_t$ over all recursive calls $(G_t,X_t)$.
        Since $|X|$ halves in each recursive call and the recursion stops when $|X| \le 1$, there are $\log k$ levels of recursion.
        At each level the corresponding sets $X_t$ are disjoint and so their total size is at most~$k$.
        The claim follows.
    \end{innerproof}

    \begin{claim}
        $|E(H)| = \Oh(k \log^2 k)$.
    \end{claim}
    \begin{innerproof}
        We have four types of edges corresponding to the four steps in the construction of $H$.
        First, we take all the edges in the digraphs constructed in the recursive calls.
        Second, we take the union of vertex-disjoint dipaths  $R_i^j$, which give at most $9\cdot |X|$ edges.
        In the third step, we insert at most 9 edges for each $v \in X$ due to \Cref{claim:construct:diff}.
        The fourth step is most complicated as we insert at most 9 edges for each $v \in Y$, where $Y$ accumulates the newly created vertices in $R$ and all vertices created in the recursive calls, together with $X$.
        Since $|X| \le |Y|$, we can upper bound the number of newly created edges as $\Oh(|Y|)$.
        Therefore, the total number of edges is proportional to the sum of sizes of sets $Y_t$ constructed in all recursive calls $(G_t,X_t)$.
        Applying \Cref{claim:construct:vertex-count} to each recursive call yields $|Y_t| \le c \cdot |X_t| \log |X_t| \le c \cdot |X_t| \log k$ for some constant $c$.
        Furthermore, the previous claim shows that the sum of sizes of $|X_t|$ over all calls is $\Oh(k \log k)$.
        Hence, the total number of created edges is $\Oh(k \log^2 k)$.
    \end{innerproof}

    \begin{claim}
        The algorithm runs in time $\Oh(n \log k)$.
    \end{claim}
    \begin{innerproof}
        The total time used outside the recursive calls is linear in the size of the processed graph
        and
        there are $\log k$ levels of recursion.
        At each level, we process a family of vertex-disjoint subgraphs of $G$, so the total time spend on each level is~$\Oh(n)$.
    \end{innerproof}

This concludes the proof of \Cref{lem:faithful:construct}.
\end{proof}

\section{Reachability emulators in bounded VC dimension }
\label{sec:emul-vc}

In this section we give an algorithm
computing a \re for $(G,X)$ or $(G,Y,X)$,
faster than $\Oh(|G|\cdot |X|)$ whenever the VC dimension of $\reachsys(G,X)$ is bounded and $X$ is small compared to~$|G|$.
Our algorithm relies on \Cref{thm:prelim:stab} and represents each set $\reach_G(v,X) \sub X$ using only $|X|^{1-\Omega(1)}$ intervals.
To make the presentation cleaner, we encapsulate our usage of \Cref{thm:prelim:stab} inside a data structure
processing sets via their representations of sublinear size.

Let $\pi = (x_1,\dots, x_k)$ be an ordering of a set $X$.
An \coldef{interval representation} of a set $A \sub X$ with respect to $\pi$ is the minimum-size sequence of disjoint intervals $I_1,\dots,I_\ell \sub [n]$ such that $\bigcup_{i=1}^\ell \{x_j \colon j \in I_i\} = A$.
We denote the sorted list $(I_1,\dots,I_\ell)$ by $\pi(A)$.
Note that $\pi(A)$ is uniquely determined for any $\pi$ and $A$.

\begin{definition}
    The \coldef{Interval Representation Data Structure} $\dcal$ maintains a sequence $(F_1,\dots,F_m)$ of subsets of a ground set $X$ and
    an ordering $\pi = (x_1,\dots, x_k)$ of $X$ in the following way.
    \begin{enumerate}
        \item (Retrieval) Given an integer $p \in [m]$, $\dcal$ outputs $\pi(F_p)$.
        \item (Set lookup) Given an interval representation $I$ over $\pi$, $\dcal$ determines whether there is $p \in [m]$ such that $\pi(F_p) = I$. If yes,  $\dcal$ outputs $p$.
        \item (Element lookup) Given an element $x \in X$, $\dcal$ outputs the index of $x$ in the ordering $\pi$.
        \item (Insertion) Given an interval representation $I$ over $\pi$, describing a set $A \sub X$ that does not belong to~$\{F_1,\dots,F_m\}$, $\dcal$
        appends $A$ at the end of the sequence.
        \item (Consistency) The retrieval and lookup operations do not modify the data structure. The insertion operation does not affect the ordering of the sets already present in the sequence $(F_1,\dots,F_m)$.
    \end{enumerate}
\end{definition}

We explain how to implement this data structure to process retrieval and lookup operation in sublinear time whenever the maintained set family has bounded VC dimension.
This comes at a cost of processing insertions slower but we will later argue that the total number of insertions cannot be too large.
A~bound~$d$ on the VC dimension implies a bound of $2^{d+1}$ on the dual VC dimension but in our applications we can estimate the latter in a better way (see \Cref{thm:prelim:dual-vc}).
It is thus convenient to assume that both measures are bounded by $d$ in the statement below.

\begin{lemma}\label{lem:vsbased:implement}
    Let $d \ge 2$ be a constant and $k$ denote the size of the ground set $X$.
    There is an implementation of the {Interval Representation Data Structure} that works under the assumption that the maintained family $\fcal = \{F_1,\dots,F_m\}$ has both VC dimension and dual VC dimension at most $d$, and provides the following guarantees.
    \begin{enumerate}
        \item The ordering $\pi$ always has stabbing number $\Oh_d(k^{1-1/d}\cdot \log k)$ with respect to $\fcal$.
        \item The retrieval operation takes time $\Oh_d(k^{1-1/d}\cdot \log k)$.
        \item The set lookup operation takes time $\Oh_d(k^{1-1/d}\cdot \log^2 k )$.
        \item The element lookup operation takes time $\Oh(1)$.
        \item The insertion operation takes time  $\Oh_d(k^{d+3})$.
    \end{enumerate}
\end{lemma}
\begin{proof}
    By \Cref{lem:prelim:size-bound} the family $\{F_1,\dots,F_m\}$ may contain at most $\Oh(k^d)$ sets.
    Since we only insert sets that do not yet appear in the sequence, the sets remain pairwise distinct and so $m$ cannot exceed $\Oh(k^d)$.

    Consider an insertion operation and let $\fcal$ refer to the family of sets present in the sequence after insertion.
    At this point we recompute the entire data structure while maintaining the ordering of $(F_1,\dots,F_m)$.
    By \Cref{thm:prelim:stab} there is $c_d \in \nn$ so that
    there exists an ordering $\pi$ of $X$ with stabbing number at most $c_d \cdot k^{1-1/d}\log k$ with respect to $\fcal$.
    Furthermore, we can compute $\pi$ in time $\Oh_d(k^3(k + m)) \le \Oh_d(k^{d+3})$.
    We store for each $i \in [m]$ the new interval representation $\pi(F_i)$
    and for each $x \in X$ we store the index of $x$ in $\pi$.
    Next, we build a sequence of pairs $(\pi(F_i),i)_{i=1}^m$ and we sort it using lexicographic ordering on the first coordinate.
    Recall that $\pi(F_i)$ is a list of pairs from $[k]$, each describing one interval.
    Sorting requires $\Oh(m \log m)$ comparisons, each taking time $\Oh_d(k^{1-1/d}\log k)$.
    This readily fits in the stated running time of insertion.

    The retrieval operation is trivial as we explicitly store each $\pi(F_i)$.
    It takes time proportional to the size of  $\pi(F_i)$, which is $\Oh_d(k^{1-1/d}\log k)$.
    The element lookup is also trivial because the answer is explicitly stored.

    Consider the set lookup operation.
    We are given an interval representation $I = \pi(A)$ of some set $A \sub X$.
    If $|I| > c_d \cdot k^{1-1/d}\log k$, we can already answer ``no'' because $A$ apparently has not been taken into account when $\pi$ was computed.
    Otherwise, we look for $I$ in the sorted sequence $(\pi(F_i),i)_{i=1}^m$.
    This requires $\Oh(\log m) \le \Oh(d \log k)$ comparisons, each taking time $\Oh_d(k^{1-1/d}\log k)$.
    Recall that $\pi(A)$ is uniquely determined so if $A$ is present in the sequence then we will find pair $(\pi(A),i)$ with some $i \in [m]$.
    Then $A = F_i$ and we output $i$.
    If the search failed, we answer that $A$ is not present.

    The consistency property follows directly from the description of operations.
\end{proof}

The following simple lemma allows us to quickly compute a union of interval representations.

\begin{lemma}\label{lem:vsbased:union}
    Given interval representations of $t\ge 1$ sets $A_1,\dots, A_t \sub X$ with respect to some ordering $\pi$ of $X$, we can compute the interval representation of $\bigcup_{i=1}^t A_i$, with respect to $\pi$, in time
    {$\Oh(s  \log s)$} where $s$ is the total number of input intervals.
\end{lemma}
\begin{proof}
    We say that two disjoint intervals $[a,b]$, $[c,d]$ touch, if $c=b+1$ or $a=d+1$.

    Let $\{[a_i,b_i]\, \colon\, i \in [s]\}$ be the collection of all intervals appearing in the interval representations of  $A_1,\dots, A_t$.
    Our goal is to represent $\bigcup_{i=1}^s [a_i,b_i]$ as a union of disjoint non-touching intervals.

    We sort the input intervals by the first coordinate in time $\Oh(s  \log s)$.
    We will maintain an interval $[a',b']$ representing the union of processed intervals.
    Initially $[a',b'] = [a_1,b_1]$.
    For each $i \ge 2$ we check if $a_i \le b' + 1$.
    If yes, we extend $[a',b']$ by setting $b' \leftarrow \max(b', b_i)$.
    Note that $a' \le a_i$ so we do not have to extend $[a',b']$ to the left.
    If  $a_i > b' + 1$ then $[a',b']$ is disjoint from the remaining intervals and it does not touch any of them.
    In this case, we output the current $[a',b']$ and continue the process with $[a',b'] = [a_i,b_i]$.
    After processing $i = s$ we also output the final $[a',b']$.
\end{proof}

Now we present an algorithm computing all the reachability sets $\reach_G(v,X)$ for $v \in V(G)$.
This algorithm is faster than $\Oh(|G|\cdot |X|)$ whenever the VC dimension of $\reachsys(G,X)$ is bounded and $X$ is sufficiently small.
To estimate the gain,
we assume a bound $d$ on both VC dimension and dual VC dimension of $\reachsys(G,X)$.
Note that a priori these two measures may not be equal, unlike in \Cref{thm:prelim:dual-vc},
and the bound on both will later follow from \Cref{obs:prelim:faithful-reach}.

We first give an algorithm for acyclic digraphs.

\newcommand{\rep}{{\color{\mycol}{\mathrm{rep}}}}
\newcommand{\scc}{{\color{\mycol}{\mathrm{scc}}}}
\newcommand{\ind}{{\color{\mycol}{\mathrm{index}}}}

\begin{lemma}\label{lem:vsbased:dag}
        Let $d \ge 2$ be a constant.
    There is an algorithm with the following specification.
    \begin{itemize}
    \item {Input:} acyclic digraph $G$, set $X \sub V(G)$, such that $\reachsys(G,X)$ has VC dimension and dual VC dimension at most $d$

    \item Output: ordering $\pi$ of $X$ with stabbing number $\widetilde\Oh_d(|X|^{1-1/d})$ with respect to $\reachsys(G,X)$ and interval representation of $\reach_G(v,X)$, for each $v \in V(G)$, with respect to $\pi$

    \item Running time: $\widetilde\Oh_d(|G|\cdot |X|^{1-1/d} + |X|^{2d+3})$
    \end{itemize}
\end{lemma}
\begin{proof}
    Consider a reverse topological ordering $(v_1,\dots, v_n)$ of $V(G)$, that is,  $v_j \overset{G}\leadsto v_i$ implies $i < j$.
    Let $R_i$ stand for $\reach_G(v_i,X)$.
    We will utilize the {Interval Representation Data Structure} to compute $R_i$ one by one.
    Since $\reachsys(G,X)$ has bounded VC dimension and dual VC dimension, we can use the implementation from \Cref{lem:vsbased:implement}.

    \paragraph{Invariant.}
    At the beginning of the $j$-th step we will keep the following invariant.
    The {Interval Representation Data Structure} $\dcal$ stores a sequence comprising subsets $\{R_i \colon i < j\}$ of the ground set $X$ and for each $i < j$ we have computed integer $\ind[i]$ equal to the position of the set $R_i$ in the sequence.

    \paragraph{Algorithm.}
    Let $j \in [n]$; we now describe the $j$-th step of the algorithm.
    Let $N_j$ be the set of indices $i$ for which $(v_j,v_i) \in E(G)$.
    We know that $N_j \sub [j-1]$ due to topological ordering.
    Our goal is to compute the interval representation of the set $R_j$.
    This set is given by the union of $R_i$, over $i \in N_j$, together with $\{v_j\}$ whenever $v_j \in X$.
    For each $i \in N_j$ we perform the retrieval operation on $\dcal$ for $p = \ind[i]$  to read the interval representation of $R_i$ with respect to the ordering of $X$ currently stored in $\dcal$.
    Additionally, if $v_j \in X$ we perform the element lookup operation to obtain  the (singleton) interval representation of $\{v_j\}$.
    Then we compute the interval representation $I$ of
    the union of these sets using \Cref{lem:vsbased:union}, which yields a representation of $R_j$ (see \Cref{fig:interval} on page \pageref{fig:interval}).
    Next, we perform the set lookup operation for $I$ to check if $R_j$ is already present in $\dcal$.
    If yes, $\dcal$ outputs the index $p$ of $R_j$ on the list; we set $\ind[j] = p$.
    Otherwise, we perform the insertion operation for $I$.
    Since $R_j$ is appended at the end of the sequence in $\dcal$, we
     write at $\ind[j]$ the current length of the sequence.

    This process terminates after the $n$-th step.
    After that we read the stored ordering $\pi$ of $X$ by
    performing the element lookup operation on $\dcal$ for each $v_i\in X$.
    Finally, we read
    the interval representation of each $R_i$, $i \in [n]$, by
    performing the retrieval operation  with $p = \ind[i]$.

    \paragraph{Correctness.}
    The invariant is clearly satisfied for $j = 1$ and for larger $j$ we proceed via induction.
    Let $\pi$ be the ordering of $X$ stored in $\dcal$ at the beginning of the $j$-th step.
    The invariant implies that for each $i \in N_j$ we read the correct interval representation of the set $R_i$ with respect to $\pi$.
    Then $I$ is the interval representation of $R_j$ with respect to $\pi$.
    If $R_j$ is not already present in $\dcal$ then it is inserted
    at the end of the $j$-th step.
    Due to the consistency property of the {Interval Representation Data Structure} every operation preserves the validity of the indices stored in the array $\ind$.
    Hence the invariant is maintained.

    \paragraph{Running time.}
    Let $k = |X|$.
    The set lookup operation is performed once for each vertex in $G$.
    A single invocation takes time $\widetilde\Oh_d(k^{1-1/d})$, as specified in \Cref{lem:vsbased:implement}.
    The retrieval operation is performed once for each edge in $G$ and then $n$ times to read the output.
    \Cref{lem:vsbased:implement} guarantees that each call takes time $\widetilde\Oh_d(k^{1-1/d})$.
    The element lookup takes constant time and is performed at most twice per vertex.
    All these calls together yield time $\widetilde\Oh_d(|G|\cdot k^{1-1/d})$.

    For each $j \in [n]$ we invoke \Cref{lem:vsbased:union} to compute the union of at most $|N_j|+1$ interval representations, each with $\widetilde\Oh_d(k^{1-1/d})$ intervals.
    This takes time near-linear in the total number of intervals.
    Summing over all $j \in [n]$, this gives $\sum_{j=1}^n \Ot_d\left((|N_j|+1)\cdot (k^{1-1/d})\right) = \widetilde\Oh_d(|G|\cdot k^{1-1/d})$.

    The insertion operation can occur at most $\Oh(k^d)$ times because this is the maximum number of sets in $\reachsys(G,X)$ due to \Cref{lem:prelim:size-bound}.
    Each insertion takes time $\Oh_d(k^{d+3})$ by \Cref{lem:vsbased:implement}.
    In total we get  $\Oh_d(k^{2d+3})$.
\end{proof}

To turn this reachability representation into an emulator of small size, we employ the well-known \coldef{segment tree} technique attributed to Bentley~\cite{bentley1977algorithms}.
    \begin{lemma}[{e.g.,~\cite[Section~1.2.3.1]{PreparataS85}}]\label{lem:segment-tree}
    Let $L=\{v_1,\ldots,v_{\ell}\}$ be any set.
    In $\Oh(\ell)$ time one can build an $\Oh(\ell)$-node rooted binary tree $T$ with leaves $L$ and the following property.
    For any $1\leq a\leq b\leq \ell$, there exists a subset $T_{[a,b]}\subseteq V(T)$ of $\Oh(\log\ell)$ nodes of $T$ such that the leaf-descendants
    of the nodes in $T_{[a,b]}$ are precisely $\{v_a,\ldots,v_b\}$.
    Moreover, for any $a,b$, $T_{[a,b]}$ can be computed in $\Oh(\log{\ell})$ time.
    \end{lemma}

We are ready to present the main lemma of this section.

\begin{lemma}\label{lem:vsbased:general}
    Let $d \ge 2$ be a constant.
    There is an algorithm with the following specification.
    \begin{itemize}
    \item {Input:} digraph $G$, sets $Y, X \sub V(G)$, such that $\reachsys(G,X)$ has VC dimension and dual VC dimension at most~$d$

    \item Output: \re for $(G,Y,X)$ of size $\widetilde\Oh_d(|Y| \cdot |X|^{1-1/d} + |X|)$

    \item Running time: $\widetilde\Oh_d(|G|\cdot |X|^{1-1/d} + |X|^{2d+3})$
    \end{itemize}
\end{lemma}
\begin{proof}
    Let $D$ be a DAG obtained from $G$ by contracting each strongly connected component (SCC) to a single vertex.
    We can compute it in linear time, together with a mapping $\scc \colon V(D) \to 2^{V(G)}$ representing the SCCs of $G$.
    Let $X_D$ be the set of those vertices $v \in V(D)$ such that $\scc(v) \cap X \ne \emptyset$.
    Clearly, $|X_D| \le |X|$ since $\scc$ represents a partition of $V(G)$.
    Let $\rep \colon X_D \to X$ be a mapping that assigns to $v \in X_D$ an arbitrary element from $\scc(v) \cap X$.

    \begin{observation}\label{obs:vsbased:scc}
        Let $v \in V(D)$, $u \in \scc(v)$, and $x \in X_D$.
        Then $v \overset{D}\leadsto x$ if and only if $u \overset{G}\leadsto \rep(x)$.
    \end{observation}

    \begin{claim}
        The VC dimension and the dual VC dimension of $\reachsys(D,X_D)$ is at most $d$.
    \end{claim}
    \begin{innerproof}
        The (dual) VC dimension of $\reachsys(G,\rep(X_D))$ is at most $d$ since $\rep(X_D) \sub X$ (see \Cref{lem:prelim:truncate}).
        We claim that for every $A \in \reachsys(D,X_D)$ the set $\rep(A)$ belongs to \linebreak $\reachsys(G, \rep(X_D))$.
        To see this, consider $v \in V(D)$ such that $\reach_D(v,X_D) = A$.
        Then for any $u \in \scc(v)$ we have that $\reach_G(u,\rep(X_D)) = \rep(A)$ due to \Cref{obs:vsbased:scc}.
        The mapping $\rep$ is thus a bijection between the ground sets that transforms $\reachsys(D,X_D)$ into a set system contained in $\reachsys(G,\rep(X_D))$.
        The claim follows.
    \end{innerproof}

    Thanks to the above claim, we can
     execute the algorithm from \Cref{lem:vsbased:dag} for $D$ and $X_D$ within the desired running time bound.
    We obtain an ordering $\pi$ of $X_D$ with stabbing number $\Ot_d(|X_D|^{1-1/d})$ and a corresponding interval representation of each  set $\reach_D(v,X_D)$, $v \in V(D)$.

    Let us now construct the sought reachability emulator $H$ for $(G,Y,X)$.
    Let $T$ be the segment tree of~\Cref{lem:segment-tree} with leaves
    $X_D$ ordered by $\pi$ (that is, $X_D\subseteq V(T)$).
    Here, we assume $V(T)$, and thus also $X_D$, to be disjoint from $X\cup Y$.
    The vertex set of the emulator $H$ is $V(T)\cup X\cup Y$.
    The edges of $H$ include:
    \begin{itemize}
        \item The edges of the tree $T$ directed from parents to children.
        \item For each $x\in X$, an edge $(x_D,x)$ from the unique vertex $x_D\in X_D$ such that $x\in \scc(x_D)$.
        \item For each $y\in Y$, $\Ot_d(|X|^{1-1/d})$ edges to certain nodes of $T$. Specifically, if $y\in \scc(y_D)$, then for each interval $I=[a,b]$ in the interval representation of $\reach_D(y_D,X_D)$, we have an edge $(y,z)$ for each $z\in T_{[a,b]}$, where $T_{[a,b]}$ is defined as in~\Cref{lem:segment-tree}.
    \end{itemize}
    An example is given on \Cref{fig:segment-tree}. By construction, the emulator $H$ has $\Oh_d(|X|+|Y|)$ vertices and $\widetilde\Oh_d(|Y|\cdot |X|^{1-1/d}+|X|)$ edges.
    Moreover, it can be constructed in time $\widetilde\Oh_d(|Y|\cdot |X|^{1-1/d}+|X|)$.

\begin{figure}[h]
    \centering
    \includegraphics[width=0.7\textwidth]{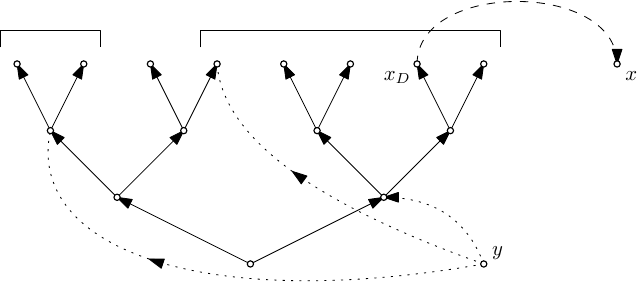}
    \caption{The construction of the \re $H$.
    The vertices of $D$ are the leaves of the tree $T$, ordered according to $\pi$.
    For each $x \in X$ there is an edge from $x_D$ to $x$.
    For each $y \in Y$  we insert edges from $y$ to $t_i \in V(T)$ so that the combined reachability of the vertices $t_i$ in $T$ corresponds to $\reach_D(y_D,X_D)$.
    In this example, the interval representation of  $\reach_D(y_D,X_D)$ is $([1,2],\, [4,8])$.
    }
    \label{fig:segment-tree}
\end{figure}

    \begin{claim}
    The constructed digraph $H$ is indeed a reachability emulator for $(G,Y,X)$.
    \end{claim}
    \begin{innerproof}
        Let $y\in Y$ and $x\in X$. Let $y_D,x_D$ be the vertices representing the SCCs of $y$ and $x$ in $D$, respectively.
        By the definition of $D$, $x\in \reach_G(y,X)$ is equivalent to $x_D\in \reach_D(y_D,X_D)$.

        First suppose that $x\in \reach_G(y,X)$.
        Since $(x_D,x)\in E(H)$, it is enough to argue that\linebreak $x_D\in \reach_H(y,X_D)$.
        Recall that the tree $T$ of~\Cref{lem:segment-tree} with $L=X_D$ ordered by $\pi$ is included in $H$.
        Next, $y$ has a direct outgoing edge to some $v\in V(T)$ such that $x_D$ is a leaf descendant of $v$ in $T$.
        As the edges of~$T$ are directed downwards, there exists a path $y\to v\to x_D$ in $H$, as desired.

        Now assume that $x\in \reach_H(y,X)$.
        Consider any $y\to x$ path $P\subseteq H$.
        Split $P$ into minimal non-empty paths $P_1=y_1\to x_1,\ldots,P_k=y_k\to x_k$ with both endpoints in $X\cup Y$.
        Note that each $x_i$ has exactly one incoming edge in $H$, one from $(x_i)_D$, whereas $y_i$ has only direct edges to $V(T)$ in $H$.
        Moreover, $x_i\in X$.
        By construction, each  path $P_i$ is of the form $(y_i,z)\cdot (z\to_T (x_i)_D)\cdot ((x_i)_D,x_i)$, where $z\in T_{[a,b]}$ and $I=[a,b]$ is a part of the interval representation of $\reach_D((y_i)_D,X_D)$.
        Since the edges of $T$ are directed downwards in $H$, the only leaves of $T$ (or, in other words, the only vertices in $X_D$) that $z$ can reach are its leaf descendants in $T$.
        By the guarantee in~\Cref{lem:segment-tree}, all leaf descendants of $z$ belong to $I\subseteq \reach_D((y_i)_D,X_D)$ as well.
        We obtain that $(x_i)_D\in \reach_D((y_i)_D,X_D)$, and thus $x_i\in \reach_G(y_i,X)$ as well.
        For $i<k$, $y_{i+1}=x_{i}$, so we conclude that there is a path $y=y_1\to \ldots \to y_{k}\to x_k=x$ in $G$. $x\in \reach_G(y,X)$ follows, as desired.
    \end{innerproof}
This concludes the proof of \Cref{lem:vsbased:general}.
\end{proof}

Now we justify that the additive term $|X|^{2d+3}$ in the running time can be removed at the expense of a slight increase in the emulator size.
Note that $|X| \le |G|$ so for sufficiently small $\delta$ the first summand becomes dominating in the running time below.

\begin{lemma}\label{lem:faithful:optimize-general}
    Let $\delta \in (0,1]$ and $d \ge 2$ be constants.
    There is an algorithm with the following specification.
    \begin{itemize}
    \item {Input:} digraph $G$, sets $Y, X \sub V(G)$, such that $\reachsys(G,X)$ has VC dimension and dual VC dimension at most~$d$

    \item Output: \re for $(G,Y,X)$ of size $\widetilde\Oh_d(|Y| \cdot |X|^{1-\delta/d} + |X|)$

    \item Running time: $\widetilde\Oh_d(|G|\cdot |X|^{1-\delta/d} + |X|^{1 + (2d+2)\cdot\delta})$
    \end{itemize}
    \end{lemma}
\begin{proof}
    Let $k = |X|$ and $\ell = |Y|$.
    Partition $X$ into $q=\Oh(k^{1-\delta})$ blocks $X_1,\ldots,X_q$ of size at most $\lceil k^\delta \rceil$.
    For each $i=1,\ldots,q$, construct a reachability emulator $H_i$ for $(G,Y,X_i)$ using~\Cref{lem:vsbased:general}.
    The total size of all reachability emulators $H_i$ is
    \[\Ot\left(k^{1-\delta}\cdot \left(\ell\cdot (k^\delta)^{1-1/d}+k^\delta\right)\right)=\Ot\left(\ell\cdot k^{1-\delta/d}+k\right),\]
    whereas the total time required to compute them is
    \[\Ot\left(k^{1-\delta}\cdot \left(|G|\cdot (k^\delta)^{1-1/d}+(k^\delta)^{2d+3}\right)\right)=\Ot\left(|G|\cdot k^{1-\delta/d}+k^{1+(2d+2)\delta}\right).\]
    In the end one can obtain a reachability emulator for $(G,Y,X)$ from $H_1,\ldots,H_q$ using~\Cref{l:assemble-from-parts}.
    The resulting emulator has asymptotically the same size as $|H_1|+\ldots+|H_q|$.
\end{proof}

Finally, we formulate a convenient corollary to be used in both our algorithms.
When $H$ is a \fre for $(G,X)$ then the VC dimension and the dual VC dimension of  $\reachsys(H,X)$ are not larger than the ones in $\reachsys(G,X)$ due to \Cref{obs:prelim:faithful-reach}.
By \Cref{obs:prelim:faithful-subset}, the same upper bound applies to $\reachsys(H,Z)$ for any subset $Z \sub X$.
By applying \Cref{lem:faithful:optimize-general} to $H, Y, Z$ we obtain the following.

\begin{corollary}\label{lem:faithful:optimize}
    Let $\delta \in (0,1]$ and $d \ge 2$ be constants.
    There is an algorithm with the following specification.
    \begin{itemize}
    \item {Input:} digraph $H$, sets $Y, Z \sub X \sub V(H)$

    \item {Assumption:} there exists a digraph $G$ with $X \sub V(G)$, such that $\reachsys(G, X)$ has VC dimension and dual VC dimension at most~$d$, and $H$ is a \fre for $(G,X)$

    \item Output: \re for $(G,Y,Z)$ of size $\widetilde\Oh_d(|Y| \cdot |Z|^{1-\delta/d} + |Z|)$

    \item Running time: $\widetilde\Oh_d(|H|\cdot |Z|^{1-\delta/d} + |Z|^{1 + (2d+2)\cdot\delta})$
    \end{itemize}
    \end{corollary}

\section{Offline Dynamic Reachability}
\label{sec:offline}

In this section we prove \Cref{thm:main:offine}.
For the offline algorithm, reducing the additive term in the running time in \Cref{lem:faithful:optimize} does not lead to any beneficial trade-off so we will simply set the parameter $\delta$ to one.
We refer to a single edge update (insertion or deletion) or a reachability query as an \coldef{event}.
The algorithm first divides the sequence of events into \coldef{large chunks} of length $t = n^{ {25}/{49}}$, i.e., slightly larger than $\sqrt n$.
We show that a single chunk of this form can be processed in near-linear time.

\begin{lemma}
    There is an algorithm that, given an $n$-vertex planar digraph $G$ and a sequence of $t = n^{ {25}/{49}}$ query/update events, answers all the queries in total time $\Ot(n)$.
\end{lemma}
\begin{proof}
    Let us refer to the $i$-th event as $e_i$.
    For each $i \in [t]$ let $G_i$ be the digraph obtained from $G$ by applying all the updates from the events $e_1,e_2,\dots, e_{i-1}$.
    When the $i$-th event is a query $(u \to v)$, our goal is to report whether $u \overset{G_i}\leadsto v$.

    Let $X \sub V(G)$ be the set of vertices appearing in at least one event.
    Clearly, $|X| \le 2t$.
    Let $s = n^{ {3}/{49}}$.
    We divide the interval $[t]$ into $t / s$ \coldef{small chunks} of size $s$.
    For $j \in [t / s]$ let $i(j)$ denote the index of the first event of the $j$-th block and $Y_j \sub X$ be the set of vertices appearing in at least one event in the $j$-th block.
    Similarly as before, we have $|Y_j| \le 2s$.

    Let $G \ominus X$ be the graph obtained from $G$ by removing all the edges with both endpoints in $X$.
    Note that $G \ominus X = G_i \ominus X$ for all $i \in [t]$ because edge
    updates can affect only the vertex pairs from $X \times X$.
    We invoke \Cref{lem:faithful:construct}
    to compute a \fre $H$ for $(G \ominus X,X)$ of size $\Oh(t\cdot\log^2{n})$ in time $\Oh(n\cdot\log n)$.
    By \Cref{lem:emulator-union-simple}, for each  $i \in [t]$ the graph $H \cup G_i[X]$ is a \fre for $(G_i,X)$.

    The algorithm processes the small chunks one by one.
    We maintain the graph $G_{i(j)}[X]$ explicitly: when $j$ increases, we apply all the edge updates in the previously considered small chunk, in time $\Oh(s)$.
    Consider $j \in [t / s]$ and graph
    $R_j = G_{i(j)}[X] \ominus Y_j$.
    It can be computed directly from $G_{i(j)}[X]$ by inspecting all pairs in $Y_j \times Y_j$, in time $\Oh(s^2)$.
    Note that $(G_{i(j)} \ominus X) \cup R_j = G_{i(j)} \ominus Y_j$.
    Since $R_j$ is planar, it has $\Oh(|V(R_j)|) = \Oh(t)$ edges.
    By \Cref{lem:emulator-union-simple} we know that $H \cup R_j$ is a \fre for $(G_{i(j)} \ominus Y_j, \,X)$.
    The graph $H \cup R_j$ is referred to as $\widehat H$ in \Cref{fig:twolayers} on page \pageref{fig:twolayers}.
    Since $G_{i(j)} \ominus Y_j$ is planar, \Cref{thm:prelim:planar-vc} implies that both the VC dimension and the dual VC dimension of $\reachsys(G_{i(j)} \ominus Y_j,\, X)$ are at most three.

    We execute the algorithm from \Cref{lem:faithful:optimize} for $\delta = 1$, $d = 3$, the digraph $H \cup R_j$, and both $Y,Z$ set to $Y_j$.
    As an outcome, we receive a \re $J_j$ for $(G_{i(j)} \ominus Y_j,\, Y_j)$ of size $\Ot(|Y_j|^{5 / 3}) = \Ot(s^{5/3})$.
    We estimate this as just $\Oh(s^2)$ because processing $J_j$ will not form a bottleneck.
    The running time of the algorithm from \Cref{lem:faithful:optimize} is $\Ot(|H \cup R_j| \cdot |Y_j|^{2/3} + |Y_j|^9) = \Ot(t\cdot s^{2/3} + s^9)$.

    To process the queries in the $j$-th small chunk, we maintain explicitly the graph $G_i[Y_j]$ for each $i \in [t]$ corresponding to an event in the $j$-th small chunk.
    Note that $J_j$ is a \re for $(G_{i} \ominus Y_j,\, Y_j)$ for each such $i$ because $G_i \ominus Y_j = G_{i(j)} \ominus Y_j$ by the definition of $Y_j$.
    \Cref{lem:emulator-union} implies that $J_j \cup G_i[Y_j]$ is a \re for $(G_i, Y_j)$.
    Consider an event $e_i$ that is a query $(u \to v)$.
    Since $u,v \in Y_j$ we have $u \overset{G_i}\leadsto v$ if and only if $u \overset{J_j \cup G_i[Y_j]}\leadsto v$.
    Therefore, to address this query, it suffices to run depth-first search on $J_j \cup G_i[Y_j]$.
    This takes time linear in the size of the graph, which is $\Oh(s^2)$.

    Let us summarize the running time.
    Computing $H$ with \Cref{lem:faithful:construct} takes time $\Ot(n)$.
    At the beginning of each small chunk we spend time $\Ot(t\cdot s^{2/3} + s^9)$ to compute $J_j$.
    This dominates the term $\Oh(s^3)$ for
    preparing the graph $R_j$ and answering all the queries in the small chunk.
    Since the number of small chunks is $t/s$, the total running time is $\Ot(n + t^2\cdot s^{-1/3} + t\cdot s^8)$.
    For our choice of $t  = n^{ {25}/{49}}$
     and $s = n^{ {3}/{49}}$ we obtain $t^2\cdot s^{-1/3} + t\cdot s^8 = 2n$.
     This concludes the proof.
\end{proof}

\Cref{thm:main:offine} follows by dividing the sequence of events into large chunks of length  $t  = n^{ {25}/{49}}$ and processing each in time $\Ot(n)$.
Then the average time spent per event is $\Ot(n^{ {24}/{49}})$ and $\frac{24}{49} < 0.49$.

\section{Online Dynamic Reachability}
\label{sec:online}

This section is devoted to the proof of \Cref{thm:planar-dyn-reach}.
We first provide a background on $r$-divisions, next we present a simplified data structure that uses only half of our framework and achieves sublinear query/update time, and then we move on to the final construction, which has been outlined in \Cref{subsec:online}.

\subsection{Graph partitions}\label{s:rdiv}
In this paper, an \coldef{$r$-division}~\cite{DBLP:journals/siamcomp/Frederickson87} $\rdiv$ of a graph $G$, for $r \in [1,n]$,
is a collection of $\Oh(n/r)$ edge-disjoint\footnote{When defining an $r$-division, pieces are not always required to be edge-disjoint and may have some limited edge overlap (e.g.,~\cite{KleinMS13}). However, note that an edge shared among two pieces connects their common boundary vertices.} subgraphs $P\subseteq G$ of size $\Oh(r)$, called \coldef{pieces}, whose union equals $G$.
Each piece $P$ shares at most $\Oh(\sqrt{r})$ vertices with other
pieces.
The \coldef{boundary vertices} $\bnd{P}\subseteq V(P)$ include all vertices shared with other pieces. In the following, we will also assume that $\bnd{P}$ may include some additional pre-defined vertices, subject to $|\bnd{P}|=\Oh(\sqrt r)$.
The boundary $\bnd\rdiv$ of $\rdiv$ is defined as  $\bigcup_{P\in\rdiv}\bnd{P}$.

The fastest algorithm for computing
$r$-divisions of $K_h$-minor-free graphs follows by
 plugging the recent linear-time  algorithm for finding $\Oh(\sqrt{n})$-size balanced separators~\cite{DBLP:journals/corr/abs-2512-01587} into the partitioning algorithm of~\cite{DBLP:journals/siamcomp/Frederickson87}.

\begin{theorem}[\cite{DBLP:journals/corr/abs-2512-01587,DBLP:journals/siamcomp/Frederickson87}]\label{t:rdiv}
  Let $G$ be a simple $K_h$-minor-free graph with $n$ vertices, where $h=\Oh(1)$.
  For any $r\in [1,n]$, an $r$-division of $G$
  can be computed in $\Oh(n\log{n})$ time.
\end{theorem}

If $G$ is a planar graph, an $r$-division $\rdiv$ with stronger topological properties is possible.
We call $\rdiv$ an \coldef{$r$-division with few holes}
if
the boundary vertices of each piece $P\in\rdiv$ are distributed
among $\Oh(1)$ \emph{pairwise disjoint} {holes} of~$P$, defined as follows.\footnote{This definition is different from the usual one~\cite{KleinMS13}; see also~\Cref{app:rdiv-few-holes}. Namely, our definition of an $r$-division does not assume a fixed embedding of the entire $G$; it only assumes some fixed embeddings of individual pieces. Moreover, a hole is defined in a more general way that suits dynamic reachability much better.}
Given an embedded piece $P$, a \coldef{hole} is a weakly simple closed curve $h=\partial D$ in the plane (i.e., a closed curve that can possibly touch itself but not cross), where $D$ is a connected open region disjoint from the embedding of $P$. The region $D$ may be unbounded.
Note that a hole of $P$ is also a hole of any subgraph $P'\subseteq P$, assuming $P'$ inherits the embedding from $P$.

The next lemma follows easily from the decomposition algorithm of~\cite{KleinMS13}.
For completeness, it is derived from~\cite{KleinMS13} in~\Cref{app:rdiv-few-holes}.\begin{lemma}\label{lem:rdiv-few-holes-weak}
Let $G$ be a planar graph with $n$ vertices.
For any $r\in [1,n]$, an $r$-division with few holes of $G$ can be computed in $\Oh(n)$ time.
\end{lemma}

\paragraph{Initial boundaries.} The algorithms behind Theorems~\ref{t:rdiv}~and~\Cref{lem:rdiv-few-holes-weak} (see also~\Cref{app:rdiv-few-holes}) both recursively decompose the graph using $\Oh(\sqrt{n})$ sized balanced (cycle) separators~\cite{DBLP:journals/jcss/Miller86}.
On the way, one decomposes pieces $P$ with some boundary set $\bnd{P}$ into parts that balance certain quantities, such as size of $P$, boundary size, or the number of holes.
It is well-known (and has been exploited before, e.g., in~\cite{DBLP:conf/soda/GawrychowskiMWW18,DBLP:conf/soda/MozesS12}) that such decomposition algorithms work with no change if initialized with a pre-defined $\Oh(\sqrt{n})$-sized initial boundary set $\bnd{G}\subseteq V$
with the purpose of producing a decomposition $\rdiv$ satisfying $\bnd{G}\subseteq \bnd{\rdiv}$.
In the case of $r$-divisions with few holes, we require $\bnd G$ to lie on $\Oh(1)$ holes of $G$.
The pre-defined vertices are included in the boundary counts balanced by the recursion; the initial bound $|\bnd G|=\Oh(\sqrt n)$ preserves the usual piece size and boundary estimates.
For details, we refer the reader to, e.g.,~\cite[Section~5]{DBLP:conf/soda/GawrychowskiMWW18}.

\paragraph{Maintaining an $r$-division under edge updates.}
Suppose a graph $G=(V,E)$ equipped with an $r$-division $\rdiv$ is subject to edge updates.
First consider edge deletions.
Note that upon a deletion of edge $e\in E$, if one naturally removes~$e$ from the piece it resides in, the obtained collection of pieces remains an $r$-division in the sense that no new boundary vertices arise.
We assume that pieces never lose their vertices which implies the pieces' boundaries are fixed when $G$ is subject to edge deletions.

Similarly, if $G$ is a planar graph and $\rdiv$ is an $r$-division with few holes, deleting a piece's edge does not invalidate its holes: the piece's embedding remains weakly on the same side of the hole.

Handling insertions requires more care and leaves space for trade-offs.
First, when a new edge $(u,v)$ is inserted, no current piece may be able to accommodate it without breaking its invariants, e.g., concerning its size or the structure of boundary vertices.
A simple but effective way to deal with this problem is to make each inserted edge a new single-edge piece of $\rdiv$.
This generally requires promoting the edge's endpoints $u,v$ to boundary vertices of $\rdiv$ and appropriate pieces.
Note that this way $\Oh(\sqrt{r})$ edges can be inserted without compromising any properties of the original $r$-division asymptotically when $r=\Oh(n^{2/3})$:
the number of new pieces is $\Oh(\sqrt{r})\subseteq \Oh(n/r)$,
and a single piece can get $\Oh(\sqrt{r})$ new boundary vertices worst-case.

In the case of planar graphs and $r$-divisions with few holes, turning some arbitrary $u,v\in V(P)$ into boundary vertices generally breaks the $\Oh(1)$-holes boundary distribution of $P$.
While there are ways to deal with that (see, e.g.,~\cite{subramanian1993fully,DBLP:journals/jcss/Charalampopoulos22}), this will not be needed in our construction.

\subsection{The data structure for minor-free graphs}\label{sec:online-minor-free}

As a warm-up, we first present a simpler reachability data structure that only utilizes emulators from \Cref{sec:emul-vc}.
For this, we can work under a weaker assumption: we consider digraphs that exclude the clique $K_h$ as a minor (in the undirected sense) for any fixed $h \in \nn$.
To the best of our knowledge, this is the first fully dynamic reachability data structure beyond planar digraphs that achieves sublinear query and update time.
In \Cref{sec:conclusion} we discuss possibilities of obtaining better algorithms for this case~as~well.

\begin{theorem}\label{thm:minor-free-dyn-reach}
    Let $h \in \nn$. There exists a deterministic fully dynamic reachability data structure for directed $K_h$-minor-free graphs with $\Oh_h(n^{1-\Omega(1/h^2)})$ worst-case update and query time.
\end{theorem}

Within this subsection we assume $h= \Oh(1)$ and so we hide in the $\Oh$-notation factors depending on~$h$.

Let $r=n^{2/3}$.
The data structure for minor-free graphs operates in phases of $\sqrt{r}=n/r$ edge updates.
When a phase starts, an $r$-division $\rdiv$ of the current graph $G$ is recomputed from scratch in $\Ot(n)$ time using~\Cref{t:rdiv}.
The amortized time required by $r$-division recomputation is thus $\Ot(n/\sqrt{r})$ per update.

Besides the $r$-division, at all times we maintain for each piece $P\in\rdiv$ a reachability emulator $H_P$ for $(P,\bnd{P})$ obtained using~\Cref{lem:faithful:optimize-general}.
Since each piece $P$ is a $K_h$-minor-free graph, by~\Cref{thm:prelim:minor-vc},
the emulator $H_P$ has size
\begin{equation}\label{eq:minor-free-piece-emulator-size}
\Ot\left(|\bnd{P}|\cdot |\bnd{P}|^{1-\delta/(h-1)}+|\bnd{P}|\right)=\Ot\left(r^{1-\delta/(2h-2)}\right)
\end{equation}
and can be computed in time
\begin{equation}\label{eq:minor-free-piece-emulator-time}
\Ot\left(|P|\cdot |\bnd{P}|^{1-\delta/(h-1)}+|\bnd{P}|^{1+(2(h-1)+2)\delta}\right)=\Ot\left(r^{3/2-\delta/(2h-2)}+r^{1/2+h\delta}\right)
\end{equation}
where $\delta\in [0,1]$ is a constant to be set later.

Upon a phase start, all $\Oh(n/r)$ reachability emulators $H_P$ are recomputed from scratch; this takes
\[\Ot\left(\frac{n}{r}\cdot \left(r^{3/2-\delta/(2h-2)}+r^{1/2+h\delta}\right)\right)=\Ot\left(nr^{1/2-\delta/(2h-2)}+nr^{h\delta-1/2}\right)\]
time.
However, since phase initialization happens every $\sqrt{r}$ updates, the amortized time overhead incurred per update is
\begin{equation}\label{eq:minor-free-piece-emulator-recompute}\Ot\left(\frac{1}{\sqrt{r}}\cdot \left(nr^{1/2-\delta/(2h-2)}+nr^{h\delta-1/2}\right)\right)=\Ot\left(n/r^{\delta/(2h-2)}+nr^{h\delta-1}\right).
\end{equation}

\paragraph{Edge updates.}
While a phase proceeds, we maintain $\rdiv$ under edge updates as described in~\Cref{s:rdiv}.
Recall that the $r$-division undergoes $\sqrt{r}=n/r$ updates during a phase.
While edge deletion does not invalidate the $r$-division, we process edge insertion $(u,v)$
by introducing a new single-edge piece and adding $u,v$ to the boundaries of their pieces.
Note that a non-boundary vertex belongs to exactly one piece.
Therefore, we need to update boundaries of at most two pieces.

Upon edge update, at most three pieces (including the potential newly created one) of $\rdiv$ are affected by either an edge set update or a boundary set extension.
For these $\Oh(1)$ pieces, an emulator $H_P$ is recomputed from scratch using~\Cref{lem:faithful:optimize-general} in time estimated in \eqref{eq:minor-free-piece-emulator-time}.

\paragraph{Answering queries.}
Consider a reachability query $(u,v)$.
First, we temporarily make $u$ and $v$ boundary vertices of~$\rdiv$.
This involves recomputing the reachability emulators $H_P$ for at most two pieces $P\in\rdiv$, using time given in \eqref{eq:minor-free-piece-emulator-time}.
The temporary change of $\bnd{\rdiv}$ and all its consequences are reverted once the query procedure terminates.

To answer the query, we simply check whether a $u\to v$ path exists in the graph $H_{u,v}=\bigcup_{P\in \rdiv} H_P$ using a linear-time algorithm.
Note that by~\eqref{eq:minor-free-piece-emulator-size}, the graph $H_{u,v}$ has size $\Ot\left(n/r^{\delta/(2h-2)}\right)$.
Hence, the time needed to process this graph is dominated by \eqref{eq:minor-free-piece-emulator-recompute}.

\begin{lemma}
    The vertex $v$ is reachable from $u$ in $G$ if and only if it is reachable from $u$ in $H_{u,v}$.
\end{lemma}
\begin{proof}
By the definition of $\rdiv$, for any two distinct pieces $P,P'\in \rdiv$, we have $V(P)\cap V(P')=\bnd{P}\cap \bnd{P'}$.
Consequently, by~\Cref{lem:emulator-union-general}, $H_{u,v}=\bigcup_{P\in\rdiv}H_P$ is a reachability emulator for $\left(\bigcup_{P\in\rdiv}P,\,\,\bigcup_{P\in\rdiv}\bnd{P}\right)=(G,\bnd{\rdiv})$.
The lemma follows by $u,v\in \bnd{\rdiv}$ which is guaranteed at this point.
\end{proof}

\paragraph{Running time.}
The amortized time bound of both update and query follows from \eqref{eq:minor-free-piece-emulator-time} and \eqref{eq:minor-free-piece-emulator-recompute}, giving
\[ \Ot\left(n/r^{\delta/(2h-2)}+nr^{h\delta-1}+r^{3/2-\delta/(2h-2)}+r^{1/2+h\delta}\right). \]
Plugging in $r=n^{2/3}$, the bound becomes $\Ot(n^{1-\delta/(3h-3)}+n^{(1+2h\delta)/3})$.
Setting $\delta=1/(2h)$ yields an $\Ot(n^{1-1/(6h(h-1))})$ bound.

So far, we have proved an amortized update bound.
However, the entire amortization here comes from an expensive rebuilding step that happens at a predictable time (it is scheduled precisely every $\sqrt{r}$ updates).
Such rebuilding is subject to a standard deamortization technique of maintaining two copies of the data structure; see, e.g.,~\cite{DBLP:conf/soda/AbrahamCK17} and~\cite[Appendix~B, Theorem~B.1]{BrandNS19arxiv}.\space{}
We have thus proved~\Cref{thm:minor-free-dyn-reach}.

We remark that for a fixed value of $h$ the terms can be balanced in a slightly better way, by incorporating more tedious analysis. This however does not improve the asymptotic term $\Omega(1/ {h^2})$.

\subsection{Decremental boundary emulator}\label{sec:decr-emulator-overview}

We now refer to the decremental transitive closure data structure of Karczmarz~\cite{DBLP:conf/soda/Karczmarz18} that will be needed in our online data structure for planar graphs.
\cite{DBLP:conf/soda/Karczmarz18} showed that a planar digraph can be maintained subject to edge deletions so that reachability queries are supported in $\Ot(\sqrt{n})$ time, and the total update time of the data structure is $\Ot(n)$.
That is, if all the edges are eventually deleted, the amortized update time is $\Ot(1)$ per deletion.
The data structure of~\cite{DBLP:conf/soda/Karczmarz18} operates by explicitly maintaining reachability matrices for boundary sets $\bnd{P}$ of all pieces $P$ of a recursive decomposition of the graph using cycle separators~\cite{DBLP:journals/jcss/Miller86} (see also~\cite{KleinMS13}).
Each $\bnd{P}$ in this decomposition lies on $\Oh(1)$ holes of $P$.
In particular, if the decomposition algorithm in~\cite{DBLP:conf/soda/Karczmarz18} is initialized with a pre-defined boundary set $\bnd{G}$ of size $\Oh(\sqrt{n})$ lying on $\Oh(1)$ pairwise disjoint holes of $G$ (as detailed in~\Cref{app:decremental-emulator}), then the data structure can maintain the full reachability matrix of $\bnd{G}$ within $\Ot(n)$ total time subject to deletions in~$G$.
By exploiting the structural properties of this reachability matrix, we can efficiently convert it into a reachability emulator for $(G,\bnd{G})$ of size $\Ot(|\bnd{G}|)=\Ot(\sqrt{n})$.
As a result, we obtain the following theorem.

\begin{theorem}[\cite{DBLP:conf/soda/Karczmarz18}]\label{thm:decremental-bnd-emulator}
  Let $G=(V,E)$ be a plane digraph. Let $\bnd{G}\subseteq V$ be a subset of $\Oh(\sqrt{n})$ vertices lying on $\Oh(1)$ pairwise disjoint holes of $G$.
    There exists a data structure, initialized in $\Ot(n)$ time, that explicitly maintains a reachability emulator for $(G,\bnd{G})$ of size $\Ot(\sqrt{n})$ when $G$ is subject to edge deletions.
    The total update time of the data structure is $\Ot(n+k\sqrt{n})$, where $k$ is the number of deletions issued.
    The data structure is randomized and the results produced are correct with high probability.
\end{theorem}
Above, by high probability correctness we mean that for any $c>0$, the constants in the data structure can be set in such a way that the probability of any of the $k+1$ emulators output being incorrect is at most $n^{-c}$.
In other words, this bounds the probability that we receive any
invalid \re for $(G,\bnd{G})$ during a fixed sequence of edge deletions.
In particular, we assume that the sequence of updates to $G$ is independent of the random bits used by the data structure.
Since in this paper we only consider reachability queries returning a single bit that is a function of the current graph~$G$, the required independence is guaranteed unless the data structure ever errs, which happens with low probability.
The high probability correctness of the data structure described in the following section is defined likewise.
More specifically, over a polynomial sequence of updates to $G$, we will be applying~\Cref{thm:decremental-bnd-emulator} for a polynomial number of auxiliary graphs, potentially padded to have $\Theta(\operatorname{poly}(n))$ vertices. The high probability correctness of the data structure's answers will follow via the union bound.

In~\Cref{app:decremental-emulator} we discuss~\cite{DBLP:conf/soda/Karczmarz18} in more detail and show how it is modified to yield~\Cref{thm:decremental-bnd-emulator}.

\subsection{The data structure for planar graphs}
Our data structure for planar graphs follows the same overall template as the earlier data structure for minor-free graphs (\Cref{sec:online-minor-free}) but the way piecewise emulators are updated is much more complex.
Moreover, the inserted edges' endpoints are treated differently than the $r$-division's boundary vertices.

Let $s,r\in [1,n]$, $s<r$, and $\delta\in [0,1]$ be parameters to be set later.
Again, the data structure operates in phases of $\sqrt{r}$ edge updates.
When a phase starts:
\begin{enumerate}
\item An {\em $r$-division with few holes} $\rdiv$ of the current graph $G$ is recomputed from scratch in $\Oh(n)$ time using~\Cref{lem:rdiv-few-holes-weak}.
\item For each piece $P\in \rdiv$, an $s$-division $\rdiv_P$ of $P$ is computed using~\Cref{t:rdiv}, with the boundary initialized to $\bnd{P}$ (as discussed in~\Cref{s:rdiv}), so that $\bnd{P}\subseteq \bnd{\rdiv_P}$ holds.

\end{enumerate}
While a phase proceeds, we denote by $G_0$ the graph at the beginning of the phase, whereas by~$G$ the current version of the graph.
Let $E^-\subseteq E(G_0)$ denote the edges deleted at least once since the phase started.
Let $E^+\subseteq E(G)$ be the currently present edges inserted in the current phase, including reinsertions of edges from $E^-$.

Denote by $G^-$ the graph $G_0$ with edges $E^-$ removed.
Hence, we have $G=G^-\cup E^+$ at all times.
Finally, let $B\subseteq V\setminus \bnd\rdiv$ be the set of endpoints of edges that ever get inserted to $E^+$ in a phase, except boundary vertices of~$\rdiv$.
Observe that $|B|\leq 2\sqrt{r}$.

Recall from the discussion in~\Cref{s:rdiv} that $r$-divisions (with few holes) remain valid under edge deletions.
Thus, $\rdiv$ computed for~$G_0$ remains valid for~$G^-$ throughout the phase.
Similarly, the piecewise $s$-divisions $\rdiv_P$ remain valid for the respective subgraphs $P\cap G^-$.
In the following, we will thus treat $\rdiv$ and all local $\rdiv_P$ as divisions for the decremental graph $G^-$ and its subgraphs, respectively.

For a piece $Q$ (of either $\rdiv$, or of $\rdiv_P$ for some $P\in\rdiv$), put $B_Q:=B\cap V(Q)$.
We use the following auxiliary data structures for each $P\in \rdiv$ during a phase:
\begin{enumerate}
\item A decremental emulator data structure from~\Cref{thm:decremental-bnd-emulator} (with the input boundary set $\bnd{P}$).
Denote by $H^*_P$ the $\Ot(\sqrt{r})$-sized emulator it maintains.
\item For each subpiece $Q\in \rdiv_P$, a \emph{faithful} emulator $H_{P,Q}$ for  $(Q,\bnd{Q}\cup B_Q)$, and a \emph{faithful} reverse emulator $H_{P,Q^r}$ for $(Q^r,\bnd{Q}\cup B_Q)$, both of size $\Ot(|\bnd{Q}\cup B_Q|)$. Recall that $Q^r$ is obtained from $Q$ by reversing all the edges.
\item Based on the above, an emulator $H_P$ for $(P,\,\bnd{P}\cup B_P)$ of size $\Ot\left(|B_P|\cdot r^{1/2-\delta/6} + \sqrt{r}\right)$.
The \linebreak construction of $H_P$ is given in~\Cref{lem:recompute-emulator}.
\end{enumerate}
By \Cref{t:rdiv},~\Cref{lem:rdiv-few-holes-weak},~\Cref{thm:decremental-bnd-emulator}~and~\Cref{lem:faithful:construct}, all the described data structures can be initialized in near-linear time when a phase starts.
Indeed, since $B=\emptyset$ at that time, we can simply put $H_P:=H_P^*$.
Consequently, the amortized time spent on this per update is $\Ot(n/\sqrt{r})$.

Let us now explain how the emulator $H_P$ can be computed from the other maintained emulators.
\begin{lemma}\label{lem:recompute-emulator}
    Let $P\in\rdiv$. Suppose $H^*_P$ and $H_{P,Q}$, $H_{P,Q^r}$ for all $Q\in \rdiv_P$ are given.
    Then $H_P$ can be computed in $\Ot\left(r^{3/2-\delta/6}/\sqrt{s}+r^{(1+8\delta)/2}\right)$ time.
\end{lemma}
\begin{proof}
Note that $B_P\subseteq V(P)\setminus \bnd{P}$.
Recall that $H^*_P$ is a reachability emulator for $(P,\bnd{P})$.
We have $(\bnd{P}\cup B_P)^2=\bnd{P}^2\cup ((\bnd{P}\cup B_P)\times B_P)\cup (B_P\times (\bnd{P}\cup B_P))$.
Thus, using~\Cref{l:assemble-from-parts} we can obtain the desired emulator $H_P$ from $H_P^*$ and reachability emulators for $(P,\bnd{P}\cup B_P,B_P)$ and $(P,B_P,\bnd{P}\cup B_P)$, constructed as follows.

Let us first consider computing an emulator for $(P,B_P,\bnd{P}\cup B_P)$.
To this end, let $H_0=\bigcup_{Q\in \rdiv_P}H_{P,Q}$.
Since for all $Q_1,Q_2\in \rdiv_P$ we have
$V(Q_1)\cap V(Q_2)\subseteq \bnd{Q_1}\cap \bnd{Q_2}$, \Cref{lem:emulator-union-general} implies that
$H_0$ is a faithful emulator for $\left(\bigcup_{Q\in\rdiv_P}Q,\,\,\bigcup_{Q\in\rdiv_P}(\bnd{Q}\cup B_Q)\right)=(P,\,\bnd{\rdiv_P}\cup B_P)$ of size
\[\Ot\left(\sum_{Q\in\rdiv_P}\left(|\bnd{Q}|+|B_Q\setminus\bnd{Q}|\right)\right)=\Ot\left(\frac{r}{s}\cdot \sqrt{s}+|B_P|\right)=\Ot\left(\frac{r}{\sqrt{s}}+|B_P|\right)=\Ot\left(\frac{r}{\sqrt{s}}\right).\]
Above, we have used $|B_P|=\Oh(\sqrt{r})$ and $s<r$.
Consequently, by~\Cref{thm:prelim:planar-vc}~and~\Cref{lem:faithful:optimize}, a reachability emulator
for $(P,B_P,\bnd{P}\cup B_P)$ of size \[
\Ot\left(|B_P|\cdot |\bnd{P}\cup B_P|^{1-\delta/3} + |\bnd{P}\cup B_P|\right)=\Ot\left(|B_P|\cdot r^{1/2-\delta/6} + \sqrt r\right)\]
can be computed in time
\[ \Ot\left(|H_0|\cdot r^{1/2-\delta/6}+|\bnd{P}\cup B_P|^{1+8\delta}\right)=\Ot\left(r^{3/2-\delta/6}/\sqrt{s}+r^{(1+8\delta)/2}\right). \]

To compute the symmetric emulator for $(P,\bnd{P}\cup B_P,B_P)$ we proceed analogously, but using the ``reverse'' data structures, as explained next.
Note that $\bigcup_{Q\in \rdiv_P}H_{P,Q^r}$ is a faithful reachability emulator for $(P^r,\bnd{\rdiv_P}\cup B_P)$.
We can thus obtain an emulator for $(P^r,B_P,\bnd{P}\cup B_P)$ using~\Cref{lem:faithful:optimize}.
A reverse of that emulator is an emulator for $(P,\bnd{P}\cup B_P,B_P)$.

The desired size bound of $H_P$ is obtained by taking into account the $\Ot(\sqrt{r})$ size bound of $H^*_P$.
\end{proof}

\paragraph{Edge insertion.}
Upon insertion of an edge $e=(u,v)$, we append it to $E^+$.
Recall that we only process pieces of $G^-$ so $e$ is not being
added to any of the maintained emulators.
However, the vertices $\{u,v\}\setminus \bnd{\rdiv}$ need to be added to the set $B$.
Suppose that one of them, say $u$, does not belong to $\bnd {\rdiv}$.
Let $P$ be the unique piece from $\rdiv$ containing $u$.
If $u \in \bnd{\rdiv_P}$ then $u$ is already considered in the faithful emulators for the $s$-division $\rdiv_P$.
Otherwise, there is a unique piece $Q$ from $\rdiv_P$ containing $u$.
In this case, we recompute the emulators $H_{P,Q},H_{P,Q^r}$ from scratch in $\Ot(|Q|)=\Ot(s)$ time using~\Cref{lem:faithful:construct}.
Next, in both cases, we recompute $H_P$ using~\Cref{lem:recompute-emulator}
in $\Ot\left(r^{3/2-\delta/6}/\sqrt{s}+r^{(1+8\delta)/2}\right)$ time.
Summing up, single edge insertion incurs time
\begin{equation}\label{eq:edge-insert}
\Ot\left(r^{3/2-\delta/6}/\sqrt{s}+r^{(1+8\delta)/2}+s\right).
\end{equation}

\paragraph{Edge deletion.}
Consider a deletion of edge $e\in E$.
Note that the set $B$ does not grow due to edge deletions.
If $e\in E^+$, that is, $e$ was inserted in the current phase, we simply remove $e$ from that set and terminate.

Otherwise, $e$ is an edge of a unique piece $P\in \rdiv$
and a unique piece $Q\in \rdiv_P$.
We delete $e$ from the pieces it resides in.
Afterwards, the maintained emulators need to be updated to reflect the changes in~$P$ and $Q$.
We thus pass the deletion to the data structure of~\Cref{thm:decremental-bnd-emulator}
which updates the emulator~$H_P^*$. The resulting emulator is correct with high probability, as discussed in~\Cref{sec:decr-emulator-overview}.
We will analyze the amortized time of this operation in the end.
Then, we recompute the faithful emulators $H_{P,Q}$, $H_{P,Q^r}$, and subsequently $H_P$ as in the case of edge insertion.
The cost of these operations is estimated in \eqref{eq:edge-insert}.

\paragraph{Query.}
To handle a query about a $u\to v$ path in the current $G$, we first proceed as in the case of edge insertion by temporarily
adding $\{u,v\}\setminus (\bnd{\rdiv}\cup B)$ to $B$ and updating the maintained emulators accordingly.
We then check whether a $u\to v$ path exists in the graph \[H_{u,v}=\left(\bigcup_{P\in\rdiv}H_P\right)\,\cup\, (\bnd{\rdiv}\cup B,E^+).\]
Recall that $\sum_{P\in\rdiv} |B_P| = |B| \le 2\sqrt r$ and so the graph $H_{u,v}$ has size
\[ \sum_{P\in\rdiv}\Ot\left(\sqrt{r}+|B_P|\cdot r^{1/2-\delta/6}\right)+\Oh(\sqrt{r})=\Ot\left(\frac{n}{\sqrt{r}}+r^{1-\delta/6}\right).\]
The above also bounds the cost of graph search on $H_{u,v}$.
The correctness is analyzed in the  lemma below.
\begin{lemma}
    The vertex $v$ is reachable from $u$ in $G$ if and only if it is reachable from $u$ in $H_{u,v}$.
\end{lemma}
\begin{proof}
Let $P,P'\in\rdiv$.
Note that $V(P)\cap V(P')=\bnd{P}\cap\bnd{P'}$.
Moreover, recall that the sets $B$ and $\bnd{\rdiv}$ are disjoint, so $B_P$ and $B_{P'}$ are also disjoint.
As a result, we have $V(P)\cap V(P')=(\bnd{P}\cup B_P)\cap (\bnd{P'}\cup B_{P'})$.
We can thus apply~\Cref{lem:emulator-union-general} to the emulators $H_{P}$, $P\in\rdiv$, obtaining that $\bigcup_{P\in\rdiv} H_P$ is a reachability emulator for
\[ \left(\bigcup_{P\in\rdiv} P,\,\,\bigcup_{P\in\rdiv}(\bnd{P}\cup B_P)\right)=(G^-,\,\bnd{\rdiv}\cup B).\]
But the graph $(\bnd\rdiv\cup B,E^+)$ is trivially an emulator for itself,
so by applying~\Cref{lem:emulator-union} to $\bigcup_{P\in\rdiv}H_P$ and the graph $(\bnd\rdiv\cup B,E^+)$, we conclude that their union $H_{u,v}$ is a reachability emulator for $(G^-\cup E^+,\bnd{\rdiv}\cup B)=(G,\bnd{\rdiv}\cup B)$.
In particular, since $\{u,v\}\subseteq \bnd{\rdiv}\cup B$, and by the definition of a reachability emulator, an $u\to v$ path exists in $G$ if and only if an $u\to v$ path exists in $H_{u,v}$.
\end{proof}
After testing reachability in $H_{u,v}$, we revert the inclusion of the temporarily inserted vertices in $B$.
The total cost of the query procedure is thus given by
\eqref{eq:edge-insert} for
the temporary changes made to $B$, and
$\Ot\left(\frac{n}{\sqrt{r}}+r^{1-\delta/6}\right)$ for running graph search on $H_{u,v}$, that is:
\begin{equation}\label{eq:online-query-cost}
    \Ot\left(\frac{n}{\sqrt{r}}+r^{1-\delta/6}+r^{3/2-\delta/6}/\sqrt{s}+r^{(1+8\delta)/2}+s\right).
\end{equation}

\paragraph{Running time.}
Recall that in each phase, all the used data structures are initialized in $\Ot(n)$ time. Hence, the amortized phase initialization cost per update is $\Ot(n/\sqrt{r})$.
Denote by $d_P$ the total number of edges deleted from piece $P$ in a phase.
Clearly, we have $\sum_{P\in\rdiv}d_P\leq \sqrt{r}$.
The total time spent on updating the emulators $H^*_P$ maintained by data structures of~\Cref{thm:decremental-bnd-emulator} in a single phase is hence
\[ \Ot\left(\sum_{P\in\rdiv}|P|+d_P\sqrt{r}\right)=\Ot\left(n+\sqrt{r}\sum_{P\in\rdiv}d_P\right)=\Ot(n+r)=\Ot(n). \]
Again, this is $\Ot(n/\sqrt{r})$ amortized time per update.

Note that the time cost~\eqref{eq:online-query-cost} of the query procedure dominates the cost of performing edge insertions and deletions (excluding the time spent in the decremental emulator data structure of~\Cref{thm:decremental-bnd-emulator}).
It also subsumes the aforementioned $\Ot(n/\sqrt{r})$ amortized cost of initialization and maintaining the emulators $H^*_P$.

Hence, to minimize the maximum of these time costs, it is enough to choose $r,s,\delta$, so that the query cost~\eqref{eq:online-query-cost} is minimized.
Intuitively, picking $s$ to be slightly but polynomially smaller than $r$ and a very small $\delta>0$ leads to an $\Ot(n/\sqrt{r}+r^{0.99})$ bound.
The optimal choice is $r=n^{37/55}$, $s=n^{73/110}$, $\delta=9/74$;
then the bound becomes $\Ot(n^{73/110})=\Oh(n^{0.664})$.
We have thus proved~\Cref{thm:planar-dyn-reach}.

\section{Improved VC dimension bound for planar reachability set systems}
\label{sec:improved-vc}

In this section we prove \Cref{thm:prelim:planar-vc} from \Cref{sec:prelim} improving the bound on the VC dimension of planar reachability set systems from four to three.

\lemPlanarVC*
\begin{proof}
    By \Cref{thm:prelim:dual-vc} it suffices to show that VC dimension of $\reachsys(G)$ is at most~three for every planar digraph.
    Suppose for the sake of contradiction that there exists a planar digraph $G$ with $v_1,v_2,v_3,v_4 \in V(G)$ such that $Y = \{v_1,v_2,v_3,v_4\}$ is shattered by $\reachsys(G)$.
    Let $\pi$ be some topological ordering of the strongly connected components of $G$.
    For each pair $(i,j)$, $1\le i<j \le 4$, there exists $u_{ij} \in V(G)$ such that $\reach_G(u_{ij}, Y) = \{v_i,v_j\}$; among possibly multiple choices for $u_{ij}$ we choose a vertex whose strongly connected component, denoted $U_{ij}$, is the latest in the ordering~$\pi$.

    Let $P_{ij}^{i}$ and $P_{ij}^{j}$ be the corresponding dipaths from $u_{ij}$ to $v_i$ and $v_j$, respectively.
    Next, for  $i \in [4]$ let $V_i$ denote the set of all vertices appearing on dipaths $P_{-}^{i}$ minus the vertices in the sets $U_{-}$.

    \begin{claim}\label{claim:prelim:antichain}
        For each distinct $i, j \in [4]$ and each vertex $x \in V_i$ it holds that $v_j \not\in \reach_G(x,Y)$.
    \end{claim}
    \begin{innerproof}
        Suppose that $x \in V_i$ and $v_j \in \reach_G(x,Y)$.
         Then $ \{v_i,v_j\} \sub \reach_G(x, Y)$.
         Hence, the only vertex of the form $u_{-}$ that can reach $x$ is $u_{ij}$ which implies
        $x \in V(P_{ij}^{i})$ and in fact  $ \{v_i,v_j\} = \reach_G(x, Y)$.
        Since $V_i \cap U_{ij} = \emptyset$ we obtain $x \not\in U_{ij}$.
        It follows that the strongly connected component of $x$ is later than $U_{ij}$ in the ordering $\pi$.
         This contradicts the choice of $u_{ij}$.
    \end{innerproof}

    It follows that the sets $V_1,V_2,V_3,V_4$ are pairwise disjoint.

\begin{figure}[t]
    \centering
    \includegraphics[width=0.8\textwidth]{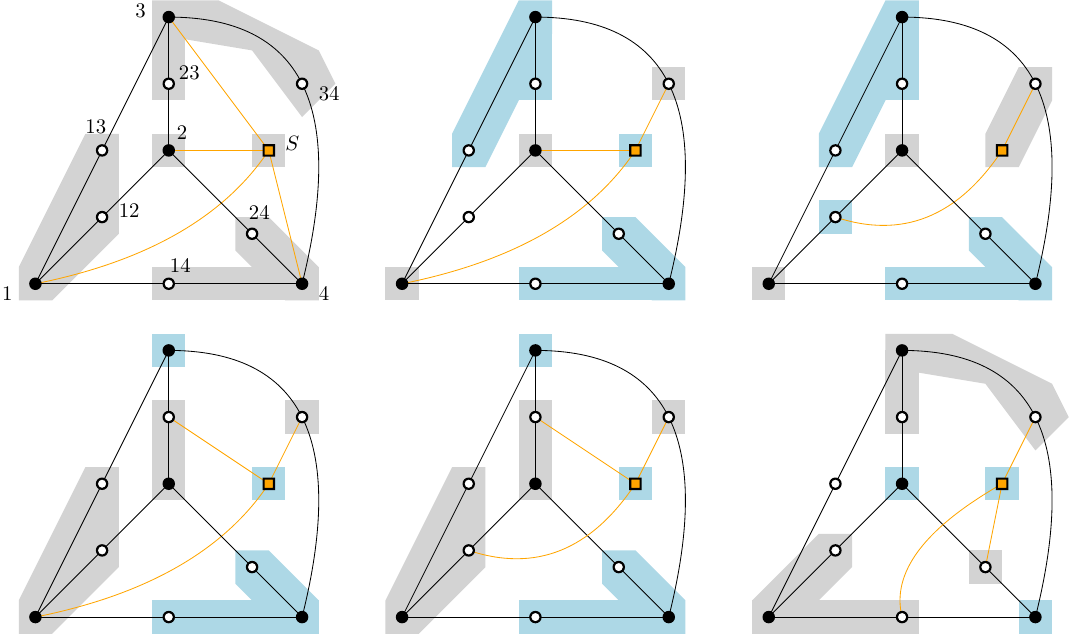}
    \caption{Illustration of the graph $H$ in the six scenarios in the proof of \Cref{thm:prelim:planar-vc}.
    The black vertices originate from contracting the sets $V_1,V_2,V_3,V_4$ and
    the white disks represent the vertices corresponding to the sets $U_-$, as marked with the labels in the top left graph.
    These vertices together induce a subdivision of $K_4$.
    The orange square is the vertex created by contracting the set $S$.
    For each $i \in [4]$ it must either be adjacent to the $V_i$-vertex or have a common neighbor with the $V_i$-vertex.
    Regardless of the choice of a face to place the $S$-vertex, drawing one of these edges is impossible without a crossing.
    To justify this formally, we consider all the possibilities modulo symmetry.
    In each scenario we have marked a minor model of either $K_5$ or $K_{3,3}$, which certifies that $H$ is not planar.
    }
    \label{fig:planar-vc}
\end{figure}
    \begin{claim}\label{claim:prelim:connected}
        For each $i \in [4]$ the set $V_i$ is non-empty and the graph $G[V_i]$ is weakly connected.
    \end{claim}
    \begin{innerproof}
        Observe that $\reach_G(v_i, Y) = \{v_i\}$
        because if $v_i \overset{G}\leadsto v_j$ for some $j \ne i$ then \linebreak $\{v_i,v_j,v_k\} \sub \reach_G(u_{ik}, Y)$ for any $k \notin \{i,j\}$.
        Consequently, $v_i \in V_i$ and it suffices to show that every vertex in $V_i$ can reach $v_i$ within $G[V_i]$.

        Let $x \in V_i$.
        This vertex lies on dipath $P_{ij}^{i}$ for some $j \ne i$.
        We claim that every vertex $y$ located on the dipath $P_{ij}^{i}$ after $x$ also belongs to $V_i$.
        Suppose it does not.
        By \Cref{claim:prelim:antichain} no vertex in $Y \sm \{v_i\}$
         is reachable from $x$ and hence also from $y$.
        Therefore, $y$ cannot belong to any set $U_{-}$ because then it would reach two vertices from $Y$.
        This implies $y \in V_i$.
        Consequently, $x$ can reach $v_i$ within $G[V_i]$ along the dipath $P_{ij}^{i}$.
    \end{innerproof}

    By \Cref{claim:prelim:antichain} all the sets $\{U_{ij} \colon 1 \le i < j \le 4\}$ and $V_1,V_2,V_3,V_4$ are pairwise disjoint.
    By \Cref{claim:prelim:connected} each of them induces a weakly connected subgraph.
    Contracting each of these sets into single vertices yields a minor model of a subdivision of $K_4$ (see \Cref{fig:planar-vc}).
    Let $U$ denote the union of all sets $U_-, V_-$.

    Since $Y$ is shattered, there is also a vertex $z \in V(G)$ for which $\reach_G(z,Y) = Y$.
    Note that $z \not\in U$.
    For $i \in [4]$ let $P_z^i$ be some fixed dipath from $z$ to $v_i$.
    Next, let $S_z^i$ be the maximal prefix of  $P_z^i$ that is disjoint from $U$.
    We define $S$ as the union of $\{S_z^i \colon i \in [4]\}$.
    \begin{claim}\label{claim:prelim:1234}
        For each $i \in [4]$ either $V_i \cap N_G(S) \ne \emptyset$ or there is $j \ne i$ so that $U_{ij} \cap N_G(S) \ne \emptyset$.
    \end{claim}
    \begin{innerproof}
        We have $v_i \in V_i \sub U$ hence the dipath $P_z^i$ enters the set $V_i$ at some point.
        Suppose however that the first vertex $x$ on $P_z^i$ that lies in $U$ is not from $V_i$.
        But $x \not\in V_j$ for any $j \ne i$ due to \Cref{claim:prelim:antichain}.
        Hence $x \in U_{jk}$ for some $j,k \in [4]$.
        This implies that $\{v_i,v_j,v_k\} \sub \reach_G(u_{jk},Y)$.
        But $u_{jk}$ reaches only two vertices from $Y$, so
     $i \in \{j,k\}$. The claim follows.
    \end{innerproof}

    Modulo symmetry, one of the six following scenarios occurs: $S$ is adjacent to either (a) $V_1,V_2,V_3,V_4$ or (b) $V_1,V_2,U_{34}$ or (c) $U_{12},U_{34}$ or (d)  $V_1,U_{23},U_{34}$ or (e) $U_{12},U_{23},U_{34}$ or (f) $U_{14},U_{24},U_{34}$.
    Contract each of the sets $S$, $U_-$, and $V_-$ into single vertices, and delete all other vertices and all edges except those of the indicated subdivision of $K_4$ and the edges from $S$ specified by one of the six scenarios.
    The resulting minor $H$ of $G$ is a subdivision of $K_4$ plus the vertex originating from $S$.
    In \Cref{fig:planar-vc} we show that such a graph $H$ cannot be planar, so $G$ is not planar as well.
    The obtained contradiction implies that no set of four vertices can be shattered by $\reachsys(G)$ hence its VC dimension is at most three.
\end{proof}

\section{Lower bound for planar \re{s}}
\label{sec:lower-bound}

\begin{figure}[t]
    \centering
    \includegraphics[width=0.9\textwidth]{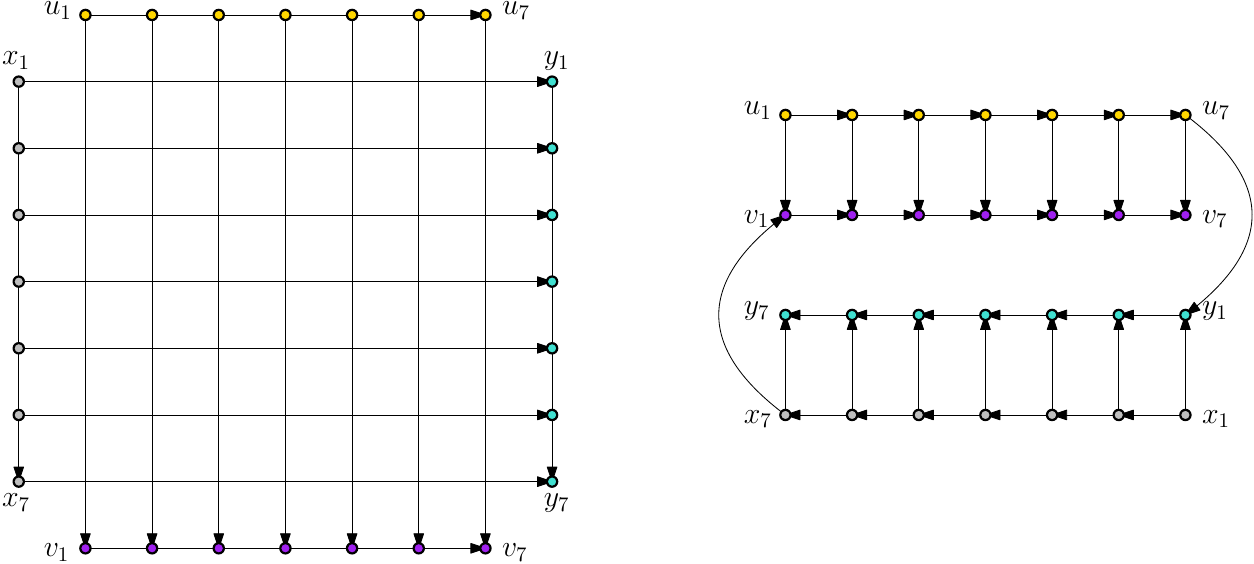}
    \caption{Left: The construction $(G,T)$ from \cite{Goranci20} for which every reachability preserving minor requires $\Omega(|T|^2)$ vertices.
    We orient the horizontal edges to the right and the vertical ones to the bottom.
    The terminal set $T$ comprises four groups $u_i, v_i,x_i, y_i$, $i \in [k]$, occupying the boundary of the grid.
    We have marked the first and the last vertex in each group.
    Right: A planar \re for $(G,T)$ on $\Oh(|T|)$ vertices.
    The lower bound from \cite{Goranci20} does not apply when one can use a different planar graph with unrestricted arrangement of terminals.
    }
    \label{fig:counter}
\end{figure}

\begin{figure}[h]
    \centering
    \includegraphics[width=0.5\textwidth]{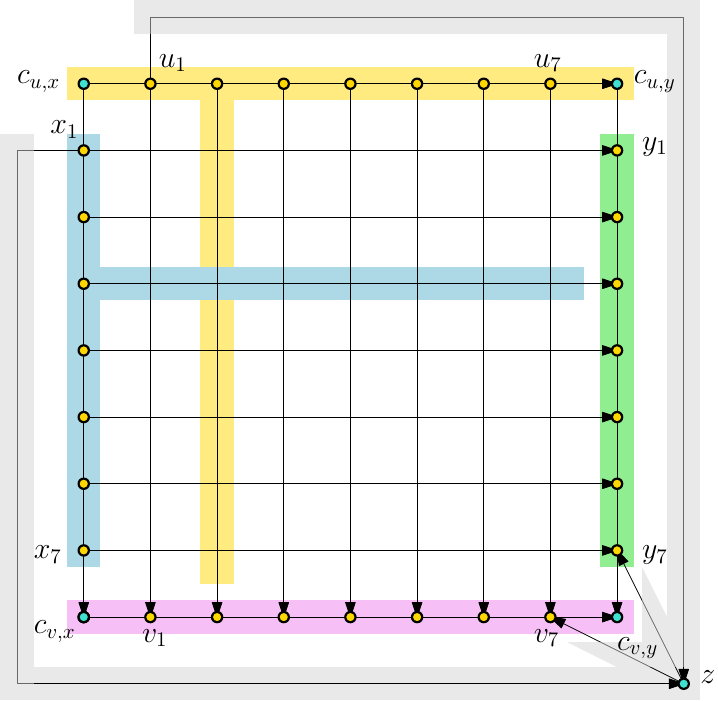}
    \caption{Illustration to \Cref{thm:lb}.
    The graph $G$ is obtained from the one in \Cref{fig:counter} by inserting the  corners $c_{u,x}, c_{u,y}, c_{v,x}, c_{v,y}$, a new vertex $z$, and edges $(x_1,z)$, $(u_1,z)$, $(z,y_k)$, $(z,v_k)$.
    All the five inserted vertices also become terminals, in addition to  $u_i, v_i,x_i, y_i$, $i \in [k]$.
    The colored areas illustrate a minor model of $K_5$ appearing when the horizontal path $P_3$ and the vertical path $Q_2$ have empty intersection in the emulator (\Cref{claim:lb:intersect}).
    }
    \label{fig:lower-bound}
\end{figure}

Goranci, Henzinger, and Peng~\cite{Goranci20}
considered {\em reachability preserving minors} which are reachability emulators that can be obtained from the original digraph by edge removal, edge contraction, and vertex removal, without changing the locations of terminals.
They gave an example showing that a reachability preserving minor for $(G,T)$ may need $\Omega(|T|^2)$ vertices, even when $G$ is a directed grid and $T$ comprises its boundary vertices.
Their example however does not rule out possibility of constructing \re{s} that are not minors of $G$ but they are still planar.
This kind of sparsification has been used for preserving approximate distances in undirected planar graphs~\cite{ChangKT22, chang2022arxiv} which was later improved to ensure the minor relation~\cite{chang2025distance}.
In our setting, we show on \Cref{fig:counter} that
rearranging the terminal locations in
the example from~\cite{Goranci20} allows for a planar \re of size $\Oh(|T|)$.
So their argument is not sufficient to completely rule out repetitive reachability compression discussed in \Cref{sec:sparsification}.

We strengthen the argument from~\cite{Goranci20} to argue that the same lower bound applies when we do not require $H$ to be a minor of $G$, but we only request $H$ to be planar.
We need more freedom in setting up the terminals so they no longer lie on a single face but they still occupy only $\Oh(1)$ faces, which is a common setup in arguments relying on, e.g., $r$-divisions.

\thmLB*
\begin{proof}
    The construction of $(G,T)$ is depicted on \Cref{fig:lower-bound}.
    The graph $G$ is a directed $(k+2)\times(k+2)$ grid, oriented bottom-right, plus a special vertex $z$.
    The terminal set $T$ is described in the caption therein.

    Consider a planar \re $H$ for $(G,T)$ that has minimum possible number of vertices.
    We will prove that $|V(H)| \ge (k-2)^2$.
    To simplify our arguments, we first show that $H$ cannot contain edges that could be contracted.

    \begin{claim}\label{claim:lb:adj-reach}
        Let $u,v \in V(H)$ and $(u,v) \in E(H)$.
        Then $\reach_H(u,T) \ne \reach_H(v,T)$.
    \end{claim}
    \begin{innerproof}
        Suppose that $\reach_H(u,T) = \reach_H(v,T)$.
        Note that each $t \in T$ has different set $\reach_H(t,T)$
        so at least one of $u,v$ does not belong to $T$.
        We claim that digraph $H'$, obtained by contraction of
        $\{u,v\}$ into a vertex $uv$,
         remains a \re for $(G,T)$.
        Contracting an edge can only increase reachability between some vertices so we only need to show $t_1 \overset{H'}\leadsto t_2 \Rightarrow t_1 \overset{H}\leadsto t_2$.
        Consider any $t_1,t_2 \in T$ and an $(t_1,t_2)$-dipath $P$ in $H'$.
        If $uv \not\in V(P)$ then $P$ is also present in $H$.
        Suppose then that $uv \in V(P)$.
        It follows that $t_1 \overset{H}\leadsto v$ and $u \overset{H}\leadsto t_2$.
        But $\reach_H(u,T) = \reach_H(v,T)$ hence $v \overset{H}\leadsto t_2$ and so $t_1 \overset{H}\leadsto t_2$.
        Therefore, $H'$ is indeed a \re for $(G,T)$.
        As contraction preserves planarity, this contradicts the choice of $H$ as a minimal planar emulator.
    \end{innerproof}

    \begin{claim}\label{claim:lb:adj-path}
        Let $a,b \in T$ and $w \in V(H) \sm \{a,b\}$ be a vertex lying on an
         $(a,b)$-dipath $P$ in $H$.
         Then $\reach_H(b,T) \subsetneq \reach_H(w,T) \subsetneq \reach_H(a,T)$.
    \end{claim}
    \begin{innerproof}
        We have $\reach_H(b,T) \subseteq \reach_H(w,T) \subseteq \reach_H(a,T)$ because $a \overset{H}\leadsto w \overset{H}\leadsto b$.
        If $\reach_H(b,T) = \reach_H(w,T)$ then all vertices on $P$ after $w$ also share the same reachability set.
        In particular, two such vertices are adjacent, which contradicts \Cref{claim:lb:adj-reach}.
        The case $\reach_H(a,T) = \reach_H(w,T)$ is symmetric.
    \end{innerproof}

Next, we will show that all the edges $(a,b)$ in $G$ connecting the terminals on the grid must be also present in $H$.
To this end, we first check that any potential internal vertex on an $(a,b)$-dipath in $H$ has the same backward reachability in $T$ as $a$.

    \begin{claim}\label{claim:lb:cases}
        Let $a,b \in T \sm \{z\}$ be such that $(a,b) \in E(G)$ and $w \in V(H) \sm \{a,b\}$ be a vertex lying on an
         $(a,b)$-dipath in $H$.
         Next, let $t \in T$
         satisfy $t \overset{H}\leadsto w$.
        Then  $t \overset{H}\leadsto a$.
    \end{claim}
    \begin{innerproof}
        By \Cref{claim:lb:adj-path}, the set $\reach_H(w,T)$ must be of the form $\reach_H(b,T) \cup N$ where $N$  is  a proper and non-empty subset of $\reach_H(a,T) \sm \reach_H(b,T) = \reach_G(a,T) \sm \reach_G(b,T)$.
        In particular, this is impossible when $\reach_G(a,T) = \reach_G(b,T) \cup \{a\}$.
        This immediately rules out the following cases: (1) $b = c_{v,y}$ (bottom-right corner), (2) $a = c_{v,x}$ (leaving bottom-left corner), (3) $a = c_{u,y}$ (leaving top-right corner), (4) $a = v_i$, $b = v_{i+1}$ (bottom row), and (5) $a = y_i$, $b = y_{i+1}$ (right column).

        We inspect the remaining cases and exploit the fact that $t \overset{H}\leadsto w$ implies $t \overset{H}\leadsto b$ and thus $t \overset{G}\leadsto b$ but the reachability set of $t$ must be strictly greater than the one for $b$.

        \begin{enumerate}
            \item Case $a = c_{u,x}$ (top-left corner).
            We consider $b = u_1$ since the subcase $b = x_1$ is analogous.
            The only terminals that reach $u_1$ are $c_{u,x}$ and $u_1$ so these are the only candidates for $t$.
            \Cref{claim:lb:adj-path} rules out the case $t=u_1$.
            This implies  $t=a$.

            \item Case $a = x_k$, $b = c_{v,x}$ (entering bottom-left corner).
            We have $N \subset \reach_G(x_k,T) \sm \reach_G(c_{v,x},T) = \{x_k,y_k\}$.
            Reaching both $c_{v,x}$ and one of these vertices is possible only
            when $t$ = $x_i$ for some $i \in [k]$ or $t = c_{u,x}$.
            Hence, $t \overset{H}\leadsto x_k$.

            \item Case $a = u_k$, $b = c_{u,y}$ (entering top-right corner) is analogous to the previous one.

            \item Case $a = u_i$, $b = u_{i+1}$ (top row).
            If $i > 1$ then $N \subset \{u_i,v_i\}$ and if $i = 1$ then $N \subset \{u_1,v_1,z\}$.
            Reaching both $u_{i+1}$ and one of these vertices is possible only
            when $t=u_j$ for $j \le i$ or $t = c_{u,x}$.
            In both cases we have $t \overset{H}\leadsto u_i = a$.

            \item Case $a = x_i$, $b = x_{i+1}$ (left column) is analogous to the previous one.

        \end{enumerate}
    \end{innerproof}

    \begin{claim}\label{claim:lb:edge}
        Let $a,b \in T \sm \{z\}$ be such that $(a,b) \in E(G)$.
        Then $(a,b) \in E(H)$.
    \end{claim}
    \begin{innerproof}
        Since  $a \overset{G}\leadsto b$ it must be
        $a \overset{H}\leadsto b$.
        Let $P$ be some $(a,b)$-dipath in $H$.
        Suppose for the sake of contradiction that the first vertex $w$ after $a$ on $P$ is not $b$.
        Note that $w \not\in T$ because otherwise \Cref{claim:lb:cases} applied to $t=w$ implies that  $w \overset{H}\leadsto a$ which contradicts \Cref{claim:lb:adj-path}.
        We claim that contracting the edge $(a,w)$ does not create new reachabilities among $T$, similarly as in \Cref{claim:lb:adj-reach}.
        Suppose that $t_2 \in T$ is reachable from $t_1 \in T$  by a dipath $Q$ in the contracted graph.
        If $Q$ is not present in $H$, then $Q$ uses the new vertex and so $t_1 \overset{H}\leadsto w$ and $a \overset{H}\leadsto t_2$.
        But then \Cref{claim:lb:cases} guarantees that $t_1 \overset{H}\leadsto a$ which in turn implies $t_1 \overset{H}\leadsto t_2$.
        As in \Cref{claim:lb:adj-reach}, this contradicts the minimality of $H$.
    \end{innerproof}

    Let $Z_x,Z_u,Z_y,Z_v$ be some dipaths in $H$ connecting respectively: $x_1 \to z$, $u_1\to z$, $z \to y_k$, $z \to v_k$.
    Let $Z_\cup$ be the set of internal vertices on  $Z_x,Z_u,Z_y,Z_v$, plus $z$.
    We will use $Z_\cup$ in a  construction of a hypothetical minor model and we will need to ensure that it is disjoint from other branch sets.

    \begin{claim}\label{claim:lb:disjoint-zt}
        It holds that $T \cap Z_\cup = \{z\}$.
    \end{claim}
    \begin{innerproof}
        We show that no vertex from $T$ is internal on $Z_x$ and the remaining cases are analogous.
        Any vertex on $Z_x$ is reachable from $x_1$ and reaches $z$.
        The only vertices from $T$ that reach $z$ are $z,x_1,u_1$ and $c_{u,x}$.
        The first two are the endpoints of $Z_x$.
        Neither $u_1$ nor $c_{u,x}$ are reachable from $x_1$.
        The claim follows.
    \end{innerproof}

    For each $i \in [k]$ we have  $x_i \overset{G}\leadsto y_i$ and $u_i \overset{G}\leadsto v_i$ hence the same holds in $H$.
    Let us fix  a $(x_i,y_i)$-dipath $P_i$ in $H$ and
    a $(u_i,v_i)$-dipath $Q_i$ in $H$.
    Again, we need to ensure several disjointness properties.

    \begin{claim}\label{claim:lb:disjoint-parallel}
        For each distinct $i,j \in [k]$ the dipaths $P_i,P_j$ are vertex-disjoint.
        Similarly, the dipaths $Q_i,Q_j$ are vertex-disjoint.
    \end{claim}
    \begin{innerproof}
        Assume w.l.o.g. $i < j$.
        Suppose that $P_i,P_j$ share a vertex.
        Then $x_j \overset{H}\leadsto y_i$.
        This violates the emulator specification since $y_i$ is not reachable from $x_j$ in $G$.
        An analogous argument works for $Q_i,Q_j$.
    \end{innerproof}

    \begin{claim}\label{claim:lb:disjoint-corners}
        For each $i \in [2,k-1]$ no $t \in T \cup Z_\cup$ appears as an internal vertex on  $P_i$ (resp. $Q_i$).
    \end{claim}
    \begin{innerproof}
        We prove the claim for $P_i$ w.l.o.g.
        When $t \in V(P_i)$ then
        $x_i \overset{H}\leadsto t$ and $t \overset{H}\leadsto y_i$.
        This rules out the top row $(t \ne u_j)$ and the bottom row $(t \ne v_j)$.
        If $t = x_j$ or $t = y_j$ then it must be $i=j$ so $t$ is not internal.
        If $t \in Z_\cup$ then  $z \overset{H}\leadsto t$ or $t \overset{H}\leadsto z$.
        This implies that
         $z \overset{H}\leadsto y_i$ or $x_i \overset{H}\leadsto z$.
         Consequently, $z \overset{G}\leadsto y_i$ or $x_i \overset{G}\leadsto z$.
         None of these connections occur for $i \in [2,k-1]$.
    \end{innerproof}

We are ready to prove the main argument that lower bounds the size of $H$.

    \begin{claim}\label{claim:lb:intersect}
        For each $i,j \in [2,k-1]$ the dipaths $P_i,Q_j$ share at least one vertex.
    \end{claim}
    \begin{innerproof}
        Assume for the sake of contradiction that $P_i,Q_j$ are vertex-disjoint.
        We will use this assumption to construct a  $K_5$-minor in $H$.
        We define
        \begin{itemize}
            \item $C_u = \{u_i \colon i \in [k]\} \cup \{c_{u,x}, c_{u,y}\} \cup V(Q_j) \sm \{v_j\}$
            \item $C_x = \{x_i \colon i \in [k]\} \cup V(P_i) \sm \{y_i\}$
            \item $C_v =  \{v_i \colon i \in [k]\} \cup \{c_{v,x}, c_{v,y}\}$
            \item $C_y =  \{y_i \colon i \in [k]\}$
        \end{itemize}

        We claim that $(C_u,C_x,C_v,C_y,Z_\cup)$ are branch sets of a minor model of $K_5$ in $H$.
        See \Cref{fig:lower-bound} for an illustration.

        First, we check that each of these sets induces a connected subgraph of $H$.
        The top-row set $\{u_i \colon i \in [k]\} \cup \{c_{u,x}, c_{u,y}\}$ is connected in $H$ due to \Cref{claim:lb:edge}.
        The same argument applies to bottom row, left column, and right column.
        Appending a path ($Q_j - v_j$ or $P_i - y_i$) to a set containing its endpoint preserves connectivity.
        Also, $H[Z_\cup]$ is connected as a union of four paths visiting $z$.

        Second, we check that every two branch sets are connected by an edge.
        The set $Z_\cup$ is adjacent to $x_1,u_1,v_k,y_k$ which belong to the four remaining branch sets.
        By \Cref{claim:lb:edge} the edge $(c_{u,x}, x_1)$ is present in $H$ which gives a connection between $C_u$ and $C_x$.
        The same argument applies to the remaining three pairs meeting at the corners.
        The adjacency for $(C_u,C_v)$ and $(C_x,C_y)$
        follows from the fact the path $Q_j$ (resp. $P_i$) intersects both $C_u,C_v$ (resp. both $C_x,C_y$) and does not use any vertices from outside these sets.

        Third, we check that the branch sets are disjoint.
        The set $Z_\cup$ is disjoint with any other due to \Cref{claim:lb:disjoint-zt} and \ref{claim:lb:disjoint-corners}.
        Next, $V(P_i)$ and $V(Q_j)$ are disjoint by the initial assumption.
        By \Cref{claim:lb:disjoint-corners}, no terminal appears as an internal vertex on $P_i$ or $Q_j$ so appending them to $C_u$ and $C_x$ does not create overlaps with any other branch sets.

        We have constructed a minor model of $K_5$ which certifies that $H$ cannot be planar.
        This contradicts the assumption that $P_i,Q_j$ are vertex-disjoint.
    \end{innerproof}

    For $i,j \in [2,k-1]$ let $h_{i,j} \in V(H)$ be a vertex in the intersection of $P_i$ and $Q_j$.
    \Cref{claim:lb:disjoint-parallel} implies that $h_{i,j} \ne h_{i',j'}$ for distinct pairs $(i,j)$ and $(i',j')$.
    It follows that $|V(H)| \ge (k-2)^2$.
\end{proof}

\section{Conclusion}
\label{sec:conclusion}

We have presented new algorithms for dynamic planar reachability, breaking the $\Oh(n^{1/2})$ and  $\Oh(n^{2/3})$ barriers in the offline and online settings, respectively.
The first result rules out the possibility of repeating
the hardness reduction known for dynamic planar shortest paths~\cite{Abboud16} and encourages one to seek further improvements in the running time.
We first discuss the perspective of improving the online algorithm.

\paragraph{Online reachability.}

The only source of randomization and amortization in \Cref{thm:planar-dyn-reach} comes from \cite{DBLP:conf/soda/Karczmarz18}
which, in turn, inherits these from \cite{DBLP:journals/siamcomp/Bernstein16}.
It is plausible that this component could be replaced with a deterministic one based on~\cite{Diks07} that could also yield worst-case bounds below $n^{2/3}$.
It is unclear however whether that data structure can be adapted to the specific scenario occurring in~\Cref{thm:planar-dyn-reach}.

When it comes to substantial running time improvements,
we see two independent directions.
Several constructions of planar \re{s} rely on the assumption that the terminals lie on $\Oh(1)$ faces~\cite{subramanian1993fully, Diks07, DBLP:conf/soda/Karczmarz18}.
However, it is currently not clear how to avoid the $\Ot(r)$ rebalancing bottleneck when maintaining such $r$-divisions under boundary set extensions.
A more efficient
approach to maintaining an $r$-division with few holes under $\Ot(\sqrt{r})$ edge insertions could help achieve a much better update bound.
It is worth noting that
when the maximum degree is $\Oh(1)$ then a
piece $P$ with~$k$ additional terminals picked in an adversarial way could be decomposed into $\Oh(k)$ pieces of total boundary size $\Ot\left(\sqrt{rk}\right)$ while avoiding the $\Ot(r)$ bottleneck if one stores a recursive decomposition of $P$ (see, e.g.,~\cite{DBLP:journals/talg/Charalampopoulos22}).
However, this approach alone seems not enough to obtain a faster data structure.

Recall from~\Cref{subsec:online} that
Thorup's emulators~\cite{thorup2004compact} can be used in place of Subramanian's emulators~\cite{subramanian1993fully}, removing the dependency on $r$-divisions with few holes and thus the aforementioned $\Ot(r)$ bottleneck.
However, to
bring the update/query time closer to $n^{1/2}$
using Thorup's emulators one
probably needs to dynamize their construction.

\paragraph{Subpolynomial time.}
It is tempting to try to extend the offline scheme described in \Cref{subsec:twolayers} to incorporate three or more layers of chunks.
Unfortunately, the \re $J$ computed for $(\widehat H,Y)$ in the second layer has lower quality than the first one: it has either size $|Y|^2$ or it is subquadratic but not faithful anymore.
In the first case, no beneficial trade-off is possible for the third layer.
In the second case, we do not know how to extract smaller emulators from $J$ in subquadratic time.

There is also a hard limit to this analysis due to the fact that the bound for computing $J$ contains the super-linear term $|\widehat H| \cdot |Y|^{1-\delta}$ where $\delta = \frac 1 3$.
Together with the bound for the time needed to compute $\widehat H$, this results in the terms $n^{1-\alpha} + n^{\alpha - \beta/3}$ in the amortized upper bound where $\alpha > \beta > 0$ govern the sizes of chunks in the first two layers.
It holds that $\max(1-\alpha, \,\alpha - \beta/3) \ge 2/5$.
Similarly, for every fixed $\delta < 1$ one can obtain a lower bound of~$n^{\Omega(1)}$ on the corresponding sum.
Hence the key to bypass this barrier would be to compute $J$ in almost-linear time in $|\widehat H|$.

\begin{question}
    Let $G$ be a planar digraph and $X \sub V(G)$.
    Is there an algorithm that, given a \re $H$ for $(G,X)$ and $Y \sub X$, runs in time $|H|^{1+o(1)}$ and outputs a \re for $(G,Y)$ of size $|Y|^{1+o(1)}$?
\end{question}

A positive resolution of this question would result in an offline algorithm with  $n^{o(1)}$ query/update time by following the scheme from~\cite{chen2020fast}.
As the case of dynamic planar approximate distances suggests~\cite{ChangKT22, chang2022arxiv, Filtser24}, this may be the key to decipher the online complexity of dynamic planar reachability as well.
It is intriguing that we know that a \re for $(G,Y)$ of size $\Ot(|Y|)$ exists and $H$ contains all the information needed to compute it but we do not know how to extract this information when we cannot access the planar structure of $G$.
In fact, it is unclear whether this problem can be solved even in polynomial time since we do not know how to reconstruct $G$ from~$H$ alone.
It may be the case that such an algorithm does not exist if we allow $H$ to be an arbitrary \re and a potential solution may require
restricting both emulators (input and output) to some well-behaving class, e.g., faithful.

\paragraph{Minor-free digraphs.}

In \Cref{thm:minor-free-dyn-reach} we have presented a reachability data structure for $K_h$-minor-free digraphs with
truly sublinear query/update time.
It is desirable to transfer our results for planar digraphs to $K_h$-minor-free digraphs but there are two concrete reasons that make it challenging.
First, in the online algorithm we exploit the fact that the computed $r$-division has a nice topological structure (``few holes'') and so we can use \Cref{thm:decremental-bnd-emulator} to compress the information about paths connecting terminals lying on $\Oh(1)$ faces.
Second, it is unclear whether the construction from \Cref{sec:faithful} can be generalized to $K_h$-minor-free digraphs
due to the need of finding balanced separators formed by $\Oh(1)$ dipaths (\Cref{lem:faithful:balanced}).
It has been claimed that separators with a~similar property exist in $K_h$-minor-free digraphs~\cite{Bhattacharyya12} but we are only aware of a construction that (1) concerns shortest-paths separators and (2) estimates the running time as an~unspecified polynomial~\cite{abraham2006object}.
Since this construction requires full access to the Robertson-Seymour decomposition of a graph, improving the running time to near-linear (or close-to-linear) seems challenging.

\paragraph{Online dynamic distances.}

Finally, recall that for the task of maintaining exact distances in a dynamic undirected planar graph the best query/update time remains $\Ot(n^{2/3})$~\cite{FR06,DBLP:journals/talg/KaplanMNS17} and the optimal $\Ot(\sqrt{n})$ bound has been only achieved in the offline setting~\cite{Das22}.
It has been pointed out in~\cite[\S 2]{Das22} that, in order to improve upon $2/3$ in the online setting, it suffices to maintain a so-called {\em dense distance graph} under edge weight increases in truly sublinear update time.
While the set systems studied in this work are inappropriate to address this challenge, there are other, graph-distance-related, set systems that also enjoy bounded VC dimension in planar graphs~\cite{li2019planar, LeW24, KarczmarzZ25}.
Our study of dynamic reachability
shows that the advantages of such set systems can be leveraged to shave off polynomial factors in dynamic data structures and suggests a way to tackle the dynamic distances as~well.
\appendix

\newcommand{\tr}{\mathcal{T}}

\section{Decremental reachability emulator}\label{app:decremental-emulator}
In this section, we outline how the data structure of~\cite{DBLP:conf/soda/Karczmarz18} implies~\Cref{thm:decremental-bnd-emulator}.\footnote{Throughout this paper, references to sections and numbered statements in~\cite{DBLP:conf/stoc/ItalianoKLS17,DBLP:conf/soda/Karczmarz18} use the numbering of the respective full versions linked in the bibliography.}

\paragraph{Decremental transitive closure.} Let us first review the decremental transitive closure data structure of~\cite{DBLP:conf/soda/Karczmarz18}.
By the reachability-preserving reduction of~\cite[Lemma~5.2]{DBLP:conf/soda/Karczmarz18}, based on~\cite[Lemma~6.1 and Corollary~6.2]{DBLP:conf/stoc/ItalianoKLS17}, we may assume that $G$ has a \emph{simple recursive decomposition}.
This decomposition is a tree $\tr(G)$ representing a hierarchy of $\Oh(n)$ subgraphs $P\subseteq G$ of total size $\Ot(n)$, called pieces.
The graph $G$ is the root of $\tr(G)$, and the leaf pieces have size $\Oh(1)$.
Each non-leaf piece $P$ has two children $P_1,P_2$ such that $P_1\cup P_2=P$.
The piece~$P$ comes with a boundary set $\bnd{P}$ of vertices that includes vertices shared with pieces that are neither ancestors nor descendants of $P$ in $\tr(G)$. For non-leaf pieces $P$, we have $\bnd{P}\subseteq \bnd{P_1}\cup \bnd{P_2}$ and $V(P_1)\cap V(P_2)\subseteq \bnd{P_1}\cap \bnd{P_2}$.
Moreover, $\bnd{P}$ lies on $\Oh(1)$ pairwise disjoint simple holes (recall the definition from~\Cref{s:rdiv}),
called \emph{separator curves} in~\cite{DBLP:conf/stoc/ItalianoKLS17,DBLP:conf/soda/Karczmarz18} (see \cite[Definition~4.1]{DBLP:conf/stoc/ItalianoKLS17}). \space{}
Finally, we have
$|\bnd{P}|=\Oh(\sqrt{n})$ for every piece $P$ and $\sum_{P\in\tr(G)}|\bnd{P}|^2=\Ot(n)$.

For any digraph $H$, let $H^*$ denote the transitive closure of $H$, i.e., a graph on $V(H)$ such that $(u,v)\in E(H^*)$ if and only if there is a  path $u\to v$ in $H$.
For each piece $P$, the data structure \emph{explicitly} maintains the induced subgraph $P^*[\bnd{P}]$ of the transitive closure of $P$.
For leaf pieces, these constant-size graphs are recomputed directly.
Observe that $P^*[\bnd{P}]$ can be equivalently viewed as a matrix of $|\bnd{P}|^2$ bits encoding the reachability between vertices $\bnd{P}$ in $P$.
When an edge $e$ is deleted from $G$, it is also deleted from all the pieces $P\in\tr(G)$ it resides in.
As a result, each piece $P$ is also a decremental graph, and so is the transitive closure $P^*[\bnd{P}]$.
Consequently, due to any sequence of edge deletions issued to $G$, each $P^*[\bnd{P}]$ can undergo at most $|\bnd{P}|^2$ updates, where a $1$-bit turns into a $0$-bit if $P^*[\bnd{P}]$ is viewed as a matrix.

The way the boundaries of the decomposition's pieces are defined implies that $P^*[\bnd{P}]$ can be computed inductively as $P^*[\bnd{P}]=(P_1^*[\bnd{P_1}]\cup P_2^*[\bnd{P_2}])^*[\bnd{P}]$.
The main technical idea behind the data structure of~\cite[Theorem~5.17]{DBLP:conf/soda/Karczmarz18} is that for each non-leaf piece $P$, $P^*[\bnd{P}]$ can be explicitly maintained subject to batches of updates to the ``children'' graphs $P_1^*[\bnd{P_1}]$ and $P_2^*[\bnd{P_2}]$ in $\Ot\left(|\bnd{P_1}|^2+|\bnd{P_2}|^2\right)$ total time using a Monte Carlo randomized algorithm correct with high probability.
As a result, the matrices $P^*[\bnd{P}]$ can be updated bottom-up after every edge deletion in $G$.
All the updates are thus processed in $\Ot\left(\sum_{P\in\tr(G)}|\bnd{P}|^2\right)=\Ot(n)$ total time.

\paragraph{Adjustments to~\cite{DBLP:conf/soda/Karczmarz18}.} Let us now describe the modifications needed  to derive~\Cref{thm:decremental-bnd-emulator} from the outlined data structure.
Recall that $G$ comes with pre-defined boundary $\bnd G$ of size $O(\sqrt{n})$ located on $\Oh(1)$ pairwise disjoint holes of $G$.
To prepare the graph $G$ for decomposition described in~\cite{DBLP:conf/stoc/ItalianoKLS17}, we need to ensure that $G$ is simple, connected, constant-degree, each non-hole face is a triangle, and the holes are disjoint simple faces of $G$.
To this end, we first make $G$ connected using auxiliary edges outside the holes.
For the pre-defined vertices $B\subseteq \bnd{G}$ of hole $h$, draw an auxiliary facial cycle of size $\Oh(|B|)$ inside $h$, meeting $G$ exactly in $B$ and using one corner at each vertex of $B$.
Afterwards, the vertices $B$ lie on a simple face of $G$; relabel $h$ to be that face.
Note that at this point the holes are simple and disjoint.
Next, apply the vertex-cycle expansion of~\cite[Section~8.1]{DBLP:conf/stoc/ItalianoKLS17} in the given embedding to make $G$ bounded-degree.
This step makes each $h$ grow by a constant factor; it is also easy to assign cycle representatives $B'$ of $B$ on the expanded hole $h$.
Then, triangulate the non-hole faces using the construction with new vertices from~\cite[Section~6]{DBLP:conf/soda/Karczmarz18}, which preserves the embedding, simplicity, and bounded degree.
Finally, apply the decomposition of~\cite[Section~8.2]{DBLP:conf/stoc/ItalianoKLS17} with the obtained holes (faces) as initial holes.
The resulting root boundary has size $\Oh(|\bnd G|)$ and contains representatives of $\bnd G$ on $\Oh(1)$ pairwise disjoint simple faces; the same subtree bound $\Oh(n\log n+|\bnd G|^2)=\Ot(n)$ preserves the required quantitative properties.
After deleting the auxiliary edges other than the edges inside the strongly connected vertex gadgets of both expansion steps, restricting the maintained root matrix to these representatives and relabeling them explicitly maintains $G^*[\bnd G]$. \space{}
The total update time of the data structure remains $\Ot(n)$.

We now explain how the matrix $G^*[\bnd{G}]$ can be efficiently converted into an $\Ot(\sqrt{n})$-sized reachability emulator for $(G,\bnd{G})$.
We rely on the following reachability matrix decomposition from~\cite{DBLP:conf/stoc/ItalianoKLS17}, also used in~\cite{DBLP:conf/soda/Karczmarz18}.
Roughly speaking, the decomposition is based on the simple observation that if four vertices $a,b,c,d$ lie on a single hole of $G$ in that (cyclic) order, then any two $a\to c$ and $b\to d$ paths have to cross.
This can be used to prove that the reachability between any non-interleaving sets $A,B$ on a single hole has a concise representation in the form of $|A|$ intervals over a subset of $B$ reachable from $A$, i.e., has complexity $\Oh(|A|+|B|)$ as opposed to $\Oh(|A|\cdot |B|)$.
For example, if $\bnd{G}$ lies on a single hole of $G$, then the set of all pairs $(u,v)\in \bnd{G}\times\bnd{G}$ can be recursively partitioned into $\Oh(|\bnd{G}|)$ subsets $A\times B$ of the aforementioned form so that the total sum of $|A|+|B|$ is $O(|\bnd{G}|\log|\bnd{G}|)$.
Indeed, if one splits $\bnd{G}$ into non-interleaving sets $X,Y$ of roughly equal size, then the sought decomposition can be formed to include $X\times Y$, $Y\times X$, and the recursive decompositions of $X\times X$ and $Y\times Y$.
Reachability between pairs from distinct holes can be handled similarly, see~\cite{DBLP:conf/stoc/ItalianoKLS17} for a proof.
\newcommand{\aparts}{\mathcal{A}}
\newcommand{\act}{\operatorname{act}}

\begin{lemma}[{\cite[Lemma~4.7]{DBLP:conf/stoc/ItalianoKLS17}}]\label{lem:reach-matrix-parts}
    Let $U\subseteq V(G)$ lie on $\Oh(1)$ pairwise disjoint simple holes of a plane digraph $G$.
    Let $A\in \{0,1\}^{U\times U}$ be the reachability matrix such that for $s,t\in U$, $A_{s,t}=1$ if and only if there is a path $s\to t$ in $G$.
    For any $X,Y\subseteq U$, let $A^{X,Y}$ denote the submatrix of $A$ with rows $X$ and columns $Y$.
    Moreover, let $\act(A^{X,Y})$ denote the subset of columns $Y'\subseteq Y$ that contain at least one $1$ in $A^{X,Y}$.

    There exists an ordering $\prec$ of $U$ and a partition of $A$ into a set $\aparts=\{A^{S_1,T_1},\ldots,A^{S_k,T_k}\}$ of disjoint submatrices of $A$ satisfying the following conditions.
    \begin{enumerate}[(1)]
    \item For each $s,t\in U$, there exists exactly one $A^{S_i,T_i}\in\aparts$ such that $s\in S_i$ and $t\in T_i$.
    \item For each $A^{S_i,T_i}\in\aparts$ with $T'=\act(A^{S_i,T_i})$, the ones in each row of $A^{S_i,T'}$ form $\Oh(1)$ intervals of the set $T'$ ordered by $\prec$.
    \item For any $s\in U$ ($t\in U$, respectively) there are $\Oh(\log{|U|})$ matrices $A^{S_i,T_i}\in\aparts$ such that $s\in S_i$ ($t\in T_i$, resp.).
\end{enumerate}
The diagonal entries form singleton submatrices.
The sets $S_i,T_i$ that define the partition $\aparts$ do not depend on the entries of $A$, but only on the cyclic order of $U$ on the $\Oh(1)$ holes the set $U$ lies on. The sets $S_i,T_i$ and the order $\prec$ can be computed in $\Oh(|U|^2)$ time {given the vertices $U$ partitioned into the $O(1)$ curves along with the cyclic orders of these vertices on the curves.}\space{}
\end{lemma}

We additionally maintain the matrix $A=G^*[\bnd{G}]$ partitioned by applying~\Cref{lem:reach-matrix-parts} to the representatives of $\bnd G$ on the expanded root boundary.
That is, when the data structure is initialized, we compute the partition $\aparts=\{A^{S_1,T_1},\ldots,A^{S_k,T_k}\}$.
Note that $\sum_i^k (|S_i|+|T_i|)=\Oh(|\bnd{G}|\log{n})=\Ot(\sqrt{n})$ and $\sum_i^k |S_i|\cdot |T_i|=\Oh(|\bnd{G}|^2)=\Oh(n)$ by items (3) and (1) of~\Cref{lem:reach-matrix-parts}, respectively.

Each $A^{S_i,T_i}$, being a submatrix of $A$, also undergoes monotonic updates when edges are deleted from~$G$.
The total number of such updates is clearly at most $|S_i|\cdot |T_i|$.
Subject to such updates, we maintain the set $T_i'=\act(A^{S_i,T_i})$, as defined in~\Cref{lem:reach-matrix-parts}, by storing the count of ones in each column of $A^{S_i,T_i}$.
Moreover, for each $s\in S_i$, we maintain the set $\{t:A^{S_i,T_i}_{s,t}=1\}$ of ones in the row $s$ of $A^{S_i,T_i}$ in the form of a minimum-size collection $J_s$ of $\Oh(1)$ intervals over $T_i'$.
Recall that this is always possible by~\Cref{lem:reach-matrix-parts}.

Whenever an edge of $G$ is deleted, (1) the set $T_i'$ might shrink due to some column becoming all-zeroes and (2) some ones in the row $s$ of $A^{S_i,T_i}$ might turn to zeroes.
We need to update $J_s$ to reflect these changes.

To this end, we process column removals from $T_i'$ one by one, and subsequently entry changes in row $s$ all at once.
First consider a change to $T_i'$. Suppose some $z$ is removed from $T_i'$ as a result of the update, i.e., the column $z$ becomes all-zeroes.
For each interval $[x,y]\in J_s$, if $x=y=z$, then such an interval is removed from $J_s$.
Otherwise, if $x<y$, then if $x=z$, the endpoint $x$ is replaced with the successor of $x$ in $T_i'$. Similarly, if $y=z$, then the endpoint $y$ is replaced with its predecessor in $T_i'$.
Note that the above operations can be implemented in $\Ot(1)$ time if $T_i'$ is additionally maintained sorted, e.g., in a balanced binary tree. Thus, a single removal from $T_i'$ can be processed in $\Ot(|J_s|)=\Ot(1)$ time.

When changes to $T_i'$ are processed, each interval in $J_s$ already has endpoints in the updated set~$T_i'$.
We merge any intervals that have become adjacent in this updated order, in $\Ot(1)$ time, so that the collection is again minimum-size.
The next step is to process entry changes in row $s$.
Let $X$ be the subset of such $x\in T_i'$ that $A_{s,x}=1$ before the update and $A_{s,x}=0$ after the update.
We need to replace $J_s$ with a minimal collection of intervals over $T_i'$ whose union equals $\left(\bigcup J_s\right)\setminus X\subseteq T_i'$.
Note that any resulting interval $I$ will be a subinterval of some $[a,b]$ from the current $J_s$.
To compute the part of the sought collection that is a subset of $[a,b]\in J_s$, sort the elements $\{x_1,\ldots,x_l\}=X\cap [a,b]$ in increasing order in $\Ot(|X|)$ time.
For each $x_i$, compute its predecessor $x_i^-$ and successor $x_i^+$ in $T_i'$.
Finally, add the non-empty intervals out of $[a,x_1^-],[x_1^+,x_2^-],\ldots,[x_{l-1}^+,x_l^-],[x_l^+,b]$ to the constructed collection.
This way, the new collection $J_s$ can be found in $\Ot(|X|)$ time.
Recall from~\Cref{lem:reach-matrix-parts} that the resulting minimal collection has size $\Oh(1)$.

Throughout, for each row $s\in S_i$, there are at most $|T_i|$ removals from $T_i'$, and at most $|T_i|$ entry changes in row $s$ of $A^{S_i,T_i}$, so maintaining the interval representations of $\{t:A_{s,t}^{S_i,T_i}\}$ takes $\Ot(|T_i|)$ total time. Over all $s\in S_i$, this amounts to $\Ot(|S_i|\cdot|T_i|)$ total time over all updates to $A^{S_i,T_i}$.

When the data structure updates an entry $(s,t)$ of the matrix $G^*[\bnd{G}]$ (i.e., the entry is changed from~$1$ to $0$) the update is passed to the unique $A^{S_i,T_i}\in\aparts$ with $(s,t)\in S_i\times T_i$.
Consequently, the required interval representations of rows of all $A^{S_i,T_i}\in\aparts$ are maintained in $\Ot\left(\sum_i^k|S_i|\cdot|T_i|\right)=\Ot(n)$ time.

Finally, the following lemma is used to efficiently convert the reachability matrix on $\bnd{G}$ into a reachability emulator.
\begin{lemma}\label{lem:matrix-to-emulator}
Given the partition $\aparts=\{A^{S_1,T_1},\ldots,A^{S_k,T_k}\}$ of the reachability matrix $A=G^*[\bnd{G}]$ from \Cref{lem:reach-matrix-parts} along with the respective sets $T_i'=\act(A^{S_i,T_i})$ and interval representations of ones in all rows of $A^{S_i,T_i'}$, one can construct a reachability emulator for $(G,\bnd{G})$ of size $\Ot(\sqrt{n})$ in $\Ot(\sqrt{n})$ time.
\end{lemma}
\begin{proof}
    For each $A^{S_i,T_i}$, we will construct a reachability emulator $H_i$ for $(G,S_i,T_i)$ of size $\Ot(|S_i|+|T_i|)$ in $\Ot(|S_i|+|T_i|)$ time.
    Then the result will follow by applying~\Cref{l:assemble-from-parts} to $H_1,\ldots,H_k$.
    The resulting emulator size and running time is $\Ot\left(\sum_{i=1}^k|H_i|\right)=\Ot\left(\sum_i^k |S_i|+|T_i|\right)=\Ot(\sqrt{n})$, as desired.

    The emulator $H_i$ is obtained using the segment tree trick, similarly as we did in~\Cref{lem:vsbased:general}.
    We construct a segment tree $F$ from~\Cref{lem:segment-tree} with leaves $T_i'$ (ordered by $\prec$) and edges of the tree directed from parents to children.
    The vertex set of the emulator is $S_i\cup T_i\cup V(F)$ and the tree $F$ is included in $H_i$.
    For each row $s\in S_i$, we compute the decomposition of the $\Oh(1)$ intervals in the interval representation of the ones in row $s$ of $A^{S_i,T_i'}$ into $\ell=\Oh(\log{n})$ intervals $[a_1,b_1],\ldots,[a_\ell,b_\ell]$ represented in the segment tree~$F$.
    Finally, we add directed edges from $s$ to vertices $F_{[a_j,b_j]}$ for $j=1,\ldots,\ell$ which ensures that $s$ can reach a leaf $t\in T_i'$ if and only if $A^{S_i,T_i}_{s,t}=1$.
    Clearly, the total number of edges in the emulator is $\Oh(|T_i'|+|S_i|\log{n})=\Ot(|S_i|+|T_i|)$.
    One can argue that $H_i$ is indeed a reachability emulator for $(G,S_i,T_i)$ analogously as was done in the proof of~\Cref{lem:vsbased:general}.
\end{proof}

After each update to $G$ and updating the data structures associated with the partition $\aparts$, we recompute the reachability emulator for $(G,\bnd{G})$ using~\Cref{lem:matrix-to-emulator}.
Since this costs $\Ot(\sqrt{n})$ worst-case time per deletion,~\Cref{thm:decremental-bnd-emulator} follows.

\section{Constructing an $r$-division with few holes}\label{app:rdiv-few-holes}

Let us begin by discussing the discrepancies between the usual definition of an $r$-division of few holes \cite{KleinMS13} and the one that we use.
In~\cite{KleinMS13}, an embedding of $G$ is assumed, and the pieces of an $r$-division with few holes $\rdiv$ are assumed to share that embedding.
The pieces of $\rdiv$ are supposed to be connected and are not necessarily pairwise disjoint.
The holes of a piece $P$ are defined to be the faces of~$P$ that are not originally faces of $G$.
For disambiguation, let us call these the \coldef{face-holes} of $P$.
The (bounding cycles of) face-holes of~$P$ contain the boundary vertices of $P$ exclusively, and each vertex of $\bnd{P}$ lies on (the bounding cycle of) some face-hole of $P$.

With such a definition of an \coldef{$r$-division with face-holes},~\cite{KleinMS13} show:

\begin{theorem}[\cite{KleinMS13}]\label{t:rdiv-few-holes}
  Let $G$ be a simple triangulated connected plane graph with $n$ vertices.
  For any $r\in [1,n]$, an $r$-division with few face-holes of $G$
  can be computed in $\Oh(n)$ time.
\end{theorem}
It is worth noting that the bounding cycles of the pieces' face-holes in the decomposition of~\Cref{t:rdiv-few-holes} are not necessarily simple or vertex-disjoint in general.

We now explain how to convert the decomposition produced by~\Cref{t:rdiv-few-holes} so that it complies with our requirements posed on an $r$-division with few holes $\rdiv$ (see~\Cref{s:rdiv}).
\Cref{lem:rdiv-few-holes-weak} will follow.

Let $G$ be an arbitrary planar digraph.
We first embed $G$, connect its components and triangulate it using auxiliary edges, thus obtaining $G^+\supseteq G$.
Then, we apply~\Cref{t:rdiv-few-holes} to $G^+$, obtaining an $r$-division with few face-holes $\rdiv^+$.
Recall that the pieces of $\rdiv^+$ might not be necessarily edge-disjoint.

With $\rdiv^+$ in hand, we construct the pieces $P\in\rdiv$ from the corresponding pieces $P^+\in\rdiv^+$ by first setting $P=P^+\cap G$ and then assigning edges appearing in multiple obtained pieces to one arbitrary of them to guarantee edge-disjointness.
We put $\bnd{P}:=\bnd{P^+}$.
Note that the bounding cycle of any face-hole of $P^+$ forms a hole of any subgraph $P$ of $P^+$ according to our definition of a hole.
The bounding cycles of face-holes of~$P^+$ are not necessarily pairwise disjoint, though.
To guarantee the pairwise disjointness of holes of $P$, we simply assign every boundary vertex $v\in \bnd{P}$ to an arbitrarily chosen hole such that $v$ lies on the bounding cycle of the corresponding face-hole in $P^+$.
Finally, we adjust the other holes that originally passed through $v$ slightly to bypass $v$ without touching other holes.
By construction, $\rdiv$ is an $r$-division with few holes according to the definition in~\Cref{s:rdiv}.

\bibliographystyle{alphaurl}
\bibliography{bib}

\end{document}